\documentclass [prl,amsmath,showpacs,twocolumn,preprintnumbers,superscriptaddress,nofootinbib]{revtex4-1}
\usepackage{amsmath,amssymb,mathtools}
\usepackage{enumerate}
\usepackage{hyperref}
\usepackage{amsthm} 
\theoremstyle{plain}
\newtheorem{lemma}{Lemma}[section]
\theoremstyle{definition}    

\theoremstyle{remark}        
\usepackage{amsmath}
\usepackage[svgnames]{xcolor}
\usepackage{subcaption}
\usepackage{amsmath,amssymb,mathtools}
\usepackage{hyperref}     
\usepackage{physics}
\newcommand{\Klein}{\mathbb Z_{2}\!\times\!\mathbb Z_{2}}
\newcommand{\Cl}{\mathrm{Cl}_{1,3}}
\usepackage{enumitem}
\usepackage{amsfonts}
\newcommand{\suchthat}{\;\ifnum\currentgrouptype=16 \middle\fi|\;}
\begin{document}

\title{Graded Paraparticle Algebra of Majorana Fields for Multidimensional Quantum Computing with Structured Light}
\author{Fabrizio Tamburini} 
\email{fabrizio.tamburini@gmail.com}
\author{Nicol\'o Leone}
\email{nicolo.leone@rotonium.com}
\author{Matteo Sanna}
\email{matteo.sannai@rotonium.com}
\author{Roberto Siagri}
\email{roberto.siagri@rotonium.com}
\affiliation{Rotonium -- Quantum Computing, Le Village by CA, Piazza G. Zanellato, 23, 35131 Padova PD, Italy. }

\begin{abstract}
We present a theoretical framework that integrates Majorana's infinite-component relativistic equation within the algebraic structure of paraparticles through the minimal nontrivial $\mathbb{Z}_2 \times \mathbb{Z}_2$--graded Lie algebras and $R$-matrix quantization. 
By mapping spin-dependent mass spectra to graded sectors associated with generalized quantum statistics, we derive an equation embodying Majorana's mass-spin relation describing Majorana quasiparticles of structured light carrying spin and orbital angular momentum. 
These quanta in the $\mathbb{Z}_2 \times \mathbb{Z}_2$--graded algebras and $R$-matrix formulations extend the previous results from superconducting qubits to photonic platforms and set up deterministic 2-photon gates involving at least two qubits encoded in a single photon without nonlinear effects. 
This makes feasible general quantum computing pathways exploiting fractional statistics through Nelson's quantum mechanics and implement a novel procedure for error correction in photonic platforms. Furthermore, this approach makes possible to set paraparticle-based quantum information processing, beyond fermions and bosons, using graded qudits.
\end{abstract}

\maketitle

\section{Introduction}
Paraparticles are hypothetical quantum particles that extend the standard classification of particles in quantum field theory beyond the well-known fermions and bosons \cite{green}. 
Yet, after six decades of theoretical work, no experimental platform has demonstrated deterministic and scalable paraparticle statistics in the laboratory.
These entities emerge from parastatistics, a generalization of quantum statistics that allows for more complex symmetry properties under particle exchange \cite{Ohnuki,Kamefuchi}. 
While fermions like electrons obey the Pauli exclusion principle and antisymmetric wavefunctions, and bosons, like photons, follow symmetric wavefunctions and can occupy the same quantum state, paraparticles are theorized to obey intermediate symmetry rules characterized by a parameter, $p$,  known as the parastatistics order.

For instance, a system of parafermions of order $p$ allows at most $p$ particles to occupy the same antisymmetric state, unlike ordinary fermions where only one is allowed. This framework arises naturally in mathematical formulations involving trilinear commutation relations rather than the standard bilinear ones, and it retains consistency with the spin-statistics theorem under certain conditions.
Paraparticles have surfaced in supersymmetry, quantum groups and higher-dimensional field theories, exploring extensions of the Standard Model, but remain still experimentally elusive. They also sharpen our understanding of symmetry by extending the Young-tableaux machinery \cite{fulton} used for bosons and fermions to the full permutation group. 
These diagrammatic tools are used in the representation theory of symmetric and general linear groups used to describe the symmetry properties of bosonic systems or identical particles obeying parastatistics.
In quantum computing and quantum information (see, for an introductory reference, Ref. \cite{NielsenChuang2000}), Wang and Hazzard~\cite{wang} demonstrated that paraparticle statistics can be digitally simulated using superconducting qubits. 

Here we propose, instead, first the conceptual design of a direct photonic realization of paraparticle algebras using structured light modes, providing a continuous-variable and scalable platform to simulate also in this case exotic statistics beyond fermions and bosons, including Majorana photonic quasiparticles through the graded Lie superalgebras for paraparticles as proposed by Toppan~\cite{toppan} then we focus more on a mathematical formalism for more general photonic applications. 
As an example, we describe what are the actual properties of an ideal platform with single--photon computation in a spin-orbit coupled waveguide that natively realizes the full $\mathbb Z_{2}\!\times\!\mathbb Z_{2}$-graded paraparticle algebra that can be extended to higher dimensional qudits.
After setting the necessary algebraic background we derive a Jordan--Wigner map from graded paraboson/parafermion operators then applied to spin--orbit photonic modes and, as an example, build a deterministic universal gate set inside the photonic ququart and quantify its tolerance to realistic loss and dephasing. 

Finally, we sketch some ideal gate implementation referring to current integrated-photonics technology and outline how the large-$p$ limit recovers Maxwell--Boltzmann behavior, establishing a bridge between paraparticle algebra, scalable photonic hardware, and the spectral geometry.

\section{Statistics of Parabosons and Parafermions}
Parabosons and parafermions that obey parastatistics of order $p$ are associated with trilinear commutation or anticommutation relations and belong to representations of the symmetric group rather than the permutation group used for bosons and fermions.
Adopting the minimal nontrivial grading parabosons of order $p$ are generalizations of bosons and satisfy the following trilinear commutation relations for creation, different from standard bosonic commutators \cite{Govorkov}, with $[a_i,a_j^\dag] = \delta_{ij}$,
\begin{equation} 
\begin{split}
&[a_i,[a_j^\dag,a_k]] = 2\delta_{ij}a_k,  \quad [a_i,[a_j,a_k]] =0 ,
\\
&[a_i,[a_j^\dag,a_k^\dag]] = 2\delta_{ij}a_k^\dag - 2\delta_{ik}a_j^\dag
\end{split}
\label{1}
\end{equation}
and up to $p$ identical parabosons can occupy the same quantum state.
The states so constructed obey symmetrized Young tableaux, with at most $p$ rows. 
When $p=1$, the statistics reduce to ordinary bosons.

Parafermions, instead, obey the trilinear anticommutation relations:
\begin{equation} 
\begin{split}
&\{f_i,\{f_j^\dag,f_k\}\} = 2\delta_{ij}f_k,  \quad \{f_i,\{f_j,f_k\}\} =0.
\\
&\{f_i,\{f_j^\dag,f_k^\dag\}\} = 2\delta_{ij}f_k + 2\delta_{ik}f_j .
\end{split}
\label{2}
\end{equation}
In both cases, the parameter $p$ determines the degree of deviation from conventional statistics, with higher $p$ values allowing more symmetry flexibility. Parastatistics find formal use in quantum field theory and group theory, but so far remain theoretical constructs without experimental evidence. For Lorentz covariant fields also according to Greenberg--Messiah construction \cite{GreenbergMessiah}, the commutator (paraboson) version of Eq.~\ref{1} must be used for integer spin, and the anticommutator (parafermion) version for half-integer spin (Eq.~\ref{2}) otherwise micro-causality does not hold.

\subsection{Large $p$ values}

As the Green index $p \rightarrow \infty$, both parabosons and parafermions will change their characteristics. A paraboson system loses every trace of its generalised-Bose exchange constraint and behaves, for all practical purposes, like a collection of distinguishable particles recovering the full Bose-Einstein condensation behavior and the exclusion constraint disappears making the occupation statistics, parastatistics, converge to Maxwell--Boltzmann instead of bosonic symmetrization; in the classical limit the quantum correlations and exchange effects that distinguish fermions and bosons are gradually erased. For parafermions the operator algebra really collapses to canonical Fermi anticommutator $\{f,f^\dagger\}=1$ while the occupancy statistics become Maxwell–Boltzmann.

Green trilinear relations for a single paraboson mode of order $p$ of Eq. \ref{1} and \ref{2} with $p$-deformed commutator reproduces Green's trilinear relations implies the operator identity from the Green trilinear algebra
\begin{equation}
[b,b^{\dagger}]
   \;=\;
   1+\frac{2}{p}\,b^{\dagger}b
     \;-\;\frac{2}{p}\,bb^{\dagger}.
\label{eq:pComm}
\end{equation}
here explicitly written with commutators and anticommutators together to describe a paraboson mode $b$, $b^{\dagger}$ of order $p$ through the trilinear Green identities,
\begin{subequations}
\label{eq:GreenTri}
\begin{align}
&\bigl[b,\{\,b^{\dagger},b\,\}\bigr] = +\,2\,b,
\label{eq:GreenTri-a}\\
&\bigl[b^{\dagger},\{\,b^{\dagger},b\,\}\bigr] = -\,2\,b^{\dagger}.
\label{eq:GreenTri-b}
\end{align}
\end{subequations}
The number operator and $p$-deformed commutator are defined with the occupation operator
$\hat N:=b^{\dagger}b$ with $\{\,b^{\dagger},b\,\}=2\hat N+1$.
Using \ref{eq:GreenTri} then one derives $[\,\hat N,b^{\dagger}\,]=+\,b^{\dagger}$, $[\hat N,b\,]=-\,b$
and the $p$-deformed single--mode commutator can be written also 
\footnote{A concrete Fock representation shows that each application of $b^{\dagger}$ effectively increases $\hat N$ by one and adds a factor $2/p$; see, e.g., O.~W.~Greenberg, \emph{Phys.\ Rev.\ Lett.}\ \textbf{13},~598 (1964).} making explicit the index number $p$ in terms of the occupation operator $\hat N$
\begin{equation}
[\,b,b^{\dagger}\,]_{p}
     \;=\;
     1+\frac{2}{p}\,\bigl(2\hat N+1\bigr) .
\label{eq:pComm}
\end{equation}
Equation~\ref{eq:pComm} reduces to the canonical Bose commutator $\,[b,b^{\dagger}]=1\,$ when $p\to\infty$, while the full Green relations \ref{eq:GreenTri} remain intact for any finite $p$.

This is more evident with the $p$-deformed factorial and Fock ladder, which is given by introducing
\begin{equation}
[n]!_{p}
  :=\prod_{k=1}^{n}\Bigl(k+\frac{k(k-1)}{p}\Bigr),
\end{equation}
and the normalised states
\begin{equation}
\lvert n\rangle_{p}
   :=\frac{(b^{\dagger})^{n}}{\sqrt{[n]!_{p}}}\,\lvert0\rangle ,
\qquad
b\,\lvert0\rangle=0 .
\end{equation}
Ladder matrix elements are
\begin{subequations}
\begin{align}
&b^{\dagger}\lvert n\rangle_{p}
=\sqrt{\,n+1+\dfrac{n(n+1)}{p}}\;\lvert n+1\rangle_{p},\\
&b\lvert n\rangle_{p} =\sqrt{\,n+\dfrac{n(n-1)}{p}}\;\lvert n-1\rangle_{p}.
\end{align}
\end{subequations}
For fixed $n$, then $[n]!_{p}=n!\,\bigl[1+\mathcal O(1/p)\bigr]$,  one obtains
\begin{equation}
\lvert n\rangle_{p}
   \xrightarrow[p\to\infty]{}\;
   \lvert n\rangle_{\text{boson}}\;=\;
   \frac{(b^{\dagger})^{n}}{\sqrt{n!}}\lvert0\rangle .
\end{equation}
The probability weight for $k$ identical parabosons occupying one mode
contains the factor
\begin{equation}
\prod_{m=1}^{k-1}\Bigl(1-\frac{m}{p}\Bigr)
      \;=\;1-\mathcal O\!\bigl(k^{2}/p\bigr).
\end{equation}
Thus, for $p\gg k$,
the exclusion constraint disappears and the statistics converge to Maxwell--Boltzmann (MB) instead of Bose--Einstein (BE). 

In other words, standard bosons and fermions ($p=1$) are therefore special, low--order points of the broader Green hierarchy, as $p$ grows, exchange correlations fade and the system behaves as if the particles were distinguishable. 
Higher $p$ just means one is summing over $p$ identical local fields.  Green order $p$ is an internal degeneracy, not a spin label.
The same limit for parafermions removes the Pauli cap and also yields Maxwell--Boltzmann populations.
Similarly, parafermions tend to Bosonic statistics. 
For the sake of completeness, in this case, Green relations for single mode order $p$ can be also written in a mixed commutator/anticommutator notation, $\bigl[f,\{f^{\dagger},f\}\bigr]_{+}=2f$, and $\bigl[f^{\dagger},\{f^{\dagger},f\}\bigr]_{+}=-2f^{\dagger}$ where ``$+$'' and ``-'' indicate integer spin and half integer spin, respectively.
The $p$-deformed anticommutator with $\hat N=f^{\dagger}f$ is $\{f,f^{\dagger}\}_{(p)} = 1-\frac{2}{p}\,\bigl(2\hat N-1\bigr)$ that in the large-$p$ limit, $\lim_{p\to\infty}\{f,f^{\dagger}\}_{(p)} = 1$, $n_{\max}=p\;\longrightarrow\;\infty$ so the exclusion cap disappears and the statistics approach Maxwell--Boltzmann, mirroring the paraboson $\,[b,b^{\dagger}]_{(p)}$ $\to\,1$ result.
Summarizing,
\begin{equation}
\begin{aligned}
&\text{paraboson}_{\,p\to\infty}
\;\longrightarrow\;
\begin{cases}
\text{Bose algebra}~\text{(operator level)},\\[4pt]
\\
\text{MB statistics} ~\text{(occupation level)},
\end{cases}
\\[10pt]
&\text{parafermion}_{\,p\to\infty}
\;\longrightarrow\;
\begin{cases}
\text{Fermi algebra}~\text{(operator level)},\\[4pt]
\\
\text{MB statistics}~\text{(occupation level)}.
\end{cases}
\end{aligned}
\end{equation}

Physically, large $p$ systems can be used to model systems with extremely high internal symmetry or large hidden degrees of freedom, where particles behave as if they were distinguishable or non-interacting. These set of particles can be used to build quantum computational states. 
An example are Fock states of photons described in terms of paraparticles with large $p$. For parabosons the Fock space becomes indistinguishable from bosonic Fock space. For parafermions the occupation limit vanishes, and Fock states resemble classical modes.
This asymptotic behavior hints that standard quantum statistics (bosons and fermions) are merely specific, low-order cases of a broader algebraic structure.

Consider, as an example, the one--mode photon Fock ladder for large--$p$ parabosons to describe photonic qudits with OAM and polarization.
Let $b^{\dagger},b$ obey the Green paraboson algebra of order $p$, then $\bigl[b,\{b^{\dagger},b\}\bigr]_{-}=2\,b$, $\bigl[b^{\dagger},\{b^{\dagger},b\}\bigr]_{-}=-2\,b^{\dagger}$, with the vacuum $b\,\lvert 0\rangle=0$.  Define the $p$--deformed occupation ladder as $\lvert n\rangle_{p}\;=\;\frac{1}{\sqrt{[n]!_{p}}}\;(b^{\dagger})^{\,n}\lvert 0\rangle$, and $[n]!_{p}:=\prod_{k=1}^{n}\bigl(k+\tfrac{k(k-1)}{p}\bigr)$.
The deformation enters only through the ``parabosonic factorial'' $[n]!_{p}$. Matrix elements read
$b^{\dagger}\lvert n\rangle_{p} = \sqrt{\,n+1+n( n-1)/p}\;\lvert n+1\rangle_{p}$ and $b\lvert n\rangle_{p} = \sqrt{\,n+ n( n-1)/p}\;\lvert n-1\rangle_{p}$.
For large--$p$ limit and for fixed photon number $n$, $[n]!_{p}\xrightarrow[p\to\infty]{}n!$ and the ladder coefficients reduce to $\sqrt{n+1}$ and $\sqrt{n}$, so $\lvert n\rangle_{p} \xrightarrow[p\to\infty]{}\; \lvert n\rangle_{\mathrm{boson}}\,$. All $p$--dependent corrections are suppressed by $1/p$ and the algebra
recovers the canonical commutator $[b,b^{\dagger}]\!=\!1$ and the Fock space is therefore indistinguishable from the usual photonic one. In the two--mode extension, for modes $b_{1},b_{2}$, the deformed factorial becomes $[n_1]!_{p}\,[n_2]!_{p}$ and the basis states $
\lvert n_1,n_2\rangle_{p} = (b_1^{\dagger})^{n_1}(b_2^{\dagger})^{n_2}\lvert0\rangle/\sqrt{[n_1]!_{p}[n_2]!_{p}}$ again converge to $\lvert n_1,n_2\rangle_{\mathrm{boson}}$ as $p\!\to\!\infty$. A photonic qudit encoded in OAM or polarization can therefore be viewed as a large--$p$ parabosonic register where the $p$--corrections quantify how fast true Bose--statistics are approximated.

For parafermions the analogous ladder has an occupation cap $N\le p$, but as $p\to\infty$ the cap recedes and $\lvert n\rangle_{p}^{(\text{paraF})}\!\to\!\lvert n\rangle_{\text{MB}}$, i.e.\ Maxwell--Boltzmann populations with \emph{no} exchange symmetry.
Thus the bosonic paraboson tower interpolates to Bose condensation, whereas the parafermionic Majorana tower interpolates to a classical, fully distinguishable gas of particles.
This explicit Fock ladder, the ordered sequence of Fock number states,  shows that photon states can be embedded in the broader paraboson algebra, with ordinary optics recovered when the order $p$ is large, supporting the interpretation of bosons and fermions as low--order corners of a unified statistical hierarchy.
In models of high-energy physics or quantum gravity, where exotic symmetries or hidden sectors dominate, the $p \rightarrow \infty$ regime may serve as an effective approximation.
In this way is also provided a formal connection between quantum and classical statistics, suggesting how indistinguishability and statistical exclusion emerge from deeper algebraic origins.
Thus, paraparticles of high order tend to lose their distinct quantum statistical identities, merging toward a classical, boson-like limit that dissolves the constraints of both Bose and Fermi exclusion principles.

\section{Photonic OAM Majorana Quasiparticles and Paraparticle Statistics}
As we all know, boson and fermions follow different statistics. 
From Wess and Zumino to String theory \cite{zumino,string1,string2} and other formulations of high energy physics including several theories of everything have different approaches to unify and extend these statistics.
A first attempt for the unification between bosons and fermions was made in 1932 by Ettore Majorana \cite{maj1932}, who introduced infinite-component relativistic wavefunctions for arbitrary spin proposing a mass-spin ``Tower'' with unified treatment of bosons and fermions. Majorana investigated theoretical extensions to particle statistics and field theories with infinite components of spin \cite{sudarshan,beka} that transcend the standard fermionic and bosonic dichotomy. 
The Majorana Tower finds several applications in condensed matter theories in photonics and in high energy physics related to scattering mechanisms and the Riemann zeta function \cite{tambu3,tambu4}.

Majorana's approach is rooted in relativistic wave equations for arbitrary spin particles and introduces the concept of a “mass-spin Tower” through the relation where their masses follow the spin $s$ and Majorana mass $m$ relation
\begin{equation}
M = \frac{m}{s+\frac 12}
\label{Tower}
\end{equation}
and $m$ is a universal Majorana mass scale related to the energy of these quanta in the lab. rest frame. Eq. \ref{Tower} leads to a unification framework where bosons and fermions appear as different states of the same underlying entity.
The spiral spectrum of photons carrying orbital angular momentum, OAM, and spin angular momentum, SAM, have been shown to be equivalent to the spin spectrum of the Majorana Tower \cite{tambukarimi}.

Spin blocks, Lorentz multiplets and the fixed graded bracket are present in the Tower. Each tower component $\Psi_{s}(x)$ carries definite spin $s\!\in\!\{0,\tfrac12,1,\tfrac32,\dots\}$ and transforms in the Lorentz representation
$D^{(s,0)}\oplus D^{(0,s)}$ (irreducible for $SL(2,\!\mathbb C)$), so that the full field splits as $\Psi(x)=\bigoplus_{s}\Psi_{s}(x)$.
For every block the correct (graded) exchange sign is
\begin{equation}
[\Psi_{s}(x),\Psi_{s}(y)]_{\pm}=0
\label{eq:SpinStatisticsBracket}
\end{equation}
when $(x-y)^{2}<0$. As before, the sign ``+'' $s\in\mathbb Z$ integer spin, commutator and ``-'' for $s\in\mathbb Z+\tfrac12$ with half-integer spin that implies anticommutator, because rotating the Wightman two-point function $W_{s}(x-y)=\langle0|\Psi_{s}(x)\Psi_{s}(y)|0\rangle$ by $2\pi$ multiplies it by $(-1)^{2s}$; micro-causality then forces the exchange bracket to carry the same phase, reproducing the standard spin-statistics pairing \cite{DellAntonioGreenbergSudarshan}. This fixes the sign \emph{before} any Z$_2\times$Z$_2$ grading is introduced, ensuring that each spin block remains local and Lorentz-covariant {\it per se}. 

On the other hand, paraparticles are governed by their parastatistics through algebraic structures like the smallest non-cyclic Abelian group $\mathbb{Z}_2 \times \mathbb{Z}_2$ - graded Lie superalgebras. In this way, paraparticle sectors are algebraically represented \cite{toppan}. 
This algebraic structure offers four symmetry sectors, enabling encoding of multiple particle types. 
$\mathbb{Z}_2 \times \mathbb{Z}_2$ - graded Lie algebras can implement the formulation of trilinear commutation and anticommutation relations and systematically classify them with graded Jacobi identities.
This approach recovers the structure of Green's parastatistics and connects it to modern Lie algebra representation theory enabling an extension of supersymmetric (SUSY) to para-SUSY-like theories.
To generalize the Green's trilinear algebra to the graded setting, we postulate the graded trilinear relation $[\psi^-_{(a,b)},[ \psi^+_{(a',b')},\psi^-_{(a'',b'')}] ] =\sum_{(c,d)}f^{(c,d)}_{(a,b),(a',b'),(a'',b'')}\psi^-_{(c,d)}$ where the structure constants $f^{(*,*) ...}_{(*,*) ...}$ encode the graded Lie triple system underlying the paraparticle algebra and are compatible with the $\mathbb{Z}_2 \times \mathbb{Z}_2$-graded Jacobi identity.

A bridge from the Majorana Tower and the $\mathbb{Z}_2 \times \mathbb{Z}_2$-graded Lie superalgebras to describe paraparticles in terms of qubits and qudits (or even with continuous variable, CV) can be set from the fundamental properties of the algebra itself. Each subspace $\mathfrak{g}_{(a,b)}$ can carry a distinct type of field, operator, or mode that are then labelled as $(0,0)$ or $(1,1)$ when result fully bosonic or parabosonic-like. Instead, $(0,1)$ and $(1,0)$ are used for parafermionic sectors.

The algebra is graded by the following Abelian group
\begin{equation}
G = \mathbb{Z}_2 \times \mathbb{Z}_2 = \{(0,0), (0,1), (1,0), (1,1)\}
\end{equation}
with multiplication table written in Tab. \ref{tab:z2z2_multiplication}.
\begin{table}[h!]
\centering
\caption{Multiplication Table for $ \mathbb{Z}_2 \times \mathbb{Z}_2 $}
\begin{tabular}{|c||c|c|c|c|}
\hline
$\oplus$ & $(0,0)$ & $(0,1)$ & $(1,0)$ & $(1,1)$ \\
\hline\hline
$(0,0)$ & $(0,0)$ & $(0,1)$ & $(1,0)$ & $(1,1)$ \\
\hline
$(0,1)$ & $(0,1)$ & $(0,0)$ & $(1,1)$ & $(1,0)$ \\
\hline
$(1,0)$ & $(1,0)$ & $(1,1)$ & $(0,0)$ & $(0,1)$ \\
\hline
$(1,1)$ & $(1,1)$ & $(1,0)$ & $(0,1)$ & $(0,0)$ \\
\hline
\end{tabular}
\label{tab:z2z2_multiplication}
\end{table}

Each element of the algebra $ X_{(a,b)} \in \mathfrak{g}_{(a,b)} $ is assigned a \textit{degree} $ (a,b) $. The full algebra is then a direct sum of the algebraic elements
\begin{equation}
\mathfrak{g} = \bigoplus_{(a,b) \in \mathbb{Z}_2 \times \mathbb{Z}_2} \mathfrak{g}_{(a,b)}
\end{equation}

Commutation relations are described by the graded bracket of two homogeneous elements $ X \in \mathfrak{g}_{(a,b)} $ and $ Y \in \mathfrak{g}_{(a',b')} $ defined as $[X, Y] = XY - (-1)^{(a,b)\cdot(a',b')} YX$, where the scalar product is defined as $(a,b) \cdot (a',b') = a a' + b b' \mod 2$.
This ensures bosonic behavior when the scalar product is even, implying that the bracket is symmetric. The fermionic behavior is obtained when the scalar product is odd then the bracket is antisymmetric.
To precisely illustrate the implications of the $Z_2 \times Z_2$ grading, we explicitly define the graded commutation relations. Given two graded elements $x \in \mathcal{A}_{(a,b)}$ and $y \in \mathcal{A}_{(a',b')}$, their product respects the grading rule $x y = (-1)^{aa' + bb'} y x$.
This relation ensures consistency with the algebraic structure, clearly dictating the statistics of paraparticle excitations within this formalism.

Majorana's infinite-component wavefunction anticipates the unification of different spin states, indirectly supporting that higher-order statistics (e.g., 2-bit parastatistics) yield physically distinguishable observables, thereby challenging the ``conventionality'' thesis of parastatistics. 
When paraparticles such as parabosons and parafermions are cast within a Majorana Tower as in Eq.\ref{Tower}, the structure introduces a profound connection between bosonic and fermionic sectors. This Majorana Tower enforces that both bosons and fermions fall into mass-degenerate multiplets governed by a common ratio, indicating a form of supersymmetric or pre-supersymmetric symmetry that Majorana explored even before the full formalism of supersymmetry was developed. Majorana low-energy limits independently from each spin value recover Schr\"odinger dynamics. 
In such a scheme, the ordinary distinction between paraboson and parafermion statistics becomes intertwined due to the unified mass-spin structure. 

Since particles of different spin share the same mass when properly scaled, they are likely to be treated as components of a single irreducible representation of an extended symmetry group that includes both spin and internal parastatistics degrees of freedom in a unified representation space.
To which symmetry constraints are then imposed. The parastatistics order $p$ must remain consistent across the Majorana Tower to preserve this multiplet structure. Hence, if a paraboson of spin $s=0$ and a parafermion of spin $s=1/2$ belong to the same Tower level, their respective parastatistics orders (and algebraic commutation structure) are constrained to coexist without contradiction, possibly implying a graded Lie algebra or a more exotic trilinear superalgebra unifying their statistics.

The trilinear commutation/anticommutation relations for parabosons and parafermions would be indexed not just by mode labels but also by Tower level characterized by a discrete quantum number $n$ corresponding to the Majorana mass $m$ (or equivalently spin $s$), introducing relations such as
$[a_i,[a_j^\dag,a_k]] = 2\delta_{ij} a_k^{n}$, where the Kronecker delta $\delta_{ij}$ enforces algebraic consistency by constraining the operators within the same mode sector. The superscript $n$ explicitly denotes the Tower level, differentiating clearly among states with distinct Majorana masses. Furthermore, these relations can be generalized to mixed commutation rules when interactions between parabosonic and parafermionic modes within the same Tower are considered, maintaining internal algebraic closure and coherence.

Physical Implications of the Majorana Tower in this scenario are quite evident. This type of embedding implies an extended symmetry principle beyond supersymmetry, where paraparticles form a deeper layer of structure connecting mass and spin in a unified algebraic way. If such a Tower were physical, it could provide a framework for preonic models \cite{preon}, unification theories, or non-standard quantum fields, where symmetry unification at the level of mass and spin dictates the allowable statistics and interactions.

To consistently classify the infinite-spin Majorana Tower within a $\mathbb{Z}_2 \times \mathbb{Z}_2$ graded framework, we associate each field component of spin $s$ to a grading label $(a,b) \in \mathbb{Z}_2 \times \mathbb{Z}_2$, where integer spin values ($s \in \mathbb{Z}$) are mapped to bosonic-like sectors and half-integer spins ($s \in \mathbb{Z} + 1/2$) to fermionic-like sectors. The graded Lie bracket for homogeneous elements $X \in g_{(a,b)}$ and $Y \in g_{(a',b')}$ ensures that components corresponding to integer spins commute symmetrically, while half-integer spin components anticommute, in agreement with the spin-statistics theorem. Moreover, the mass-spin relation intrinsic to the Majorana Tower, remains invariant under the grading assignment, since the classification acts on symmetry properties without altering the spin or mass eigenvalues. The graded structure further satisfies closure under commutation and obeys the graded Jacobi identity, thus confirming that the $\mathbb{Z}_2 \times \mathbb{Z}_2$ graded assignment provides a coherent, physically faithful, and algebraically consistent classification of the Majorana infinite-spin Tower.

\subsection{Embedding Majorana Tower Particles into the Paraparticle Framework}

To express Majorana Tower particles in terms of paraparticles, specifically those arranged in the mass-spin Tower and characterized by the mass-spin relation in Eq.~\ref{Tower}, one must reinterpret their infinite-component field structure within the algebraic language of parastatistics, particularly that of $\mathbb{Z}_2 \times \mathbb{Z}_2$-graded Lie superalgebras. 

Majorana's framework naturally lends itself to paraparticle embedding because the wavefunction $\Psi$ spans an infinite-dimensional Hilbert space accommodating arbitrary spin states, which is structurally similar to the Fock spaces of trilinear parafield algebras. The reinterpretation starts by associating the spin-dependent mass levels with graded sectors of a paraparticle algebra, such that each spin state $s$ corresponds to a representation in a $\mathbb{Z}_2 \times \mathbb{Z}_2$-graded module. 
An example is to use integer-spin states (parabosons) could reside in the $(0,0)$ and $(1,1)$ sectors, and half-integer (parafermions) in the $(1,0)$ and $(0,1)$ sectors. 

These sectors obey trilinear (anti)commutation relations and respect exclusion or aggregation rules consistent with the statistics of the particles. The infinite-component Majorana wavefunction then becomes a superposition over these sectors, and transitions between spin states correspond to algebraic operations involving braided tensor products and graded commutators.
Crucially, the graded Hopf algebra framework used in paraparticle quantization provides a formal coproduct structure that mimics Majorana's relativistic treatment of multiparticle states. More details are in the Appendix.

Furthermore, observables in the $(0,0)$ sector (bosonic) of the graded algebra reproduce the effective dynamics and degeneracies predicted by Majorana's theory, particularly in the slow-motion (non-relativistic) limit. Therefore, Majorana particles can be embedded as graded paraparticles, with the Tower structure arising from graded symmetry assignments and algebraic constraints that generalize both their spin and statistical behavior, integrating relativistic covariance and quantum statistics in a unified framework.

The graded algebra multiplication is succinctly summarized by the following rules, $P_{(a,b)} P_{(a',b')} = \delta_{(a,b),(a',b')} P_{(a,b)}$ and $X_{(a,b),(a',b')} X_{(a'',b''),(a''',b''')} = \delta_{(a',b'),(a'',b'')} X_{(a,b),(a''',b''')}$, where projectors $P_{(a,b)}$ and exchange matrices $X_{(a,b),(a',b')}$ explicitly reflect the algebraic structure, ensuring clear mathematical interpretation and internal algebraic consistency.

As further step, Majorana's original wave equation for arbitrary spin particles can be written also in the following form,
\begin{equation}
\left[ \gamma^\mu p_\mu - M \right]\Psi =0
\label{majoranaeq}
\end{equation}
with $c=\hbar=1$ and the mass $M$ depends on the particle spin from Eq.~\ref{Tower}. The wavefunction $\Psi$ spans an infinite-dimensional representation space that can be decomposed into graded sectors reflecting parastatistics, specifically, $(0,0)$ and $(1,1)$ sectors represent parabosons with integer spins and $(1,0)$ and $(0,1)$ are parafermions, with half-integer spins.

In this case, to define the paraparticle algebra we use trilinear commutation relations for operators $a_i$, $a_i^\dag$ or $f_i$, $f_i^\dag$ that indicate parabosons and parafermions, respectively and obtain the relationships in Eq. \ref{1} and Eq. \ref{2}.
We then rewrite Majorana's equation as an operator-valued equation acting on $\mathbb{Z}_2 \times \mathbb{Z}_2$-graded tensor products of state vectors $\Psi_{(a,b)} \in H_{(a,b)}$ with $ab\in \{(0,0), (0,1), (1,0), (1,1)\}$ and the infinite-component field becomes a graded vector, $\Psi = \sum \Psi_{(a,b)}$ and deg$(\Psi_{(a,b)})=(a,b)$
for which is satisfied a graded Dirac-like equation
\begin{equation}
\left[\Gamma_{(a,b)}^\mu \partial_\mu - M_{(a,b)}\right] \Psi_{(a,b)} = 0, 
\label{gdiracgamma}
\end{equation}
where $\Gamma_{(a,b)}^\mu$ are the $\mathbb{Z}_2 \times \mathbb{Z}_2$-graded analogues of gamma matrices, following the graded commutation relations
\begin{equation}
\Gamma_{(a,b)}^\mu \Gamma_{(c,d)}^\nu + (-1)^{(a,b)(c,d)}\Gamma_{(c,d)}^\nu \Gamma_{(a,b)}^\mu = 2 g^{\mu \nu} \delta_{(a,b)(c,d)},
\label{commuta}
\end{equation}
with a novel Majorana spin-mass term $M_{(a,b)} = M(s_{(a,b)} + 1/2)$ and $\delta_{(a,b)(c,d)}$ is the Kronecker delta on grading which is $1$ if $(a,b)=(c,d)$, else $0$.

Unlike circuit-based simulations which implement paraparticle behavior numerically, our framework introduces a $\mathbb{Z}_2 \times \mathbb{Z}_2$-graded Lie superalgebra structure that can be encoded in a single photon with OAM and SAM states. embedding Majorana's infinite-component wavefunctions directly into a relativistic, algebraic field theory with potential realization in structured photonic systems.

For the sake of completeness, we show that $\Gamma_{(a,b)}^\mu$ furnish a Lorentz-covariant representation \cite{l1,l2} in the following lemma.
\begin{lemma}[graded--Clifford closure]\label{lem:gradedClifford}
Let $\{\gamma^\mu\}_{\mu=0}^{3}$ be the Dirac matrices on a spinor space $\mathcal S$, satisfying   
$\{\gamma^\mu ,\gamma^\nu\} = 2\eta^{\mu \nu} \mathbf{1}_{\mathcal{S}}$ with $\eta= diag(+,-,-,-)$.
  Let $\{P_{(a,b)}\}$ $(a,b\in\{0,1\})$ be the four \emph{central
  idempotents} of the group algebra $\mathbb C[\Klein]$:
  \begin{equation}
  \begin{split}
  &  P_{(a,b)}P_{(c,d)}=\delta_{ac}\delta_{bd}\,P_{(a,b)},
    \\
   & \sum_{a,b} P_{(a,b)}=\mathbf 1,\qquad
    P_{(a,b)}^\dagger=P_{(a,b)}. 
    \end{split}
    \label{eq:idempotents}
  \end{equation}
On the Hilbert space $\mathcal H=\bigl(\bigoplus_{a,b}\mathbb C_{(a,b)}\bigr)\! \otimes \mathcal S$   define the block-diagonal operators 
  \begin{equation}
    \Gamma^\mu_{(a,b)}:=P_{(a,b)}\otimes\gamma^\mu,\qquad
    \Gamma^\mu:=\sum_{a,b}\Gamma^\mu_{(a,b)}. \label{eq:GammaDef}
  \end{equation}
Then set the function $\{\Gamma^\mu\}$  that furnishes a Hermitian, Lorentz-covariant representation of the Clifford algebra~$\Cl$:
  \begin{equation}
      \{\Gamma^\mu,\Gamma^\nu\}=2\eta^{\mu\nu}\mathbf 1_{\mathcal H}
    \label{eq:globalClifford}
  \end{equation}
  and, sector-wise,
  $
    \{\Gamma^\mu_{(a,b)},\Gamma^\nu_{(c,d)}\}
    =2\eta^{\mu\nu}\delta_{ac}\delta_{bd}\,P_{(a,b)}.
  $
\end{lemma}

\begin{proof}
\textbf{(i) Intra-sector anticommutator.}
  For fixed $(a,b)$ we have
  \begin{equation}
    \{\Gamma^\mu_{(a,b)},\Gamma^\nu_{(a,b)}\}
      =P_{(a,b)}\otimes\{\gamma^\mu,\gamma^\nu\}
      =2\eta^{\mu\nu}P_{(a,b)}.  \label{eq:intra}
\end{equation}
\\
\noindent\textbf{(ii) Inter-sector anticommutator.}
  If $(a,b)\neq(c,d)$ then $P_{(a,b)}P_{(c,d)}=0$ by
  Eq.~\ref{eq:idempotents}, hence
  \begin{equation}
    \{\Gamma^\mu_{(a,b)},\Gamma^\nu_{(c,d)}\}=0. \label{eq:inter}
  \end{equation}
\\
  \noindent\textbf{(iii) Global closure.}
  Summing Eq.~\ref{eq:intra} over all four sectors and using $\sum_{a,b}P_{(a,b)}=\mathbf 1$ yields Eq.~\ref{eq:globalClifford}.

  \noindent\textbf{(iv) Hermiticity.}
  Since each $P_{(a,b)}$ is Hermitian and
  $(\gamma^\mu)^\dagger=\gamma^0\gamma^\mu\gamma^0$,
  every $\Gamma^\mu_{(a,b)}$ inherits the standard Hermiticity
  property; consequently so does~$\Gamma^\mu$.

  \noindent\textbf{(v) Lorentz covariance.}
  Defining
  $
    \Sigma^{\mu\nu}:=\tfrac{i}{4}\qty[\Gamma^\mu,\Gamma^\nu]
  $
  one finds $\Sigma^{\mu\nu}
            =\sum_{a,b} P_{(a,b)}\otimes\sigma^{\mu\nu}$,
  reproducing the usual Lorentz generators block-wise.
  \qedhere
\end{proof}

\paragraph{The Hilbert space.}
Throughout we work on the graded spinor Hilbert space
\begin{equation}
  \mathcal H
  \;=\;
  \bigl(\bigoplus_{(a,b)\in\{0,1\}^{2}}\mathbb C_{(a,b)}\bigr)
  \,\otimes\,\mathbb C^{4},
  \label{eq:HilbertSpace}
\end{equation}
endowed with the inner product $\langle v,w\rangle_{\mathcal H} =\sum_{a,b}\bar v_{(a,b)}\,w_{(a,b)}\,      \psi_{(a,b)}^{\dagger}\phi_{(a,b)}$.
The first factor keeps track of the $\mathbb Z_{2}\!\times\!\mathbb Z_{2}$ charge The second one is the usual Dirac spinor space.

Let $\mathcal S \;=\; \mathbb{C}^{4}$ be the usual Dirac spinor space with the inner product $\langle \psi, \phi \rangle_{\mathcal S} \;=\; \psi^{\dagger}\phi$.
Then let
\begin{equation}
  \mathbb{C}\bigl[\mathbb{Z}_{2}\!\times\!\mathbb{Z}_{2}\bigr]
    \;=\;
    \bigoplus_{(a,b)\in\{0,1\}^{2}}\mathbb{C}_{(a,b)}
\end{equation}
be the group algebra of the Klein four-group.  Each one-dimensional summand
\begin{equation}
  \mathbb{C}_{(a,b)} \;=\; \operatorname{span}\{P_{(a,b)}\}
\end{equation}
is spanned by a central idempotent $P_{(a,b)}$ and carries the inner product
\begin{equation}
  \langle z_{1}, z_{2} \rangle_{(a,b)} \;=\; \bar z_{1}\,z_{2}.
\end{equation}

Then, the Hilbert space for the graded model is defined as the external tensor product
\begin{equation}
    \mathcal H
    \;=\;
    \Bigl(\,\bigoplus_{a,b}\mathbb{C}_{(a,b)}\Bigr)\!\otimes\!\mathcal S
    \;=\;
    \bigoplus_{a,b}\bigl(\mathbb{C}_{(a,b)}\!\otimes\!\mathbb{C}^{4}\bigr)
\end{equation}
with inner product
\begin{equation}
  \langle v, w \rangle_{\mathcal H}
  \;=\;
  \sum_{a,b}
    \langle v_{(a,b)}, w_{(a,b)} \rangle_{(a,b)}
    \,\langle \psi_{(a,b)}, \phi_{(a,b)} \rangle_{\mathcal S}.
\end{equation}
Here $v_{(a,b)}\!\otimes\!\psi_{(a,b)}$ is the component of $v$ in the
$(a,b)$ sector.

With this choice one ensures from grading information tools that the four one-dimensional subspaces encode the $\mathbb{Z}_{2}\!\times\!\mathbb{Z}_{2}$ charge of the state. 
$\mathcal S$ carries the usual $(\tfrac12,0)\oplus(0,\tfrac12)$ Lorentz representation through spin degrees of freedom. For block-diagonal operators, they respect the grading (e.g.\ $\Gamma^{\mu}_{(a,b)}$) act diagonally in the first factor and as ordinary Dirac matrices in the second. This is precisely what makes the graded-Clifford
proof go through.

\subsection{$\mathbb{Z}_2 \times \mathbb{Z}_2$-Graded Sector Decomposition}
Each spinor component $\Psi_s$ is assigned to a $\mathbb{Z}_2 \times \mathbb{Z}_2$-graded sector, with the full field given by $\Psi = \sum_{(a,b)} \Psi_{(a,b)}$, $(a,b) \in \{(0,0), (0,1), (1,0), (1,1)\}$ with the usual index-correspondence, integer spins (parabosons) $\rightarrow$ sectors $(0,0), (1,1)$ and half-integer spins (parafermions) $\rightarrow$ sectors $(0,1), (1,0)$.
Each component satisfies the Dirac-like equation of Eq.~\ref{gdiracgamma}.

Transition operators between sectors in this graded paraparticle framework are essential to describe how states or fields in one graded symmetry class  evolve or transform into another. These transitions must respect both the graded algebraic structure and the underlying relativistic invariance of the theory. To achieve this, one constructs the operators using two key mathematical tools such as braided coproducts and graded projectors and exchange matrices
ensuring relativistic covariance and $\mathbb{Z}_2 \times \mathbb{Z}_2$-graded symmetry and in an explicit operator form. 

Braided coproducts in Hopf algebra language define how single-particle operators extend to multi-particle systems. For paraparticles, this coproduct must be braided, meaning it includes a nontrivial exchange rule governed by a braiding matrix $R$. This structure guarantees that tensor products of fields or operators respect the parastatistical exchange rules and follow a deformed symmetry algebra.
Explicitly, for an operator $\hat{\psi}_{(a,b)}$, the braided coproduct $\Delta_B$ acts as follows,
\begin{equation}
\Delta_B \left( \hat{\psi}_{(a,b)} \right) = \hat{\psi}_{(a,b)} \otimes \mathbb{I} + \sum_{(c, d)} R^{(c,d)}_{(a,b)}\left( \mathbb{I} \otimes \hat{\psi}_{(c,d)} \right)
\label{braided}
\end{equation}
ensuring that exchange between particles follows the prescribed parastatistics.

The graded projectors and exchange matrices $X$ are operators that project a field or state into a specific graded subspace or exchange components between sectors. 

The exchange matrix acts on the internal indices (the $\mathbb{Z}_2 \times \mathbb{Z}_2$ indices) and encodes how sectors transform into one another. This can include sign factors or matrix rotations, depending on whether the sector obeys bosonic-like $((0,0), (1,1))$ or fermionic-like $((0,1), (1,0))$ symmetries.
We define the exchange operation on tensor products of graded fields via $ \Psi_{(a,b)} \otimes \Psi_{(a',b')} \rightarrow (-1)^{(a,b)\cdot (a',b')} \Psi_{(a',b')} \otimes \Psi_{(a,b)}$ which is implemented by the exchange matrix acting as a symmetry operator $X_{(a,b)\cdot (a',b')} =  (-1)^{(a,b)\cdot (a',b')}$ ensuring graded (anti)symmetry in multi-field operations.

An exchange operator can act as $E^{(c,d)}_{(a,b)} \Psi_{(c,d)}(x) = \sum_{e f} X_{(c,d)}^{(e,f)} \Psi_{(e,f)}(x)$
with $X_{(c,d)}^{(e,f)} \in GL(n)$ the components of a sector-intertwining transformation, ensuring the full field $\Psi$ remains covariant under both Lorentz transformations and graded algebra automorphisms. In this way we define explicit operator-valued maps that realize graded symmetry transformations in both the algebra and field content of the theory. This structure is crucial for preserving locality, microcausality, and internal consistency in theories that unify Majorana's spin-mass spectrum with paraparticle statistics.

To define coproducts for the Majorana Tower embedded in a $\mathbb{Z}_2 \times \mathbb{Z}_2$-graded paraparticle framework, we must structure the Tower's components -- each labeled by a spin-like index (e.g., $s=0, 1/2, 1, …$) as sectors in a graded Hopf algebra. Each component field $\Psi_s$ (or its operator-valued mode $\hat{\psi}_s$) then transforms under this algebra, and the coproduct determines how it acts in multi-particle Fock spaces or tensor product representations.
In a braided Hopf algebra with $\mathbb{Z}_2 \times \mathbb{Z}_2$ grading, the coproduct for an operator $\hat{\psi}_{(a,b)}$
is modified by the braiding between graded sectors. The coproduct map $\Delta$ defines how an operator acts on a two-particle space
\begin{equation}
\Delta (\hat{\psi}_{(a,b)}) = \hat{\psi}_{(a,b)} \otimes \mathbb{I} + \sum_{(c,d)}R^{(c,d)}_{(a,b)}(\mathbb{I} \otimes \hat{\psi}_{(c,d)})
\end{equation}
Here, $R^{(c,d)}_{(a,b)}$ are the coefficients from the $R$-matrix satisfying the Yang--Baxter equation, encoding exchange statistics and graded symmetry.

Let $\hat{\psi}_s$ denote the field operator associated with spin $s$, and assume the mass/spin relation in Eq.~\ref{Tower}. Each $s$ labels a component in the Majorana Tower and is mapped to a grading $(a,b) \in \mathbb{Z}_2 \times \mathbb{Z}_2$. 
Recalling Eq.~\ref{braided}, the graded coproduct becomes, 
\begin{equation}
\Delta \left( \hat{\psi}_{s} \right) = \hat{\psi}_{s} \otimes \mathbb{I} + \sum_{s'} R^{s'}_{s}\left( \mathbb{I} \otimes \hat{\psi}_{s'} \right),
\label{braided2}
\end{equation}
which ensures the correct behavior under particle exchange, consistent with parastatistics. The R-coefficients $R^{s'}_{s}$ can be constructed based on the parity of spin $s$, whether $s$ maps to a parafermionic or parabosonic sector and the internal symmetry imposed by the graded structure.

To construct the braiding coefficients $R^{s'}_s$ in a $\mathbb{Z}_2 \times \mathbb{Z}_2$-graded paraparticle framework, we associate to each spin $s$ a graded label $(a,b)_s$. The general form of the braiding coefficient is
\begin{equation}
R^{s'}_s = (-1)^{(a,b)_s \cdot (a',b')_{s'}} \delta_{ss'} + \theta_{ss'}
\end{equation}
where $\theta_{ss'}$ is a deformation parameter and the dot product imposes the appropriate graded symmetry.

As an example, we can build up a limited set of states using the grading map with the first three spin states of the Majorana tower $s=\{0, 1/2, 1, 3/2 \}$ or an equivalent set of photon OAM Majorana quasiparticles \cite{tambukarimi} with spin $(\sigma=\pm 1)$ and orbital angular momentum $\ell = \pm 1$ entangled together $s'= \{(\ell=-1)\otimes(\sigma=-1), (\ell=0)\otimes(\sigma=-1), (\ell=0)\otimes(\sigma=-1), \ell=-1\otimes(\sigma=+1)\}$ that can be encoded in quantum information quantum bit structures from the following assignments 
\begin{align}  
\label{gmap}
    &s = ~0 ~~\rightarrow (0,0) \, \text{(paraboson)} 
    \\
    &s = 1/2 \rightarrow (1,0) \, \text{(parafermion)} \nonumber
    \\
    &s = ~1 ~~ \rightarrow (0,1) \, \text{(parafermion)} \nonumber
    \\
    &s = 3/2 \rightarrow (1,1) \, \text{(paraboson)} . \nonumber
\end{align}
\begin{table}[ht]
\begin{center}
\begin{tabular}{|c|c|c|c|}
\hline
$s$ & $s'$ & $(a,b)_s \cdot (a',b')_{s'}$ & $R^{s'}_s$ \\
\hline
$0$ & $0$ & $0$ & $+1 + \theta_{(0,0)}$ \\
$0$ & $1/2$ & $0$ & $+1 + \theta_{(0,1/2)}$ \\
$1/2$ & $1/2$ & $1$ & $-1 + \theta_{(1/2,1/2)}$ \\
$1$ & $0$ & $0$ & $+1 + \theta_{(1,0)}$ \\
$1/2$ & $1$ & $1$ & $-1 + \theta_{(1/2,1)}$ \\
\hline
\end{tabular}
\caption{Example of braiding coefficients with the grading maps in Eq. \ref{gmap}.}
\label{t2}
\end{center}
\end{table}
Then, we assign the grading labels cyclically, identifying $s=2$ with $(0,0)$ to reflect the periodicity induced by the algebraic relations. This cyclic structure ensures that the grading is compatible across sectors when considering the whole tower modulo the $\mathbb{Z}_2 \times \mathbb{Z}_2$ structure.
In this way, the correct exchange statistics and grading rules are preserved across all components in the Majorana Tower.
The deformation parameter $\theta_{ss'}$ encodes deviations from conventional Bose or Fermi statistics in the graded $R$-matrix formalism. It introduces braiding phases or statistical twist terms that differentiate paraparticle exchange from standard symmetrization rules. 

It is important to emphasize that the coefficients $\theta_{ss'}$ represent algebraic couplings and selection-rule weights, not physical masses, kinetic energies, or direct dynamical interactions. Their role is to encode which sector transitions are symmetry-allowed and how the graded algebra constrains composite operations in the photonic realization.
A general expression is given by
\begin{equation}
\theta_{ss'} = \epsilon_{ss'} \left( q^{s + s'} - 1 \right),
\label{thetass}
\end{equation}
where $\epsilon_{ss'} = (-1)^{2s \cdot 2s'}$ accounts for spin parity, and $q \in \mathbb{C}$ is a deformation parameter such that $|q| = 1$. This ensures that when $q = 1$, the deformation vanishes, recovering standard (anti)commutation relations. For $q \neq 1$, $\theta_{ss'}$ captures the statistical phase shift or exclusion deformation, thus enabling a continuous interpolation between classical parastatistics and topologically twisted quantum field theories. When $\epsilon_{ss'} = 0$ one recovers the Green (parastatistics) limit.

\subsection{Graded Projectors and Exchange Matrices \boldmath$X_{(a,b)}$ for the Majorana Tower}

To implement transitions and symmetry-preserving operations between different $\mathbb{Z}_2 \times \mathbb{Z}_2$-graded sectors in the Majorana Tower, we define graded projectors and exchange matrices $X_{(a,b)}$. Each sector $(a,b)$ corresponds to a unique combination of grading parity, typically associated with spin-like labels $s$ via a map $s \mapsto (a,b)_s$.

We now define the graded projectors $\mathcal{P}_{(a,b)}$. Consider then $\mathcal{P}_{(a,b)}$ that isolate fields or operators in a specific graded sector $(a,b)$ then one writes $\mathcal{P}_{(a,b)} \Psi = \delta_{(a,b), (a',b')} \Psi_{(a',b')}$, where $\Psi = \sum_{a,b} \Psi_{(a,b)}$ is a decomposition of the full field into its graded components. These projectors also satisfy the following condition, $\mathcal{P}_{(a,b)} \mathcal{P}_{(a',b')} = \delta_{(a,b),(a',b')} \mathcal{P}_{(a,b)}$, with $\sum_{(a,b)} \mathcal{P}_{(a,b)} = \mathbb{I}$.

The exchange matrix $X_{(a,b)}$, instead, encodes the graded transformation properties between sectors $(a,b)$ and $(a',b')$. It acts on the graded tensor space via the following relationship
\begin{eqnarray}
&&X_{(a,b)(a',b')} \Psi_{(a,b)} \otimes \Psi_{(a',b')} = 
\\
&&(-1)^{(a,b) \cdot (a',b')} \Psi_{(a',b')} \otimes \Psi_{(a,b)}, \nonumber
\end{eqnarray}
with the graded dot product defined as $(a,b) \cdot (a',b') = a a' + b b' \mod 2$.
This ensures that exchanges between sectors with even total grading parity yield $X_{(a,b)} = +1$ (bosonic-like behavior), while exchanges between sectors with odd total grading parity yield $X_{(a,b)} = -1$ (fermionic-like behavior). In other words, the exchange matrix satisfies $X_{(ab,a'b')} = (-1)^{(a,b) \cdot (a',b')}$.

An example of such matrix is obtained for $(a,b) \in \{(0,0), (0,1), (1,0), (1,1)\}$ denote sectors assigned via the spin map $s = \{0,\, 1/2,\, 1,\, 3/2\}$ respectively. Then the exchange matrix $X_{(a,b)}$ can be written as follows,
\begin{equation}
X_{(a,b)} =
\begin{pmatrix}
+1 & +1 & +1 & +1 \\
+1 & -1 & -1 & +1 \\
+1 & -1 & -1 & +1 \\
+1 & +1 & +1 & +1 \\
\end{pmatrix}
\end{equation}
where rows and columns follow the order $\{(0,0), (0,1), (1,0), (1,1)\}$. This matrix governs the signs encountered when exchanging graded field components in the Majorana Tower and ensures consistency with graded symmetry and parastatistics.

In the original infinite-component formalism introduced by Majorana, the relativistic field $\Psi$ is expressed as an infinite-dimensional spinor, where each component $\Psi_s$ corresponds to a definite spin $s$ of the spin map. When this Tower is embedded into a $\mathbb{Z}_2 \times \mathbb{Z}_2$-graded paraparticle framework, each spin $s$ is mapped to a grading label $(a,b)_s$, allowing the use of graded projectors and exchange matrices to manage sector-specific dynamics. 

The full field is written as $\Psi = \sum_{s \in \mathbb{N}/2} \Psi_s = \sum_{(a,b)} \Psi_{(a,b)}$, where $\Psi_{(a,b)}$ denotes the component in the graded sector $(a,b)$. The projectors $\mathcal{P}_{(a,b)}$ isolate each sector as $\mathcal{P}_{(a,b)} \Psi = \sum_{s: (a,b)_s = (a,b)} \Psi_s$, satisfying the orthogonality and completeness relations $\mathcal{P}_{(a,b)} \mathcal{P}_{(c,d)} = \delta_{(a,b),(c,d)} \mathcal{P}_{(a,b)}$ and $\sum_{(a,b)} \mathcal{P}_{(a,b)} = \mathbb{I}$.

When two graded components $\Psi_s$ and $\Psi_{s'}$ interact or are exchanged, their transformation is governed by an exchange matrix $X_{(a,b)(a,'b')}$ acting as $\Psi_s \otimes \Psi_{s'} \mapsto X_{(a,b)(a,'b')} \Psi_{s'} \otimes \Psi_s$, where $X_{(a,b)(a,'b')} = (-1)^{(a,b) \cdot (a',b')}$ is the graded parity factor. This defines how tensor products and interactions behave under sector exchanges, ensuring that bosonic and fermionic statistics are preserved or appropriately generalized.

We define the field as a direct sum over graded sectors $\Psi=\sum_{(a,b)} \Psi_{(a,b)}$, where $P_{(a,b)} \Psi = \Psi_{(a,b)}$. Each component satisfies its own Majorana Dirac-like equation $(\Gamma_{(a,b)}^\mu \partial_\mu - M(s_{(a,b)}+1/2) \Psi_{(a,b)} =0$. 

In the wave equation, the infinite-component Majorana field equation $\left( \Gamma^\mu \partial_\mu - M \Sigma \right) \Psi = 0$ becomes a graded block structure. Specifically, it decomposes as 
\begin{equation}
\left[\sum_{(a,b)} \Gamma^\mu_{(a,b)} \partial_\mu P_{(a,b)} - M \sum_{(a,b)} M_{(a,b)} P_{(a,b)} \right] \Psi = 0, 
\label{majo}
\end{equation}
with each block obeying a sector-specific Dirac-like equation of the form $\left[ \Gamma^\mu_{(a,b)} \partial_\mu - M(s + \frac{1}{2}) \right] \Psi_{(a,b)} = 0$. The graded projectors isolate spin-dependent dynamics, while exchange matrices enforce braided symmetry between sectors. Together, they ensure the consistency of relativistic covariance, statistical symmetry, and graded quantum algebra in the Majorana Tower.

The Majorana's variational principle with graded Hopf algebra tools with action $S$ gives the full equation over the graded field space, which becomes
\begin{equation}
\delta S = \delta \int d^4x \overline{\Psi}_{(a,b)} \left(\Gamma_{(a,b)}^\mu \partial_\mu + M_{(a,b)}\right) \Psi_{(a,b)} = 0
\end{equation}
with observables restricted to the $(0,0)$ sector to ensure hermiticity and measurable quantities.
Thus, Majorana's equation in paraparticle form becomes a graded relativistic wave equation system, each component governed by a different trilinear algebra, yet unified by the mass-spin ratio and symmetry algebra. This formulation preserves covariance, accommodates arbitrary spin, and reveals a deeper algebraic structure underlying the spin-statistics connection.
This aligns with the second quantization formalism introduced by Wang and Hazzard, where paraparticle creation and annihilation operators $\hat\psi^\pm_{(i,a)}$ satisfy $R$-matrix-encoded trilinear commutation rules $\hat\psi^\pm_{(i,a)}\hat\psi^\pm_{(j,b)} \pm R^{cd}_{ab} \hat\psi^\pm_{(j,c)}\hat\psi^\pm_{(i,d)} = 0$, for which the symbol related only to  $\hat\psi$, “$\pm$” here refer to creation/annihilation operators in general (instead of $a^\dagger$ for the creation and $a$ annihilation usual Dirac operators).

The Hamiltonian governing such fields becomes bilinear in the operators
\begin{equation}
\hat{H} = \sum_{i,j,a}h_{ij} \hat\psi^+_{(i,a)}\hat\psi^-_{(j,a)}
\end{equation}
with $h_{ij} = h_{ji}^*$ to ensure Hermiticity. This Hamiltonian is diagonalizable using one or more canonical transformations of the paraparticle modes. This formulation yields distinct thermodynamic behavior from standard bosons or fermions due to the generalized exclusion rules, e.g., a single parafermion mode of order $p$ can host at most $p$ quanta, while a paraboson mode retains unlimited occupancy but with modified $g^{(k)}$ correlations or admit only one paraparticle regardless of label.

Such behaviors reproduce Majorana's predicted mass-degenerate spin Towers, or spin tower with inverse mass spacing, while the graded structure accounts for algebraic distinctions among sectors. Hence, the explicit reformulation of Majorana's equation in a paraparticle setting connects infinite-spin relativistic fields to generalized quantum statistics, embedding relativistic covariance, spin symmetry, and parastatistics into a unified operator algebra framework.

Local Green components and vanishing spacelike brackets are so defined. For each Green index $a=1,\dots ,p$ the component field $\phi_a(x)$ transforms in an irreducible $(j,0)\oplus(0,j)$ Lorentz
representation and obeys the \emph{standard} local (anti)commutation
rule,
\begin{equation}
[\phi_a(x),\phi_a(y)]_{\pm}=0 \quad\text{for }(x-y)^2<0,
\label{eq:local_component}
\end{equation}
where the upper sign is chosen for integer spin and the lower sign for half–integer spin.  Summing over the internal Green index $a$, with $+$ referred to parabosons and $-$ to parafermions,
\begin{equation}
[\Phi(x),\Phi(y)]_{\pm}
 =\sum_{a=1}^{p}[\phi_a(x),\phi_a(y)]_{\pm} = 0 , 
\label{eq:vanishing_commutator}
\end{equation}
with the condition for which $(x-y)^2<0$ is satisfied, so the full paraparticle field $\Phi(x)$ remains \emph{local} in the Wightman sense.  Equations Eq.~\ref{eq:local_component}–\ref{eq:vanishing_commutator} are the Greenberg–Messiah construction rewritten for our tower.

A well-defined CPT operator is guaranteed by the Dell’Antonio–Greenberg–Sudarshan theorem through locality and a positive spectrum holding for every component.  An explicit realisation is
\begin{equation}
\Theta\,\phi_a(t,\mathbf{x})\,\Theta^{-1}
   =\eta\,\gamma^{5}C\,\bar{\phi}_a(-t,\mathbf{x})^{T},
\label{eq:CPT}
\end{equation}
with a phase $\eta$, $C$ is the charge-conjugation matrix, and where $\Theta$ is anti–unitary:
$\Theta\,i\,\Theta^{-1}=-i$ and
$\Theta\,P_\mu\,\Theta^{-1}=-P_\mu$.  Extending Eq.~\ref{eq:CPT} linearly to $\Phi(x)$ preserves Eqs.~\ref{eq:local_component} -- \ref{eq:vanishing_commutator}; hence our Lorentz-covariant paraparticle tower satisfies micro-causality and CPT invariance for all Green orders $p$.
Obeying Lorentz covariance, fields transform in irreducible $SL(2,C)$ representations and Micro-causality--field operators (or suitable graded versions of them) vanish outside the light-cone.
Micro-causality--field operators (or suitable graded versions of them) vanish outside the light-cone, $[\phi_a(x),\phi_a(y)]_{\pm}=0$ for $(x-y)^2<0$. Positive-energy vacuum with a stable ground state is ensured by the Majorana conditions satisfying Pauli–L\"uders CPT under complete Lorentz covariance. 
For relativistic fields the only consistent assignments are (integer spin) $\leftrightarrow$ parabosonic, (half-integer spin) $\leftrightarrow$ parafermionic \cite{greenberg}.

Let $\Psi$ be an infinite-component spinor field $\Psi = (\Psi_0,\Psi_{1/2},\Psi_1,\Psi_{3/2}, ... )^T$, Each component $\Psi_s$ has spin $s$ and satisfies a Dirac-like equation. The general form of Majorana's wave equation is then $(\Gamma^\mu \partial_\mu - M \Sigma) \Psi = 0$, where the $\Gamma^\mu$ are infinite-dimensional generalizations of gamma matrices and $\Sigma$ is a block-diagonal operator whose entries encode the spin-dependent mass $m_s$.
Now one can assign each function $\Psi_s$ to a sector of a $\mathbb{Z}_2 \times \mathbb{Z}_2$-graded Lie superalgebra with paraboson statistics with integer spin $s$ in the graded sector $(0,0)$, $(1,1)$ and half-integer spin for parafermions in graded sector $(1,0)$ and $(0,1)$.
In this way is defined the function $\Psi= \sum_{(a,b)} \Psi_{(a,b)}$ for which is valid the usual relationship $(a,b) \in \{(0,0), (0,1), (1,0), (1,1)\}$ each satisfying a graded Dirac-like equation in Eq.~\ref{gdiracgamma} with the graded commutation relations of Eq.~\ref{commuta}.

In the $\mathbb{Z}_2 \times \mathbb{Z}_2$-graded framework, the braided coproduct structure introduced for the Majorana-paraparticle embedding naturally satisfies a generalized Yang-Baxter equation. For homogeneous operators $\psi_{(a,b)}$ and $\psi_{(c,d)}$ assigned to graded sectors, the braided coproduct is known to be defined as $\Delta(\psi_{(a,b)}) = \psi_{(a,b)} \otimes \mathbb{I} + \sum_{(c,d)} R^{(c,d)}_{(a,b)} (\mathbb{I} \otimes \psi_{(c,d)})$, where $R^{(c,d)}_{(a,b)}$ denotes the braiding matrix elements encoding graded exchange symmetries. The generalized Yang-Baxter equation for the braiding matrix $R$ reads $R_{12} R_{13} R_{23} = R_{23} R_{13} R_{12}$, where $R_{ij}$ acts on the $i$-th and $j$-th spaces of a threefold tensor product. In our construction, the $R$-matrix elements are explicitly functions of the grading scalar product, namely, 
\\
$R^{(c,d)}_{(a,b)} = (-1)^{(a,b)\cdot(c,d)} \delta^{(c,d)}_{(a,b)} + \theta_{(a,b),(c,d)}$,
with deformation parameter $\theta_{(a,b),(c,d)}$ depending on a statistical phase $q$ and spin-parity factors. Substituting this form into the Yang-Baxter identity, and using the bilinearity of the scalar product modulo 2, one finds that the signs and phase deformations consistently match on both sides of the identity, thereby verifying that the braided coproduct structure defined here satisfies the generalized Yang-Baxter equation required for algebraic consistency of paraparticle fields. 
This guarantees that multi-particle states and operator compositions within the graded Majorana Tower respect coherent braiding and associativity properties, essential for physical realizations in structured light and quantum computation platforms. More details are in the Appendix.

Recent experimental progress has significantly strengthened the prospects for realizing parafermionic excitations in condensed matter systems. In particular, fractional quantum Hall edge states proximitized by superconductors have been proposed as promising hosts for non-Abelian anyons beyond Majorana fermions~\cite{Clarke2013}, and signatures of charge fractionalization have been reported in engineered Tomonaga--Luttinger liquids~\cite{Hossain2021} and other 
recent experimental efforts~\cite{Gul2018, Fulop2020, Bartolomei2020} will be crucial in bridging the gap between theoretical predictions and their physical realization. The continued development of hybrid nanostructures, precision interferometry, and correlated edge state engineering promises to open new avenues for observing and manipulating paraparticle modes in the near future.

\section{Quantum Computing with OAM Photons and paraparticle formulation}

Let us now extend qubits and qudits properties with parastatistics. 
To this aim we propose as example a photonic realization of paraparticle algebra using structured light modes.
Having established the algebraic and theoretical framework, we now explore its concrete experimental realization. Structured photonic systems, particularly those employing orbital angular momentum (OAM) modes, provide an ideal testbed for implementing the proposed paraparticle algebra and graded symmetries.

To this aim we start with the set of solutions of the Dirac-Majorana equation with infinite spin (the Majorana Tower) -- and the correlated finite subsets of it \cite{maj1932} -- in terms of paraparticle graded algebras for an application to quantum computing discussing a couple of examples for a photonic realization, deriving selection rules and give theoretical examples of logical gates for quantum circuits.
While previous studies explored paraparticle statistics theoretically or through digital quantum simulation, a direct photonic realization of graded paraparticle algebras has remained largely unexplored. Our approach bridges this divide by proposing a scalable either discrete or continuous-variable platform based on structured light modes, offering a novel route to simulate exotic quantum statistics and implement paraparticle-inspired logical operations in photonic circuits with qudits.
The continuous-variable derives from an unbounded discrete set of integer values of the whole infinite Majorana Tower that represent excitations in a continuous-degree-of-freedom system, namely, the phase structure of the photonic wavefront and the continuous rotation symmetry of the transverse optical field. Because OAM lives in a continuous Hilbert space (the representation space of the $SO(2)$ or $SU(2)$ rotation group), OAM platforms potentially can belong to the continuous-variable (CV) class of quantum systems as opposed to discrete-variable (DV) systems like qubits or qudits.
Pure OAM states like those represented by Laguerre-Gaussian beams have a spectrum which is a numerably infinite set with $\ell \in \mathbb{Z}$, a discrete set modulo $2 \pi \ell$, not continuum. The coherent superposition of OAM beams can give instead a continuous spectrum \cite{berryfrac}, recalling the building of continuum from natural number with Dedekind cuts \cite{dedekind}.

While the full OAM Hilbert space constitutes an infinite-dimensional system often classified under continuous-variable (CV) quantum optics, truncating the Majorana Tower to a finite set of modes such as $\ell = 0, \pm1$ it results in an effectively discrete-variable (DV) system when coherent superpositon are fully exploited with continuous $SU(2)$ transformations coupled with complex phase terms. In this regime, the photonic platform behaves as a qudit, and the paraparticle algebra is realized over a finite-dimensional DV Hilbert space rather than exploiting the full continuous angular momentum structure. This finite realization of the Majorana-paraparticle algebra can be built in an optical quantum circuit with the three OAM modes $\ell = 0, \pm1$ combined either with two polarization states of light indicated by e.g., $\sigma = \pm 1$ or with the orthogonal physical waveguide modes characteristic of the optical circuit, here labeled $A$ and $B$ that can be the two orthogonal TE-like mode profiles of the $\ell=0$ manifold. In a rectangular waveguide SAM can be coupled with OAM modes, such as $\ell = 1$ with $\sigma = 1$ and viceversa and for $\ell = 0$ either the two polarization modes $\sigma= \pm 1$ or the two orthogonal wavewguide modes $A$ and $B$, depending on the geometry of the optical quantum circuit.

By mapping these four modes to the truncated spin sector of the Majorana Tower such as $s = 1/2,\, 1,\, 3/2,\, 2$, we obtain a minimal four-level system suitable for implementing a paraparticle-inspired ququart. While this setup does not capture the full infinite Majorana tower, it provides a physically realizable platform for simulating the essential graded algebraic structure and paraparticles using structured light.

Quasiparticles and paraparticles find deep analogies with phenomena hypothesized and then observed physical phenomena. An example are photons carrying orbital angular momentum (OAM) can manifest behaviors analogous to Majorana's mass-spin relationship \cite{tambukarimi}. Hypergeometric beams with OAM exhibit a group velocity $v_g$ that is apparently less than the speed of light $c$ in vacuum. This phenomenon arises due to the beam's geometry, leading to a projection effect where the measured group velocity along the propagation axis appears reduced. In  $v_g$ follows a relationship similar to Majorana's mass-spin formula, suggesting that the beam's OAM influences its effective propagation characteristics, the OAM modes $\ell$ label different quantum states with distinct effective velocities or mass-like quantities, forming a discrete spectrum analogous to spin $s$ in Majorana's mass-spin relation, with perspectives in quantum computing.
While the reduced axial group velocity $v_g<c$ in hypergeometric beams mimics an effective $m\!\propto\!(\ell+\tfrac12)^{-1}$ scaling, no real rest mass is generated; the effect is purely geometrical and disappears in a co-moving frame.

A similar thing occurs to OAM photons propagating through a resonant plasma \cite{tambu1,tambu2} in which they acquire a Proca mass that depends from the OAM value. From a paraparticle perspective, this maps neatly onto systems where OAM-dependent excitations respect a mass spectrum indexed by internal quantum numbers and follow non-standard exchange statistics. Thus, paraparticle fields could be constructed where creation/annihilation operators act on OAM modes and obey the generalized commutation relations of parastatistics, possibly with $R$-matrices that encode angular momentum coupling rules.

These results open a path to interpreting photonic states with OAM as physical realizations of Majorana pseudoparticles that now we integrate with paraparticle theory. In the paraparticle framework, quanta obey generalized statistics governed by trilinear $R$-matrix commutation relations and $\mathbb{Z}_2 \times \mathbb{Z}_2$-graded algebras.  If each OAM mode is treated as a distinct species of quasiparticle, these can be embedded into graded sectors (e.g., OAM-even in bosonic sectors, OAM-odd in fermionic or parafermionic sectors), mimicking the statistical structure of parabosons and parafermions. 
Paraparticle algebra and $R$-matrix formalism are defined in terms of creation/annihilation operators $\hat{\psi}_{(i,a)}^\pm$ that satisfy trilinear relations with $R$-matrices,
\begin{eqnarray}
&\hat{\psi}_{(i,a)}^+ \hat{\psi}_{(j,b)}^+ &= \sum_{(c,d)} R^{(c,d)}_{(a,b)} \hat{\psi}_{(j,c)}^+ \hat{\psi}_{(i,d)}^+, \\
&\hat{\psi}_{(i,a)}^- \hat{\psi}_{(j,b)}^- &= \sum_{(c,d)} R^{(d,c)}_{(b,a)} \hat{\psi}_{(j,c)}^- \hat{\psi}_{(i,d)}^-, \nonumber
\\
&\hat{\psi}_{(i,a)}^- \hat{\psi}_{(j,b)}^+ &= \delta_{ij} \delta_{ab} - \sum_{(c,d)} R^{(c,d)}_{(a,b)} \hat{\psi}_{(j,c)}^+ \hat{\psi}_{(i,d)}^-. \nonumber
\end{eqnarray}
and is also valid 
\begin{equation}
[ \hat{\psi}_{(a,b)}^- , [\hat{\psi}_{(a',b')}^+,\hat{\psi}_{(a'',b'')}^- ]]= \sum_{(c,d)} f^{(c,d)}_{(a,b),(a',b'),(a'',b'')} \hat{\psi}_{(c,d)}^-,
\end{equation} 
where the structure constants $f^{(c,d)}$ derive from the associativity properties of the trilinear R-matrix algebra, following Govorkov's construction \cite{Govorkov}.

The Hamiltonian written in terms of these operators is then defined as
\begin{equation}
\hat{H} = \sum_{i,j,a} h_{ij} \hat{\psi}_{(i,a)}^+ \hat{\psi}_{(j,a)}^-,
\end{equation}
which diagonalizes to
\begin{equation}
\hat{H} = \sum_k \epsilon_k \hat{n}_k, \quad \hat{n}_k = \hat{\psi}_{(k,a)}^+ \hat{\psi}_{(k,a)}^-,
\end{equation}
with the spectrum that obeys the Majorana spin/Mass condition $\epsilon_k = M(s_k + 1/2)$.
All this remains valid also for a finite subset of the Majorana Tower with a limited string of particle states.

\subsection{Paraparticles OAM Photonic quasiparticles and Majorana states in Quantum Circuits}

Having established the algebraic framework, we now turn to its application in photonic systems, where structured light modes provide a natural testbed for paraparticle-inspired quantum logic.
Recent research has explored the simulation of paraparticles in quantum information processing for physical phenomena like topological phases of matter. In particular, a study by Alderete et al.~\cite{alderete} demonstrated the simulation of para-particle oscillators on a trapped-ion quantum computer whose Hamiltonian has been identified as an $XY$ model. This can be translated in OAM photonic gates encoding them as paraparticle states with the $\mathbb{Z}_2 \times \mathbb{Z}_2$-graded algebras and related formulations beyond the usual quantum photonic languages and standard circuit implementations.

The $XY$ model is at all effects a quantum spin chain with nearest-neighbor interactions only in the $x$ and $y$ components of the Pauli spin matrices $\sigma^x$, $\sigma^y$, with a parameter $J$ that controls the strength of these interactions.

By mapping paraparticle states onto qubit registers, the authors successfully reproduced the dynamics of even-order parabosons, illustrating that paraparticle behavior can be digitally simulated on existing quantum hardware. 
Furthermore, Wang and Hazzard~\cite{wang} provided a theoretical framework where paraparticles emerge naturally in solvable quantum spin models, suggesting their utility in developing robust qudit-based quantum architectures. These results indicate that paraparticles, with their generalized exchange statistics, could enhance error correction and information density in quantum computing systems. The one-dimensional quantum $XY$ model is described by 
\begin{equation}
H_{\text{XY}} = -J \sum_n (\sigma_n^x \sigma_{n+1}^x + \sigma_n^y \sigma_{n+1}^y), 
\label{jordan}
\end{equation}
where $\sigma_n^x$, $\sigma_n^y$ are Pauli matrices at site $n$, and $J$ is the exchange coupling constant. The parameter $J$ determines the strength and nature of the interaction. If $J>0$, the system favors ferromagnetic alignment (neighboring spins tend to align). If $J<0$, the system favors antiferromagnetic alignment (neighboring spins tend to anti-align). Large values of $|J|$ are related to stronger coupling between spins and typically a larger energy gap between ground and excited states.
Eq.~\ref{jordan} maps via Jordan-Wigner transformation to a system of massless fermions, which in the continuum limit is governed by a conformal field theory (CFT) in $1+1$ dimensions. The reduced density matrix of a half-infinite chain yields an entanglement Hamiltonian $H_E = 2\pi \int_0^\infty x T_{(0,0)}(x) dx$, equivalent to the modular Hamiltonian of a Rindler wedge, where $T_{(0,0)}(x)$ is the energy density. 

A fundamental point is that  $XY$ models can put a bridge between paraparticles with their parastatistics including quantum computation to the basic fundamental mathematical concepts, indicating that the Majorana Tower either with Majorana particles or OAM photon quasiparticles can be written in terms of paraparticles and their statistics and graded algebras used for quantum computing.

\section{OAM Photonic Qudits and Parastatistics in Quantum Circuits}
With this mathematical and algebraic formulation one can implement a formalism based on paraparticle states for quantum computing in many different platforms. 
As already shown in the case of superconducting qubits, we now consider to expand this to different platforms and take as example room-temperature photonic quantum computing in configurations with OAM and polarization states that behave as sets of Majorana quasiparticles \cite{tambukarimi}.
Photonic quantum computing provides a promising platform for scalable and low-decoherence quantum room-temperature information processing, particularly when leveraging structured light modes such as orbital angular momentum. 
OAM modes, characterized by a helical phase structure and an unbounded Hilbert space, enable the realization of high-dimensional qudits, which naturally align with paraparticle frameworks that extend beyond conventional qubit systems. By encoding paraparticle states in superpositions of OAM modes, one can exploit the graded exchange symmetries intrinsic to paraparticles, leading to enriched entanglement structures and robust logical operations. This approach not only enhances information density but also offers novel ways to implement deterministic gates encoded into a single photon and implement parastatistics and other types of exotic quantum statistics in a photonic medium.

\subsection{Deterministic single-photon 2-qubit gates:}
This mathematical formalism permits the construction of deterministic two- (or more) qubit deterministic gates that can be used to implement a controlled-not (CNOT) gate, which is one of the milestones in quantum computing. Usually to have two qubits one has to double the qubit state in different ways (path and other methods) and this mechanism works in a stochastic way unless making two photons interacting e.g., in nonlinear media
The fundamental limit of linear optics, with linear transformations (beam-splitter, phases, polarizers) two photons do not interact: a unit SU(N) on modes cannot create entanglement in a deterministic way if the inputs contain at most one photon per mode. To break this constraint, indistinguishability and measure are exploited: projecting some modes and applying classic corrections introduces an effective interaction that can be moved offline.
An example is a gate via teleportation-based (KLM scheme). The Knill-Laflamme-Milburn protocol provides for offline preparation of a special entangled state (e.g., $|C_S\rangle$ of 4-8 photons). Then one performs a Bell-state measurement (BSM) between the data qubits and half of the resource applying optoelectronic feed-forward ($\pi$-phase or Pauli X) on the rest of the photons.
If the resource is ready, the CZ (or CNOT in dual-rail coding) always takes place on the logic qubits, the \textit{alea} is transferred in the preparation phase, which can be repeated until it is successful, thus making the gate deterministic at a logical level \cite{klm}.
Another way is the Quasi--deterministic gate with few ancilla by Ralph, Lund, Munro \cite{lrm}, which require a single ancilla photon plus photon-counting detectors (PNRD) and QNDs based on weak measurement. The single success is $P\approx0.25$, but by chaining $2$ attempts in parallel and using feed-forward one brings $P \rightarrow 1$ with linear overhead for each single gate. Alternatively, in Measurement-Based QC (cluster states) is needed to build a photonic state cluster offline using probabilistic fusion gates . Once the graph is available, each logical CZ is obtained with simple single qubit measurements and classical corrections and again operational determinism at the expense of multiple ancilla photons. Recent demonstrations integrate delay-loop and quantum-dot sources for linear clusters of 6 photons on chips \cite{mbc}. 

With linear optics alone, direct entangling is impossible, but it can be made deterministic at the logical level thanks to the methods already discussed whose performance today (2025) is limited mainly by losses and efficiency of the sources, not by switching times: with chips with a loss $<1$ dB and deterministic sources, the error threshold of the surface code is expected to be within demonstration in the next five years.

\subsection{Deterministic 1-photon 2-qubit deterministic gates in photonic platforms}

A proof-of-principle optical circuit implementing the grade-resolved spin–orbit conversion has been disclosed in \cite{tambu5,tambu6}, that exploits OAM and SAM states and, in certain circuits, also the property of SAM-OAM classical entanglement \cite{aiello} can at all effects build up not only the usual nonlinear-based or stochastic gates, but also deterministic quantum gates realized with a single photon.

Single-Photon Two-Qubit Gates (SPTQ) are realized when the two qubits are hosted by the same photon (e.g. $|H/V\rangle \otimes |\ell = ±1\rangle$).
No interaction between photons is needed: linear optics elements that couple the two degrees of freedom imply that a deterministic unitary gate (p = 1) is realized, an intra-photon CNOT (or CZ) that correlates polarization and OAM. Useful for preparing hyper-entangled states or cluster fusion.
Thus, unlike conventional linear-optics approaches -- where two-qubit gates are inherently probabilistic, demand bulky interferometers, and scale poorly because they need extra ancilla photons and active feed-forward -- a SAM$\otimes$OAM platform delivers deterministic, on-chip gates within a single photon, eliminating post-selection altogether. There is no intermediate measurement or probability of success: the gate is unit SU(4) performed with linear optics + controlled phase. Any leaks appear only as a channel error and not as gate failure.
 
The advantages are in having probability $P = 1$ independent of detection losses (no measure involved) and do not entangle distinct photons. For multi-qubit algorithms it is necessary to duplicate the degrees of freedom or use probabilistic fusion-gates or higher OAM modes and OAM can propagate in multimode waveguides.
 
Q-plates \cite{Marrucci2006} or other setups like \cite{tambu5,tambu6} modulate SAM and OAM that define the total angular momentum Noether invariant $J$. With OAM $\ell= 0, \pm 1$ and polarization $\sigma=\pm 1$ (related to SAM) one can and encode and manipulate a four qudit state, ququart, or two qubits on a single photon.
The co-propagation of spin and orbital modes in a sub-millimetre silicon trench guarantees phase stability that would otherwise require centimetre-scale balanced paths, while the same lithographic process furnishes mode-selective phase shifters and directional couplers with insertion losses below 1 dB. Encoding two logical qubits (or a ququart) in one particle doubles the computational density per photon and reduces circuit depth, allowing a surface-code--level operation count with at least an order-of-magnitude fewer resources than path- or time-bin-based schemes. 
In rectangular multimode waveguides, the coupling between spin angular momentum (SAM) and OAM enables the realization of versatile photonic quantum gates. Single-qubit gates such as rotations around the Bloch sphere can be implemented via mode-selective phase shifters that impart controlled phase differences between coupled SAM-OAM states, effectively realizing operations like $R_z(\theta)$ and $R_x(\theta)$. Beam splitter-like operations for OAM modes, achieved through specially designed directional couplers, allow for coherent superpositions necessary for Hadamard and generalized Fourier gates. Furthermore, the SAM-OAM coupling can be engineered to perform conditional logic by designing birefringent sections where the mode with $\ell=+1$ and left-circular polarization evolves differently than $\ell=-1$ and right-circular polarization, controlled-NOT (CNOT) and controlled-phase (CZ) gates can be implemented between OAM and SAM degrees of freedom, realizing deterministic single photon gates. 
In this architecture, the paraparticle-encoded qudits residing in different SAM-OAM sectors can be deterministically manipulated using integrated optical components, enabling robust and high-dimensional photonic quantum computation. In this case we encode ququarts in the $\mathbb{Z}_{2}\!\times\!\mathbb{Z}_{2}$-graded algebra. 

In rectangular multimode optical waveguides, the boundary conditions and symmetry-breaking effects can induce coupling between SAM and OAM modes that can result totally classically entangled and obey the properties of Majorana quanta.
Modes $\phi_{(\ell,\sigma)}(x,y)$ in a rectangular waveguide, where $\ell$ indexes the OAM mode and $\sigma = \pm 1$ denotes right/left circular polarization (SAM) or, equivalently, $A$ and $B$ waveguide modes for OAM $\ell = 0$  with their different polarization states if applicable. The field can be expressed in the following way,
\begin{equation}
\Psi(x, y, z, t) = \sum_{(\ell, \sigma)} \hat{\psi}_{(\ell,\sigma)}(z,t) \phi_{(\ell,\sigma)}(x,y) e^{i(\beta_{(\ell,\sigma)} z - \omega t)}
\end{equation}
where $\beta_{(\ell,\sigma)}$ is the propagation constant. The SAM-OAM coupling arises via birefringence or boundary effects and induces a hybrid mode basis:
\begin{equation}
\ket{\ell, \sigma} \rightarrow \sum_{(\ell', \sigma')} C^{(\ell', \sigma')}_{(\ell, \sigma)} \ket{\ell', \sigma'}
\end{equation}
Here we adopt for the sake of simplicity the convention $\sigma = +1$ for left-circular polarization (LCP) and $\sigma = - 1$ for right-circular polarization (RCP), consistent with the spin angular momentum sign.
We identify this hybridization with a grading index $(a,b) = (\ell \bmod 2, (1 - \sigma)/2)$.

For $\ell=0$ the two spin--orbit hybrid modes labelled by the helicity $\sigma=\pm1$ (right- and left-handed circular polarisations) possess grades $(a,b)=(0,0)$ and $(0,1)$, respectively: the orbital part is even
[$a=0$] while a flip of $\sigma$ toggles the $b$ component.  
A polarisation rotator that reverses the helicity - e.g.\ a half-wave plate (HWP) with its
fast axis at $45^\circ$ - therefore implements the map $(a,b)\colon(0,0)\;\longleftrightarrow\;(0,1)$, which is a pure \emph{$b$-grade} operation.  In the computational model that follows such grade-changing gates are permitted only (i) when they are explicitly declared as interfaces between the even- and odd-$b$
sectors, or (ii) when they are immediately followed by a compensating step that restores the original global $\mathbb{Z}_{2}\!\times\!\mathbb{Z}_{2}$ charge.

This maps structured light modes into $\mathbb{Z}_2 \times \mathbb{Z}_2$-graded sectors, enabling their interpretation as paraparticles. Operators $\hat{\psi}_{(\ell,\sigma)}^\pm$ obey paraparticle algebra, with $R$-matrix determined by the mode coupling matrix $C^{(\ell', \sigma')}_{(\ell, \sigma)}$. Here $\sigma$ is a binary mode index; it coincides with helicity only when the coupling vanishes.

\section{Photonic realization of the \texorpdfstring{$\boldsymbol{\mathbb Z_2\times\mathbb Z_2}$}{Z2×Z2} Green algebra}
\label{sec:photonicGreen}

\subsection{Guided--mode basis}
We extend the previous paraparticle results to photonic platforms. Consider as example a rectangular $\mathrm{SiN}$ waveguide engineered to support the four
single-photon modes
$
\ket{\ell,\sigma}\in
\bigl\{\ket{+1,L},\ket{-1,R},\ket{0,A},\ket{0,B}\bigr\},
$
where
(i) $\ell\in\{-1,0,+1\}$ is the orbital-angular-momentum index and
(ii) $\sigma\in\{L,R,A,B\}$ denotes either helicity ($L,R$) or, for the $\ell=0$
subspace, two orthogonal transverse profiles ($A,B$).
Introduce canonical bosonic operators
\begin{equation}
  \bigl[a_{\ell,\sigma},a^\dagger_{\ell',\sigma'}\bigr]
  =\delta_{\ell\ell'}\,
   \delta_{\sigma\sigma'},
  \qquad
  \bigl[a_{\ell,\sigma},a_{\ell',\sigma'}\bigr]=0.
  \label{eq:bosonCCR}
\end{equation}

\subsection{Grade map and projectors}

Define a $\mathbb Z_2\times\mathbb Z_2$ grade map
\begin{eqnarray}
&g:\;(\ell,\sigma)\longmapsto
  \bigl(a,b\bigr),\qquad
  a=\ell\bmod 2, \nonumber
  \\
&b=
  \begin{cases}
    0,&\sigma=L\ \text{or}\ A,\\[2pt]
    1,&\sigma=R\ \text{or}\ B,
  \end{cases}
  \label{eq:gradeMap}
\end{eqnarray}
so that the four physical modes populate the four group elements
$(a,b)\in\{(1,0),(1,1),(0,0),(0,1)\}$.

For each grade introduce the projector
$P_{(a,b)} = \sum_{\ell,\sigma\,:\,g(\ell,\sigma)=(a,b)} \ket{\ell,\sigma}\!\bra{\ell,\sigma}.$
The graded creation/annihilation operators are then
\begin{equation}
  \psi^{\dagger}_{(a,b)}:=
  \sum_{\ell,\sigma}P_{(a,b)}^{\ell,\sigma}\,a^\dagger_{\ell,\sigma},
  \qquad
  \psi_{(a,b)}:=
  \bigl(\psi^{\dagger}_{(a,b)}\bigr)^{\!\dagger}.
  \label{eq:gradedOps}
\end{equation}

\subsection{Graded commutators}

For two homogeneous operators $X$ and $Y$ of grades $g(X)=(a,b)$, $g(Y)=(a',b')$ define
\begin{eqnarray}
& [X,Y]_{\eta}
  :=XY-\eta^{\,g(X)\cdot g(Y)}\,YX, \quad \eta=-1, \nonumber
  \\
&g\!\cdot\!g'=aa'+bb'\pmod 2.
  \label{eq:gradedBracket}
\end{eqnarray}
Using the canonical relations~\ref{eq:bosonCCR} one finds the bilinear identities
\begin{equation}
\label{eq:gradedCCR}
[\psi_{g},\psi_{g'}]_{\eta}=0, \quad [\psi_{g},\psi^\dagger_{g'}]_{\eta}=\delta_{g,g'}\,\mathbb I,
\end{equation}
which reduce to bosonic (symmetric) or fermionic (antisymmetric) commutators
according as $g\!\cdot\!g'=0$ or $1$.

\subsection{Trilinear Green relations}

Green’s order-$p$ paraparticle algebra is characterized by the nested
commutator
\begin{equation}
  [\psi_{g},[\psi^\dagger_{g'},\psi_{g''}]_{\eta}]_{\eta}
  =\frac{2}{p}
     \bigl(
       \delta_{g,g'}\,\psi_{g''}
       -\delta_{g,g''}\,\psi_{g'}
     \bigr).
  \label{eq:Green}
\end{equation}
In the present waveguide each grade hosts a single physical mode,
so $p=1$ and the prefactor is exactly~$2$.%
\footnote{A waveguide supporting $p>1$ quasi-degenerate copies of each
grade—e.g.\ through a $p$-fold transverse-mode manifold—would obey the same
relation with the factor $2/p$.}

Field expansion is obtained by defining mode functions $\varphi_{(a,b)}(x)$ for each sector and the paraparticle field reads
\begin{equation}
  \Psi(x)=
  \sum_{(a,b)\in\mathbb Z_2^2}
  \psi_{(a,b)}\,\varphi_{(a,b)}(x),
  \label{eq:fieldExp}
\end{equation}
and satisfies the graded Majorana equation introduced in the next section.

Equations~\ref{eq:gradedOps}–\ref{eq:Green} demonstrate that structured-light
modes in a four-channel SAM–OAM waveguide realize a $\mathbb Z_2\times\mathbb Z_2$ Green paraparticle algebra of order $p=1$.
This closes the gap between the algebraic framework and the physical implementation.

Consider the Majorana Tower \cite{maj1932} with mass/spin relationship of Eq. \ref{Tower}. From that we write the unified operator field equation that describes the paraparticle-graded Majorana field equation mainly from Eq. \ref{majoranaeq}, \ref{gdiracgamma} and Eq. \ref{majo} that, for a generic operator becomes
\begin{equation}
\left[ \Gamma^\mu_{(a,b)} \partial_\mu + \frac{M}{\,s_{(a,b)}+\frac12\,} \right] \hat{\psi}_{(a,b)}(x) = 0.
\label{35}
\end{equation}
with $\Gamma^\mu_{(a,b)} \,\equiv\, P_{(a,b)}\otimes\gamma^\mu$, $P_{(a,b)}:=\lvert a,b\rangle\!\langle a,b\rvert$, and $\gamma^0=\sigma_z,$, $\gamma^1=i\sigma_y$, so that each graded sector $(a,b)$ carries its own projector $P_{(a,b)}$ while the second tensor factor is the usual $1\!+\!1$-dimensional Dirac matrix $(\gamma^0,\gamma^1)$.
With this definition $\Gamma^\mu_{(a,b)}$ acts like an ordinary gamma matrix on the spinor indices and as the identity inside the chosen $\mathbb Z_{2}\!\times\!\mathbb Z_{2}$ subspace, and  Eq.~\ref{35} remains well defined for every grade.

Paraparticle formulation in rectangular multimode waveguides with SAM-OAM coupling generate a qudit-based quantum calculations on a single photon encoded in the paraparticle $\mathbb{Z}_2 \times \mathbb{Z}_2$-graded algebras, in a novel way to encode qubits and qudits. An illustrative example with ququarts is reported in the appendix.

This results in classically entangled modes, where specific OAM values are restricted to certain polarization states. For instance, a typical configuration may allow the four states $\ell = +1$ with left circular polarization (LCP), $\ell = -1$ with right circular polarization (RCP) and $\ell = 0$ with both polarizations that correspond to two different waveguide modes $(A)$ and $(B)$ that can also be $\sigma=+1$ or $\sigma=-1$.
We then assign to each mode a $\mathbb{Z}_2 \times \mathbb{Z}_2$ grading based on OAM and SAM $a = \ell \bmod 2$ and $b = (1 - \sigma)/2$, 
where $\sigma = +1$ for LCP and $-1$ for RCP. The graded sectors are thus defined as reported in Tab. \ref{tab3},
\begin{table}[ht]
\begin{center}
\begin{tabular}{|c|c|c|c|}
\hline
Mode & OAM mode $\ell$ & $\sigma$ (SAM)& Grade $(a,b)$ \\
\hline
$\phi_1$ & $\ell = +1$ & $+1$ (LCP) & $(1,0)$ \\
$\phi_2$ & $\ell = -1$ & $-1$ (RCP) & $(1,1)$ \\
$\phi_3^+$ & $\ell = 0$ & $+1$ (A) & $(0,0)$ \\
$\phi_3^-$ & $\ell = 0$ & $-1$ (B) & $(0,1)$ \\
\hline
\end{tabular}
\end{center}
\caption{Definition of the grade sectors for OAM SAM Ququart of Eq. \ref{ququart}. \textbf{Grade-changing gates:} Any $\sigma$-flip on the $\ell=0$ pair $\{|0,A\rangle,(0,0);\ |0,B\rangle,(0,1)\}$ crosses the $b$-grading.  Throughout the ququart example we assume the compensated sequence above unless a gate is explicitly labelled “$\sigma$-flip (grade-changing). $\sigma=+1$ for LCP and $\sigma -1$ for RCP}
\label{tab3}
\end{table}
Alternatively to polarization, in this photonic platform, one can include two distinct quasi-degenerate $\ell = 0$ spatial waveguide modes, labeled for convenience $A$ and $B$, alongside and independently of the helicity index~$\sigma$ and the OAM modes $\ell = +1$ and $\ell = -1$. 
This construction yields a ququart encoding composed of $\{ (\ell = +1, A), (\ell = -1, A), (\ell = 0, A), (\ell = 0, B) \}$, which we map to successive spin sectors in the Majorana tower with corresponding graded labels $(a,b)$. The matrix $\theta_{ss'}$ then specifies the allowed couplings and selection rules between these OAM-waveguide pairs: nonvanishing entries enable controlled interactions or quantum gates, while vanishing entries enforce symmetry-protected decoupling. In particular, the presence of two $\ell = 0$ modes enables the simulation of higher-spin or ancillary sectors, offering additional flexibility for implementing paraparticle-inspired logic operations in photonic circuits.

The qudit encoding is obtained using graded paraparticle sectors: each mode $\phi_{(a,b)}$ corresponds to a paraparticle operator $\hat{\psi}_{(a,b)}$, satisfying $\mathbb{Z}_2 \times \mathbb{Z}_2$ graded algebra. A four-level qudit is formed by assigning logical basis states to distinct sectors
\begin{align}
    \label{ququart}
    &\ket{0} = \phi_{(0,0)} = (\ell=0, A)
    \\
    &\ket{1} = \phi_{(0,1)} = (\ell=0, B) \nonumber
    \\
    &\ket{+} = \phi_{(1,0)} = (\ell=+1, \sigma=+1) \nonumber
    \\
    &\ket{-} = \phi_{(1,1)} = (\ell=-1, \sigma=-1) \nonumber
\end{align}
Logical operations are defined using graded exchange rules between the different $\hat{\psi}_{(a,b)}$ 
\begin{eqnarray}
&&\hat{\psi}_{(a,b)}^+ \hat{\psi}_{(a',b')}^+ = 
\\
&& = (-1)^{(a,b)\cdot(a',b')} \hat{\psi}_{(a',b')}^+ \hat{\psi}_{(a,b)}^+ + \theta_{(a,b),(a',b')} \dots \nonumber
\end{eqnarray}

The Hamiltonian governing the system is built from graded number operators
\begin{equation}
\hat{H} = \sum_{(a,b)} \epsilon_{(a,b)} \hat{n}_{(a,b)}, \quad \hat{n}_{(a,b)} = \hat{\psi}_{(a,b)}^+ \hat{\psi}_{(a,b)}^-
\end{equation}
where $\epsilon_{(a,b)}$ is the mode energy or effective mass determined by the waveguide geometry and the full field is assembled from mode-weighted operators such as
\begin{equation}
\Psi(x) = \sum_{(a,b)} \hat{\psi}_{(a,b)}(x) \phi_{(a,b)}(x).
\end{equation}
This construction preserves both the relativistic structure and the statistical symmetries of the system.

\subsection{Quantum Computation with SAM-OAM Paraparticle Modes}
Paraparticles, which obey parastatistics beyond the standard Bose-Einstein and Fermi-Dirac statistics, offer a fundamentally different approach to encoding and manipulating quantum information. Here's a detailed breakdown of how they function and why they are potentially useful in quantum computing.

Their algebraic structure enables new qubits as they are characterized by nontrivial exchange rules governed by trilinear commutation relations or braided tensor symmetries. Each paraparticle is defined within a graded sector of the type $(a,b) \in \mathbb{Z}_2 \times \mathbb{Z}_2$, with creation/annihilation operators obeying the $R$ matrix 
\begin{equation}
\hat{\psi}^+_{(j,a)}\hat{\psi}^+_{(j,b)} = \sum_{(c,d)} R^{(c,d)}_{(a,b)}\hat{\psi}^+_{(j,c)}\hat{\psi}^+_{(j,d)}
\end{equation}
This structure allows for qudits that go beyond simple $0-1$ logic, enabling richer computational states and graded exclusion principles (e.g., only a fixed number of identical paraparticles allowed in a state like paraqubits). Here for the sake of simplicity we will focus more on qubit states.

Unlike fermions or bosons, paraparticles naturally avoid certain states due to their algebraic constraints. For instance parabosons of order $p$ allow up to $p$ identical particles in the same state. Parafermions can only occupy graded antisymmetric states.
This intrinsic structure acts like built-in error detection, as illegal states are algebraically forbidden. In principle, this allows quantum computers based on paraparticles to detect and correct errors without auxiliary encoding overhead, with a finite number of quantum states or through statistical methods implementing time series of data characterized by structured noise as in Nelson's quantum mechanics \cite{nelson1} adapted to the dynamics described by $R$ discussed below.

Paraparticles can be implemented in 3D systems using algebraic braiding (i.e., exchange rules defined by R-matrices rather than spatial motion). This makes them attractive for scalable quantum computing platforms as can be mapped to multi-qubit registers.
Operators are translated into combinations of Pauli matrices or fermionic modes.
Paraparticle oscillators and dynamics are then simulated using quantum hardware such as trapped ions, superconducting qubits and photonics, as a parafermion oscillator can be encoded using a pair of qubits, with specially designed gates that emulate the graded commutation relations.

\subsection{Application to Logical Gates}

Logical gate operations in SAM-OAM paraparticle systems can be implemented using
$\hat{n}_{(0,1)}$, conditional control qubit in CNOT gates,
$\hat{E}_{(0,0),(0,1)}$ as logical $X$ gate across SAM/waveguide mode states at $\ell = 0$, finally, $\hat{E}_{(1,1),(1,0)}$ to swap OAM states with preserved polarization.
These operator structures form the algebraic foundation for deterministic, symmetry-enforced quantum logic in structured photonic systems.

This architecture supports deterministic operation, since all logic is implemented via engineered mode coupling, grading symmetries, and classical interference and is based on qudits. It is scalable using integrated photonic circuits, and grading prevents logical errors via selection rules and parastatistical exclusion.
The SAM-OAM graded path-entangled system offers a robust and compact framework for realizing three-qubit Toffoli logic gates. The unique combination of angular momentum encoding, spatial mode entanglement, and graded symmetry allows for nonclassical quantum logic to be implemented deterministically in a photonic platform.

A quantum computer can be designed using the SAM-OAM entangled modes as qudits or logical qubits following certain simple rules. Logical states are encoded in $(a,b)$-graded sectors and in this configuration they correspond to the SAM-OAM modes of the Majorana quasiparticle with $\ell=\pm1$ and $\ell=0$.
Quantum gates are implemented via exchange operations, parity-sensitive phase shifters, and optical interference. Entanglement is achieved using beam splitters and mode-selective couplers that preserve grading and the readout is performed by spatial and polarization-resolved detectors.

The logical gate operations include:
\begin{align*}
    X &: \hat{\psi}_{(0,0)} \leftrightarrow \hat{\psi}_{(0,1)} \\
    Z &: \text{Phase shift } \exp(i\pi) \text{ on } (0,1) \\
    CZ &: \exp(i\pi \hat{n}_{(0,1)} \hat{n}_{(1,1)})
\end{align*}
This graded quantum architecture enables interesting robust and deterministic quantum logic gates with topological protection provided by parastatistical exclusion and symmetry constraints.

Deterministic CNOT Gate in SAM-OAM Graded Paraparticle Systems is the first step for an example of paraparticle formalism used in quantum circuits. Now we present a deterministic implementation of the controlled-NOT (CNOT) gate using the physical quantities spin-orbit (SAM-OAM) entangled photonic modes in rectangular multimode waveguides. 
This architecture exploits the $\mathbb{Z}_2 \times \mathbb{Z}_2$ graded symmetry of paraparticles, where each mode is uniquely identified by its orbital angular momentum (OAM), spin angular momentum (SAM), and possibly its spatial path.

Mode encoding and graded symmetry are so obtained. 
Each photonic mode is classified by its $(\ell, \sigma)$ pairing, with the grading defined by $a = \ell \mod 2$ and $b = (1 - \sigma)/2$, where $\sigma = +1$ for left-circular and $-1$ for right-circular polarization.
Logical qubits are then encoded as in Tab. \ref{tab4},
\begin{table}[ht]
\begin{center}
\begin{tabular}{|c|c|c|c|}
\hline
Qubit & Mode & Grading & Logical \\
 & & $\mathbb{Z}_2 \times \mathbb{Z}_2$ & States \\
\hline
Control & $\ell = 0,$ & $(0,b)$ & $\ket{0}_C = (0,0)$, \\
            & $\text{wg /} \sigma = \pm1$ & $(0,b)$ & $\ket{1}_C = (0,1)$ \\
Target & $\ell = \pm1,$ & $(1,a)$ & $\ket{0}_T = (1,1)$, \\
           & $\sigma = \pm1$ & $(1,b)$ & $\ket{1}_T = (1,0)$ \\
\hline
\end{tabular}
\end{center}
\caption{Mode encoding and graded symmetry. The label ``wg'' stands for general waveguide mode.}
\label{tab4}
\end{table}

The CNOT gate applies an $X$ operation (bit flip) on the target qubit conditional on the control being in the $\ket{1}_C$ state. This is implemented by using the graded occupation number $\hat{n}_{(0,1)}$ to detect control parity, then executing a conditional swap of the target between $(1,0)$ and $(1,1)$. One then has to enforcing exchange symmetry using the matrix $X_{ab,a'b'} = (-1)^{(a,b) \cdot (a',b')}$.

Deterministic operation and photonic design are achieved using mode-resolved integrated waveguides for routing and encoding, control of phase shifting and interferometric designs. In this way is ensured stable logic under grading.

The quantum circuit representation with the paraparticle $(a,b)$-graded sectors is shown in Fig. \ref{cnot} with a schematic of a deterministic CNOT gate using SAM-OAM graded photonic modes. 

\begin{figure}[h]
    \centering    \includegraphics[width=0.25\textwidth]{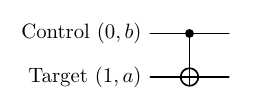}
\caption{Schematic of a deterministic CNOT gate}
\label{cnot}
\end{figure}

with control qubit $(0,b) \rightarrow (\ell=0,\ \sigma=\pm1 \lor A, B)$ and target qubit lies in sectors $(1,a) \rightarrow (\ell=\pm1,\ \sigma=\pm1)$

The control qubit is encoded in a $\ell=0$ mode with circular polarization. The target qubit, with $\ell=\pm1$, is routed through an interferometric path that swaps its logical state only when the control is in the fermionic sector $(0,1)$. The operation is deterministic via a $\pi$-phase shifter and interferometric routing.
The SAM-OAM graded framework provides deterministic logic gates by leveraging sector-specific occupation rules and topologically protected mode interactions. The CNOT gate demonstrates how single-photon logic can be both scalable and physically realizable using structured light and graded symmetry.

A further step is to realize with paraparticle formulation the Toffoli Gate. In this case we use the quantum computational degrees of  freedom of SAM-OAM paraparticles and path entanglement. Thus, we design a Toffoli (CCNOT) gate within a photonic platform based on structured modes in rectangular multimode waveguides, where spin angular momentum (SAM) and orbital angular momentum (OAM) are entangled and graded by a $\mathbb{Z}_2 \times \mathbb{Z}_2$ symmetry. The implementation leverages paraparticle grading and path entanglement to deterministically execute a three-qubit controlled logic operation.

The logical qubits are encoded in SAM-OAM-path modes as reported in the following table, Tab. \ref{tab5}. 
\begin{table}[ht]
\begin{center}
\begin{tabular}{|c|c|c|c|}
\hline
Qubit & Mode & Grading & Logical \\
 &  & $\mathbb{Z}_2 \times \mathbb{Z}_2$ & States \\
\hline
Qubit 1 & $\ell = 0,$ Path $A_p$ & $(0,b)$ & $\ket{0}_C = (0,0)$ \\
(Control) & $\text{wg /} \sigma = \pm1$ & $(0,b)$ & $\ket{1}_C = (0,1)$ \\
Qubit 2 & $\ell = \pm1,$ Path $A_p$ & $(1,a)$ & Two graded  \\
(Control) & $\sigma = \pm1$ &  & sectors \\
Qubit 3 & SAM-OAM & varies & Flip if both \\
(Target) & in Path $A_p$ / $B_p$ &  & controls = 1 \\
\hline
\end{tabular}
\end{center}
\caption{Logical qubits encoded in SAM-OAM-path modes; ``wg'' stands for general waveguide mode.}
\label{tab5}
\end{table}
Each mode combination $(\ell, \sigma)$ is mapped to a sector in the $\mathbb{Z}_2 \times \mathbb{Z}_2$ grading, and logical operations depend on exchange and interference rules between these sectors.

The Toffoli gate executes a controlled-controlled-NOT, flipping the target qubit if and only if both control qubits are in the logical 1 state:
\begin{align*}
\ket{c_1 = 0, c_2 = * , t} &\rightarrow \ket{c_1 = 0, c_2 = *, t} \\
\ket{c_1 = 1, c_2 = 1, t} &\rightarrow \ket{c_1 = 1, c_2 = 1, t \oplus 1}
\end{align*}
where the symbol ``*'' means ``any value'' that is, either $0$ or $1$.
Path-Entangled Realization is included to realize the Toffoli CCNOT gate.
The target qubit is encoded in spatial path modes ($A_p$/$B_p$). Conditional operations on the control qubits (graded sectors) induce a $\pi$ phase shift or exchange, flipping the path of the target photon. This can be realized by a combination of OAM-polarization mode converters and phase shifters. To implement the other part of the circuit one inserts interferometric elements (e.g., Mach-Zehnder) with graded switches and in this way are defined graded exchange operators enforcing $(-1)^{(a,b) \cdot (a',b')}$ behavior.

The Toffoli Gate with path entanglement and graded SAM-OAM modes is here schematized
\begin{figure}[h]
\includegraphics[width=0.4\textwidth]{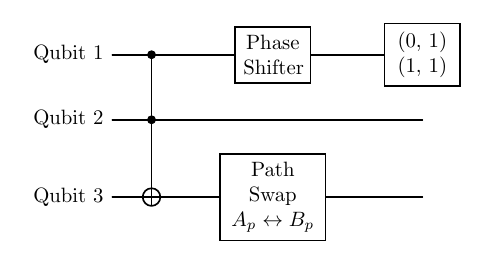}
\caption{Toffoli Gate with path entanglement and graded SAM-OAM modes}
\label{toffoli}
\end{figure}

Two control qubits are encoded in polarization-OAM states in Path $A_p$. The target qubit is routed through interferometric paths. If both control inputs are in the logical 1 sector, a $\pi$ phase shift and destructive interference switch the target from Path A to Path B, implementing a NOT operation.
Qubit 1: $(0,b_1)$, Qubit 2: $(1,a_2)$, Qubit 3: Path $A_p$/$B_p$.
\textbf{Notes:} Path-mode of the target encodes logical state. Only sectors $(0,1)$ and $(1,1)$ trigger interference path-swap. Conditional $\pi$-phase shift and graded parity ensure deterministic action. Truth tables are reported in the Appendix.

\subsection{Operator Structure of the \boldmath$\mathbb{Z}_2 \times \mathbb{Z}_2$ Graded Superalgebra}

In this section, we present explicit operator constructions for the $\mathbb{Z}_2 \times \mathbb{Z}_2$ graded superalgebra, using paraparticle operators associated with SAM-OAM modes in rectangular waveguides.

Grading and operator assignment are made by labeling OAM and SAM as follows.
Each SAM-OAM mode is labeled by orbital angular momentum $\ell$ and polarization $\sigma = \pm1$ or general waveguide modes (A and B). The grading, as is known, is defined by $a = \ell \bmod 2$, and $b = (1 - \sigma)/2$ and gives rise to the following mode-operator mapping reported in Tab \ref{tab6}.

\begin{table}[ht]
\begin{center}
\begin{tabular}{|c|c|c|c|c|}
\hline
Mode & $\ell$ & $\sigma$ & Grade $(a,b)$ & Operator \\
\hline
$\phi_1$ & $0A$ & $+1 \lor A$ & $(0,0)$ & $\hat{\psi}_{(0,0)}^\pm$ \\
$\phi_2$ & $0$ & $-1 \lor B$ & $(0,1)$ & $\hat{\psi}_{(0,1)}^\pm$ \\
$\phi_3$ & $+1$ & $+1$ & $(1,0)$ & $\hat{\psi}_{(1,0)}^\pm$ \\
$\phi_4$ & $-1$ & $-1$ & $(1,1)$ & $\hat{\psi}_{(1,1)}^\pm$ \\
\hline
\end{tabular}
\caption{Grading, operator assignment and mode-operator mapping.}
\label{tab6}

\end{center}
\end{table}
The graded commutation relations are obtained using the scalar product $(a,b) \cdot (a',b') = aa' + bb'$ mod 2, the graded commutator for homogeneous operators is
\begin{equation}
\hat{\psi}_{(a,b)}^\pm \hat{\psi}_{(a',b')}^\pm = (-1)^{(a,b)\cdot(a',b')} \hat{\psi}_{(a',b')}^\pm \hat{\psi}_{(a,b)}^\pm
\end{equation}
The graded commutation relations are enlisted in the appendix.
Bilinear and Hamiltonian operators are then defined in terms of number and exchange operators,
\begin{eqnarray}
&&\hat{n}_{(a,b)} = \hat{\psi}_{(a,b)}^+ \hat{\psi}_{(a,b)}^- 
\\
&&\hat{E}_{(a,b),(a',b')} = \hat{\psi}_{(a,b)}^+ \hat{\psi}_{(a',b')}^-
\end{eqnarray}
These obey the graded Lie algebra
\begin{eqnarray}
&&[\hat{E}_{(a,b),(a',b')}, \hat{E}_{(c,d),(c',d')}] = 
\\
&&= \delta_{(a',b'),(c,d)} \hat{E}_{(a,b),(c',d')} - \delta_{(a,b),(c',d')} \hat{E}_{(c,d),(a',b')} \nonumber
\end{eqnarray}
This structure is a non-trivial, reductive algebra that mimics the Lie algebra $\mathfrak{gl}_N$, 
which consists of all $N \times N$ complex matrices with the commutator as the Lie bracket. This algebraic equivalence occurs within each graded sector and controls sector-preserving logical gates and exchange transformations.
Thus, $\mathfrak{gl}_N$ in second quantization e.g., in paraparticle systems like \cite{wang} arises as the algebra of bilinear operators $E_{i,j} = \sum_a \hat{\psi}_{(i,a)}^+ \hat{\psi}_{(a,b)}^- $  with dimension $N^2$ with the commutation relations $[\hat{E}_{i,j}, \hat{E}_{k,l}] = \delta_{jk}\hat{E}_{i,l} - \delta_{il}\hat{E}_{k,j}$ which matches the structure constants of $\mathfrak{gl}_N$. This algebra, also written as $\mathfrak{gl} (N, \mathbb{C}) = \mathfrak{gl}_\mathbb{N}(\mathbb{C})$ describes internal symmetry transformations (e.g., flavor symmetry, mode mixing), and serves as a closure condition to ensure that paraparticle field theories are self-consistent and covariant. It contains the simple traceless subalgebra $\mathfrak{sl}_N \subset \mathfrak{gl}_N$.

An example of a model Hamiltonian for OAM--SAM coupled ququarts is here introduced to describe the dynamics of this system in rectangular waveguides. The ququart basis is so defined $|0\rangle \equiv |L,0\rangle$, $|1\rangle \equiv |R,0\rangle$, $|2\rangle \equiv |L,+1\rangle$,  $|3\rangle \equiv |R,+1\rangle$, where $L/R$ denote left- and right-circular polarization and $\ell = 0, +1$ is the OAM quantum number.

The effective Hamiltonian is of the type
\begin{equation}
\hat{H} = \sum_{i=0}^{3} \omega_i \hat{n}_i + \sum_{i\neq j} J_{ij} \hat{a}_i^\dagger \hat{a}_j + \sum_{i,j} \Delta_{ij} \hat{a}_i \hat{a}_j + \mathrm{h.c.},
\end{equation}
where $\omega_i$ are the mode frequencies, $\hat{a}_i^\dagger$ and $\hat{a}_i$ are the creation and annihilation operators for each hybrid mode, $J_{ij}$ are coherent hopping amplitudes between modes (due to spin--orbit coupling or induced perturbations), and $\Delta_{ij}$ represent parametric couplings or nonlinear interactions.

Explicitly, for the cyclic $\mathbb{Z}_4$ parafermion model, we can set:
\begin{equation}
\hat{H} = \sum_{k=0}^{3} \omega \hat{n}_k + g \sum_{k=0}^{3} (\hat{a}_k^\dagger \hat{a}_{k+1} + \mathrm{h.c.}),
\end{equation}
with cyclic boundary conditions $\hat{a}_{k+4} \equiv \hat{a}_k$, where $g$ is the coupling strength implementing the $\tau$ shift operator.

In the context of rectangular waveguides, the OAM--SAM hybrid modes experience birefringent and geometric phase shifts that can be described by:
\begin{equation}
\hat{H}_\sigma = \sum_{k=0}^{3} \phi_k \hat{n}_k,
\end{equation}
where $\phi_k$ introduces the phase pattern associated with the clock operator $\sigma_c$:
\begin{equation}
\phi_0 = 0, \quad \phi_1 = \frac{\pi}{2}, \quad \phi_2 = \pi, \quad \phi_3 = \frac{3\pi}{2}.
\end{equation}

This Hamiltonian captures the essential features of parafermionic dynamics in the photonic platform and provides a basis for quantitative predictions and experimental design.


Although our work extensively discusses the algebraic structures underlying the $Z_N$ parafermionic systems and their potential realization in photonic platforms, an explicit formulation of the model Hamiltonian or tight-binding lattice description has not been fully presented. To address this, we now outline a representative Hamiltonian that captures the essential features of the coupled orbital angular momentum (OAM)--spin angular momentum (SAM) ququart system in rectangular waveguides.

The tight-binding Hamiltonian can be written as:
\begin{equation}
H = \sum_{j} \left[ \omega_j a_j^\dagger a_j + \sum_{\langle j,k \rangle} J_{jk} a_j^\dagger a_k + \sum_{\langle j,k \rangle} \lambda_{jk} a_j^\dagger \sigma_z a_k \right],
\end{equation}
where $a_j^\dagger$ and $a_j$ are the creation and annihilation operators at site $j$, $\omega_j$ denotes the on-site energy, $J_{jk}$ the nearest-neighbor coupling, and $\lambda_{jk}$ encodes the OAM--SAM spin-orbit coupling with $\sigma_z$ acting on the polarization subspace (e.g., left- and right-circular polarizations). 

For the parafermionic sector, we can introduce generalized $\mathbb{Z}_N$ operators $\alpha_j$ obeying:
\begin{equation}
\alpha_j^N = 1, \quad \alpha_j \alpha_k = e^{2\pi i / N} \alpha_k \alpha_j \quad (j < k),
\end{equation}
which decompose under the $\mathbb{Z}_2 \times \mathbb{Z}_2$ grading as:
\begin{equation}
\alpha_j = \tau_j \sigma_j,
\end{equation}
where $\tau_j$ and $\sigma_j$ represent commuting $\mathbb{Z}_2$-graded operators. A worked detailed matrix example for $N=4$ is included in Appendix A.

\subsection{Photonic Realization of Para-Majorana Modes}

An explicit proposal for realizing para-Majorana modes in photonic systems involves a one-dimensional coupled waveguide lattice or an array of coupled optical resonators. Each lattice site can support photonic modes with synthetic degrees of freedom, such as polarization, orbital angular momentum, or frequency bins, which encode the para-particle degrees of freedom. We consider an effective Hamiltonian of the form
\begin{equation}
\begin{split}
&H = \sum_j \left[ t \, a_j^\dagger a_{j+1} + t^* \, a_{j+1}^\dagger a_j + \right.
\\
&\left.+ \Delta \, (a_j a_{j+1} + a_{j+1}^\dagger a_j^\dagger) + U (a_j^\dagger a_j)^n \right],
\end{split}
\end{equation}
where $a_j^\dagger, a_j$ are photonic creation and annihilation operators at site $j$, $t$ is the nearest-neighbor hopping amplitude, $\Delta$ is a parametric pairing term induced via four-wave mixing or external drive, and $U (a_j^\dagger a_j)^n$ is a nonlinear interaction term, such as Kerr-type or cross-phase modulation, with $n \geq 2$. This nonlinear term is essential to simulate para-statistics beyond standard quadratic Majorana models.

The system can be engineered to exhibit a synthetic $\mathbb{Z}_N$ symmetry, where $N$ determines the order of the para-particles. Properly designed coupling phases and drive configurations can create synthetic gauge fields, emulating the symmetry algebra required for para-Majorana excitations. The edge sites of the lattice can host localized para-Majorana zero modes, analogous to the zero-energy modes in topological superconducting wires.

Experimental observables in such systems may include the detection of zero-energy edge states via transmission or reflection spectroscopy, observation of fractionalized photonic excitations through interference experiments, and potential non-Abelian braiding signatures via dynamic tuning of coupling parameters. These proposals establish a promising pathway toward realizing para-Majorana physics in a highly controllable photonic platform.

We note that the $\mathbb{Z}_N$-graded parafermionic algebra can, under suitable conditions, be reduced to a $\mathbb{Z}_2 \times \mathbb{Z}_2$-graded superalgebra structure. This is achieved by decomposing the parafermionic generators into pairs of operators $(\chi_i, \eta_i)$, each graded under an independent $\mathbb{Z}_2$ symmetry. The combined grading is then given by the bidegree $(p_i, q_i)$, where $p_i, q_i \in \mathbb{Z}_2$, and the commutation relations generalize to $X_i X_j = (-1)^{p_i p_j + q_i q_j} X_j X_i$.
This decomposition preserves the essential algebraic features while enabling the construction of a superalgebraic framework, which may be particularly useful for connecting para-Majorana models to supersymmetric or superconformal extensions.

As an explicit example, consider the $\mathbb{Z}_4$ parafermion algebra generated by the clock and shift matrices:
\begin{equation}
\sigma_c = 
\begin{pmatrix}
1 & 0 & 0 & 0 \\
0 & i & 0 & 0 \\
0 & 0 & -1 & 0 \\
0 & 0 & 0 & -i
\end{pmatrix}, \quad
\tau = 
\begin{pmatrix}
0 & 1 & 0 & 0 \\
0 & 0 & 1 & 0 \\
0 & 0 & 0 & 1 \\
1 & 0 & 0 & 0
\end{pmatrix},
\end{equation}
satisfying $\tau \sigma = i \sigma \tau$. These generators can be decomposed into a $\mathbb{Z}_2 \times \mathbb{Z}_2$-graded algebra by mapping $\sigma$ and $\tau$ onto Pauli matrices acting on two qubits or on a ququart from the product between two qubits, $\sigma \to \sigma_z \otimes \mathbb{I}, \quad \tau \to \mathbb{I} \otimes \sigma_x$, where $\sigma_z$ and $\sigma_x$ are Pauli operators. This decomposition splits the parafermionic generator into $\psi_j \sim \chi_j \eta_j$, where $\chi_j$ and $\eta_j$ are operators graded under the two independent $\mathbb{Z}_2$ symmetries. Such a decomposition is permitted for $N=4$, but not for prime $N$, due to the absence of a direct product decomposition of the underlying cyclic group.

Ququarts provide a natural physical platform for realizing the $\mathbb{Z}_4$ parafermion algebra. The generalized clock and shift operators acting on a ququart match the $\mathbb{Z}_4$ relations $\sigma_c^4 = \tau^4 = 1$, $\sigma_c \tau = i \tau \sigma$,
where $\sigma_c$ and $\tau$ are the generalized Pauli operators. Furthermore, the ququart Hilbert space $\mathcal{H}_4$ can be factorized as $\mathcal{H}_2 \otimes \mathcal{H}_2$, allowing the $\mathbb{Z}_4$ grading to be decomposed into a $\mathbb{Z}_2 \times \mathbb{Z}_2$-graded superalgebra. This provides flexibility in experimental implementations, either by directly manipulating ququart degrees of freedom or by encoding them in two entangled qubits.

An example of experimental realization of parafermionic ququarts using left- and right-circular polarization (LCP, RCP) combined with orbital angular momentum (OAM) modes can be simply summarized as follows. The ququart basis can be defined as $\{|L, 0\rangle, |R, 0\rangle, |L, +1\rangle, |R, +1\rangle\}$, where $|L\rangle$ and $|R\rangle$ denote circular polarization states and $\ell$ is the OAM quantum number, or even $\{|L,+1\rangle, |L, -1\rangle, |R,+1\rangle, |R, -1\rangle \}$, in a number of six independent states if SAM and OAM are decoupled. The clock operator $\sigma$ applies combined polarization- and OAM-dependent phase shifts, while the shift operator $\tau$ cycles states across the four-dimensional basis. 
Shift operator $\tau$ is defined by the action $\tau |0\rangle = |1\rangle$, $\tau |1\rangle = |2\rangle$, $\tau |2\rangle = |3\rangle$, and $\tau |3\rangle = |0\rangle$.

The operators satisfy $\sigma_c^4 = \tau^4 = \mathbb{I}$, where
\begin{equation}
\sigma_c^4 = \mathrm{diag}(1,1,1,1), \quad 
\tau^4 = \begin{pmatrix}
1 & 0 & 0 & 0 \\
0 & 1 & 0 & 0 \\
0 & 0 & 1 & 0 \\
0 & 0 & 0 & 1
\end{pmatrix}.
\end{equation}
The key commutation relation is $\tau \sigma_c = i \, \sigma \tau$.
Explicit calculation yields
\begin{equation}
\tau \sigma_c =
\begin{pmatrix}
0 & i & 0 & 0 \\
0 & 0 & -1 & 0 \\
0 & 0 & 0 & -i \\
1 & 0 & 0 & 0
\end{pmatrix}, \quad
\sigma_c \tau =
\begin{pmatrix}
0 & 1 & 0 & 0 \\
0 & 0 & i & 0 \\
0 & 0 & 0 & -1 \\
-i & 0 & 0 & 0
\end{pmatrix},
\end{equation}
which satisfy $\tau \sigma_c = i \sigma_c \tau$.

From the $\mathbb{Z}_4$ clock and shift operators, we now decompose the four-dimensional space as a tensor product of two qubits, $\mathcal{H}_4 \cong \mathcal{H}_2 \otimes \mathcal{H}_2$,
with basis: $|0\rangle \equiv |0\rangle_A |0\rangle_B$, $|1\rangle \equiv |0\rangle_A |1\rangle_B$, $|2\rangle \equiv |1\rangle_A |0\rangle_B$, $|3\rangle \equiv |1\rangle_A |1\rangle_B$.

In this basis, define two independent $\mathbb{Z}_2$ grading operators $P = \sigma_z \otimes \mathbb{I}$, $Q = \mathbb{I} \otimes \sigma_z$, where $\sigma_z$ is the Pauli-$z$ matrix $\sigma_z$.
Explicit matrix forms are
\begin{equation}
P = 
\begin{pmatrix}
1 & 0 & 0 & 0 \\
0 & 1 & 0 & 0 \\
0 & 0 & -1 & 0 \\
0 & 0 & 0 & -1
\end{pmatrix}, \quad 
Q = 
\begin{pmatrix}
1 & 0 & 0 & 0 \\
0 & -1 & 0 & 0 \\
0 & 0 & 1 & 0 \\
0 & 0 & 0 & -1
\end{pmatrix}.
\end{equation}
The combined grading labels are $|0\rangle \to (0,0)$, $|1\rangle \to (0,1)$, $|2\rangle \to (1,0)$, $|3\rangle \to (1,1)$.
We can now rewrite $\sigma$ and $\tau$ as $\sigma = P Q$, $\tau = $(cyclic permutation operator), where the cyclic $\tau$ couples the two $\mathbb{Z}_2$ sectors.
This decomposition shows explicitly how the $\mathbb{Z}_4$ algebra factors into two commuting $\mathbb{Z}_2$ grading sectors, providing a concrete realization of the $\mathbb{Z}_4 \to \mathbb{Z}_2 \times \mathbb{Z}_2$ reduction.

An optional effective Hamiltonian for the dynamics can be also
\begin{equation}
H = \omega \sum_{k=0}^3 \sigma_c^k + g \sum_{k=0}^3 \tau^k,
\end{equation}
where $\omega$ and $g$ are experimentally tunable parameters.

These operations can be implemented in free space using wave plates, liquid crystal phase shifters, $q$-plates, and spatial light modulators, SLMs. Experimental observables include tests of the $\mathbb{Z}_4$ algebra, cyclic state evolution, and coherence between polarization and OAM degrees of freedom, offering a promising platform for photonic parafermionic simulations.

For $\ell=0$ modes a lone QWP toggles the grade between $(0,0)$ and $(0,1)$. Our default gate set therefore treats QWPs as legitimate grade-changing single-qubit rotations, or replaces an isolated QWP by the compensated sequence $\mathrm{QWP}(45^\circ)\ \rightarrow\ \mathrm{HWP}(0^\circ)\ 
  \rightarrow\ \mathrm{QWP}(45^\circ)$, which realises the same physical polarisation flip and returns the state to its original grade, thereby preserving sector-restricted logic when required.

Rectangular waveguides, instead, provide another alternative platform for realizing OAM-SAM coupled ququarts. In this setting, the hybridization of spin angular momentum (SAM, associated to circular polarization) and orbital angular momentum (OAM) arises due to the broken cylindrical symmetry of the rectangular cross-section, leading to spin--orbit coupled modes. The ququart basis $\{|L,0\rangle, |R,0\rangle, |L,+1\rangle, |R,+1\rangle\}$ maps directly onto the parafermionic algebra, where the clock operator $\sigma$ corresponds to applying OAM- and SAM-dependent phase shifts, achievable via birefringence or geometric phase elements, and the shift operator $\tau$ represents cyclic transitions between the hybrid modes, implemented through controlled spin--orbit coupling perturbations such as asymmetric tapers, strain, or nanostructured gratings. This architecture enables the exploration of $\mathbb{Z}_4$ algebraic relations and topological dynamics in an integrated photonic platform, offering a promising route toward parafermionic quantum simulation and potentially topological quantum computation.

Relevant theoretical and experimental works supporting this approach include the implementation of spin--orbit coupled ququarts in rectangular waveguides builds upon key experimental and theoretical work on spin--orbit interactions in nanophotonic systems \cite{Petersen2014, Bliokh2015, LeFeber2015, Sollner2015}, 
as well as foundational research on topological photonics and parafermionic physics \cite{Hafezi2013, Barkeshli2013, Lu2014, Fendley2012}. 
Notably, the integration of SAM-OAM coupling in guided-wave platforms and the demonstration of topological edge states in photonic lattices provide critical tools for realizing $\mathbb{Z}_4$ parafermionic models in photonic systems.

\subsection{Deterministic gate set with the
\texorpdfstring{$\mathbb Z_{2}\!\times\!\mathbb Z_{2}$}{Z2×Z2}-graded ququart}
\label{subsec:deterministicGates}

Unlike standard linear-optical processors that need two-photon interference and post-selection, our
$\mathbb{Z}_{2}\!\times\!\mathbb{Z}_{2}$-graded ququart platform implements a universal gate set deterministically with passive, loss-free spin--orbit devices, thereby removing the exponential heralding overhead.  Encoding two logical qubits inside a \emph{single} photon also doubles the Hilbert-space dimension per particle, giving higher information density and a smaller component count while the built-in grade symmetry provides intrinsic error-suppression.

In this case we can build two internal qubits inside one photon. The four logical basis states defined in Eq.~\ref{ququart} may be viewed as \emph{two} qubits carried by the same photon, $Q_a := a=\ell\bmod2$ and $Q_b := b= (1-\sigma)/2$, where $\ell\in\{-1,0,+1\}$ is the orbital angular momentum (OAM) and $\sigma=\pm1$ the helicity.  Table~\ref{tab:internal-qubits} lists the mapping together with off-the-shelf optical elements that act deterministically on each qubit.

\begin{table}[h]
\centering
\begin{tabular}{|c|c|c|l|}
\hline
grade & physical & logical & loss-free  \\
bit & DOF & states & devices \\ \hline
$Q_a$  & OAM & $\{|0\rangle\!=\!a=0$,& Dove prism \\
  & parity &$|1\rangle\!=\!a=1\}$ & in Sagnac, \\
& & &  $q=\tfrac12$ q-plate \\
$Q_b$  & helicity $\sigma$ & $\{|0\rangle\!=\!\sigma=+1,$ &HWP stack\\
 & & $\;|1\rangle\!=\!\sigma=-1\}$ &QWP--HWP--QWP  \\ \hline
\end{tabular}
\caption{Internal qubit decomposition of the photonic ququart and
deterministic one-photon devices that address them.}
\label{tab:internal-qubits}
\end{table}

\begin{table}[b]                     
\caption{Grade sectors for the spin--orbit ququart.  
Left (right) circular polarisation corresponds to $\sigma=+1$ ($-1$). The grade bits are $a=\ell\bmod2$ and $b=(1-\sigma)/2$.}
\label{tab:gradeSectors}
\begin{ruledtabular}                
\begin{tabular}{cccc}
Mode & OAM $\ell$ & SAM $\sigma$ & Grade $(a,b)$ \\ \hline
$\phi_{1}$        & $+1$ & $+1$ (LCP) & $(1,0)$ \\
$\phi_{2}$        & $-1$ & $-1$ (RCP) & $(1,1)$ \\
$\phi_{3}^{A}$    & $0$  & $+1$ (LCP) & $(0,0)$ \\
$\phi_{3}^{B}$    & $0$  & $-1$ (RCP) & $(0,1)$ \\
\end{tabular}
\end{ruledtabular}
\end{table}

The primitive single-qubit rotations can be defined in the following way,
\begin{enumerate}[label=\arabic*. , leftmargin=*,
                  itemsep=0pt,        
                  topsep=2pt]
  \item \textbf{On $Q_b$ (SAM):} a half-wave plate (HWP) set at
        $45^{\circ}$ realises the Pauli $X_b$, while a
        QWP--HWP--QWP sequence spans arbitrary $\mathrm{SU}(2)_b$
        rotations.
  \item \textbf{On $Q_a$ (OAM parity):} a $90^{\circ}$-rotated Dove
        prism inside a polarisation-insensitive Sagnac interferometer
        applies $Z_a$; a $q=1/2$ q-plate tuned to
        $\delta=\pi$ swaps $\ell=0\leftrightarrow\pm1$ and realises
        $X_a$.
\end{enumerate}
All these operations are passive; no ancilla photons or post-selection are required. Deterministic two-qubit entangling gate is so realized on the degrees of freedom of a single photon. Set a $q=1$ q-plate tuned to retardation $\delta=\pi$ couples the two degrees of freedom as \cite{Marrucci2006,Nagali2009},
\begin{equation}
\mathrm{QP}_{\pi}:
|a\rangle_{Q_a}\,|b\rangle_{Q_b}\longmapsto
|a\oplus b\rangle_{Q_a}\,|b\rangle_{Q_b},
\end{equation}
which is a controlled-NOT with $Q_b$ the control and $Q_a$ the target. Prepending or postpending a HWP swaps the control and target roles.

An alternative all-glass implementation employs a polarization Sagnac interferometer with an embedded $45^{\circ}$ Dove prism, demonstrated by Fiorentino \& Wong \cite{Fiorentino2004} as a fully deterministic CNOT between polarisation and momentum of a single photon.
The gate set
\begin{equation}
\mathcal G_{\text{det}}= \bigl\{\,X_b,\;H_b,\;Z_a,\;H_a,\;
\textbf{CNOT}_{b\rightarrow a}\bigr\}
\label{75}
\end{equation}
(where $H$ denotes the Hadamard rotation on the relevant qubit) is universal for $\mathrm{SU}(4)$ on the ququart and contains only deterministic, single-photon devices.

Grade-conservation rule: operations that flip exactly one grade bit, e.g.\ $X_b$, $X_a$ or the above CNOT, change the $\mathbb Z_{2}\!\times\!\mathbb Z_{2}$ charge. Our computational model therefore requires either explicit declaration of such gates as ``grade interfaces'' or a compensating inverse later in the circuit so that the net global grade is conserved.  Because every gate in $\mathcal G_{\text{det}}$ is loss-free, inserting such compensators does not reduce determinism or fidelity.

Deterministic inter-photon entangling gates can be realised by combining the single-photon CNOT above with a polarization Sagnac that couples one photon's $Q_b$ to the path qubit of a second photon, or a semiconductor spin--photon interface providing a giant optical non-linearity.  Both preserve the global $\mathbb Z_{2}\!\times\!\mathbb Z_{2}$ charge and extend the present scheme toward scalable photonic paraparticle logic.

\subsection{Universality of the deterministic gate set \texorpdfstring{$G_{\mathrm{det}}$}{Gdet}}
\label{sec:universality}
The logical ququart $\{|L,0\rangle,\,|R,0\rangle,\,|L,+1\rangle,\,|R,+1\rangle\}$ is isomorphic to two qubits
$Q_a\otimes Q_b$ with the identification $|\ell\bmod 2\rangle_{Q_a}\otimes|(1-\sigma)/2\rangle_{Q_b}$  
(cf. Table VII).  
In this basis the deterministic gate set introduced in Eq.~\ref{75} reads (Subscripts indicate the qubit on
which the single-qubit operator acts.)
\[
G_{\mathrm{det}}
=\{X_{b},\,H_{b}\}\;\cup\;\{Z_{a},\,H_{a}\}\;\cup\;\{\mathrm{CNOT}_{b\to a}\}.
\tag{U1}
\]

\paragraph{Step 1: local controllability.}
The pairs $\{H,Z\}$ and $\{H,X\}$ each generate $\mathrm{SU}(2)$ because $ZH= e^{\tfrac{\pi}{2}i(\sigma_x+\sigma_z)/\sqrt2}$ and together with $H$ they form a basis of the Lie algebra $\mathfrak{su}(2)$.
Hence $\{Z_{a},H_{a}\}$ gives arbitrary single-qubit rotations on $Q_a$, and $\{X_{b},H_{b}\}$ does the same on $Q_b$.

\paragraph{Step 2: entangling power of $\mathrm{CNOT}_{b\to a}$.}
Write a two-qubit gate in its Cartan (KAK) normal form
\[
U=(k_1\!\otimes\!k_2)\, e^{i\left(c_1\sigma_x\otimes\sigma_x+ c_2\sigma_y\otimes\sigma_y+ c_3\sigma_z\otimes\sigma_z\right)}\, (k_3\!\otimes\!k_4).
\]
The triplet $(c_1,c_2,c_3)\in[0,\pi/2]^3$ lives in the Weyl chamber; if any component is non-zero the gate is entangling. $\mathrm{CNOT}_{b\to a}$ is locally equivalent to $\exp\!\bigl(\tfrac{i\pi}{2}\,\sigma_z\!\otimes\!\sigma_x\bigr)$; its Cartan coordinates are $(\frac{\pi}{2},0,0)$, so it lies on the chamber edge of perfect entanglers.

\paragraph{Step 3: Brylinski–Brylinski criterion.}
A gate set is universal for $\mathrm{SU}(4)$ iff it contains a dense subgroup of local operations
$\mathrm{SU}(2)\otimes\mathrm{SU}(2)$ and at least one entangling two-qubit gate. These conditions are satisfied by steps 1 and 2, therefore $G_{\mathrm{det}}$ is universal.

\paragraph{Constructive three-\textsc{cnot} decomposition.}
For completeness we give an explicit synthesis of an arbitrary $U\in\mathrm{SU}(4)$ in terms of elements
of $G_{\mathrm{det}}$.  Following the algorithm of Vatan--Williams \cite{vatan}, one writes
\begin{equation}
\begin{split}
&U = (k_1\!\otimes\!k_2)\; \mathrm{CNOT}_{b\to a}\; (k_3\!\otimes\!k_4)\; \circ
\\
&\circ~ \mathrm{CNOT}_{b\to a}\; (k_5\!\otimes\!k_6)\; \mathrm{CNOT}_{b \to a}\;
(k_7\!\otimes\!k_8),
\end{split}
\end{equation}
where each $k_j\in\mathrm{SU}(2)$ can be compiled with at most
three gates from $\{H,Z\}$ (on $Q_a$) or $\{H,X\}$ (on $Q_b$)
using Euler angles.  The total depth is therefore $3$ CNOT $+ 24$ single-qubit rotations,
matching the optimal gate counts known from the standard
two-qubit model.

A remark on grade conservation: as $X_b$ flips $b$ while preserving $a$, and $Z_a$ flips the phase of $a$ while preserving $b$, every gate changes either zero or exactly one $\mathbb{Z}_2\times\mathbb{Z}_2$ charge bit.  Circuits that must return to a fixed global grade can therefore insert a compensating gate without affecting universality or determinism.


\section{Error Correction with Paraparticles and Fractional Nelson's Quantum Mechanics}
Error correction can be realized in a way that preserves all the set of qubits, through the repetition of the calculation and statistical analysis of the results obtained so far.
Fractional Nelson's quantum mechanics extends the standard stochastic-quantization picture by letting the underlying particle paths follow Lévy flights rather than Brownian motion, so the dynamical fractional L\'evy order $\alpha$ becomes a tunable parameter that continuously sweeps through bosonic, fermionic and intermediate paraparticle statistics. Embedding our spin-orbit photonic ququarts (or higher-dimensional qudits) in such a fractional medium therefore gives a hardware-level control parameter for both the graded spectrum and the braid phase, enabling deterministic gate synthesis while naturally modeling the power-law, non-Markovian noise that dominates integrated-photonics platforms.

The adapting of Nelson's Quantum Mechanics to Paraparticles is obtained through generalized statistics extending the usual boson/fermion dichotomy, with trilinear commutation or anticommutation relations rather than bilinear ones.
To integrate paraparticles into Nelson's formalism, the standard Wiener process must generalize to accommodate nontrivial statistical exchanges, potentially via algebraic deformations of Brownian motions governed by R-matrix quantization and graded Lie algebras.

Nelson's stochastic formulation of quantum mechanics provides an elegant interpretation of quantum dynamics through underlying stochastic processes, described mathematically by forward and backward stochastic differential equations (SDEs) \cite{nelson2}. The standard form of these equations is given by $dX(t) = b_+(X(t),t)\,dt + dW_+(t)$ and backward in time, $dX(t) = b_-(X(t),t)\,dt + dW_-(t)$, where $W_\pm(t)$ represent standard Wiener processes, and the quantum potential emerges naturally through the difference of these drift terms. They represent standard Brownian motions. 
Replacing Brownian trajectories in Nelson's stochastic quantisation by L\'evy flights produces a fractional Schr\"odinger (or Majorana) equation whose order $\alpha\!\in\!(1,2]$ continuously
interpolates between bosonic, fermionic and paraparticle exchange statistics.\cite{Laskin2000,Tarasov2006}.

In our graded model the same index~$\alpha$ can be mapped onto the
paraparticle order $p$ (and hence onto the $\mathbb Z_{2}\!\times\!\mathbb Z_{2}$
grade), providing a \emph{single tunable parameter} that controls both the spectrum
(Eq.~\ref{Tower}) and the braided $R$-matrix.

From an experimental viewpoint, fractional dynamics emerges naturally in photonics whenever long-range evanescent coupling or non-local metasurfaces are engineered: recent optical platforms already realise the fractional Schr\"odinger equation with unit efficiency and no post-selection \cite{Longhi2015}.
Embedding our graded ququart in such a medium therefore lets one dial the effective order $\alpha$--hence the paraparticle
statistics--in situ.  This adds a hardware knob for deterministic gate synthesis and offers built-in modelling of power-law, non-Markovian noise, which is the dominant decoherence channel in heterogeneous
integrated photonics.

In short, the fractional Nelson formalism supplies both a seamless mathematical bridge between different paraparticle orders and an experimentally accessible control parameter that can be exploited for robust, deterministic gate design in spin--orbit photonic circuits.

In the fractional Schr\"odinger/Majorana equation the Laplacian is replaced by its Riesz--fractional power, with the fractional order symbol $\alpha$
\begin{equation}
  i\hbar\partial_t\psi =
  \frac{\hbar^{\alpha}}{2m}\,
 D_\alpha \bigl(-\Delta\bigr)^{\alpha/2}\psi
  + V\,\psi,
  \qquad 1<\alpha\le 2,
  \label{eq:fracSE}
\end{equation}
so the real number $\alpha$ is the \emph{order} of the derivative and $D_\alpha$ the related scale factor. 
For $\alpha=2$ Eq.~\ref{eq:fracSE} reduces to the ordinary Schr\"odinger equation; lowering $\alpha$ from $2$ toward $1$ stretches the derivative into a non-local operator with power-law tails \cite{Laskin2000,Tarasov2006}.

In Nelson's stochastic interpretation the same $\alpha$ is the stability  index of the underlying Lévy flights: a Brownian path ($\alpha=2$) acquires heavy-tailed step lengths $P(\ell)\propto\ell^{-(1+\alpha)}$ when $1<\alpha<2$.  Within our graded model we identify $\alpha$ with the paraparticle order $p$, turning it 
into a \emph{continuous hardware knob} that sweeps smoothly from bosonic statistics ($\alpha=2$) through parafermionic regimes down toward the fermionic limit ($\alpha\!\to\!1$).  Photonic platforms that engineer long-range evanescent coupling already realise Eq.~\ref{eq:fracSE} optically and allow in-situ tuning of $\alpha$, providing deterministic, grade-preserving control over both the spectrum 
(Eq.~\ref{Tower}) and the braid phase \cite{Longhi2015}.

Besides the deterministic potential $V$, the quantum potential $Q(x,t)$ emerges naturally as
\begin{equation}
Q(x,t) = \frac{\hbar^2}{2m}\frac{\Delta\sqrt{\rho(x,t)}}{\sqrt{\rho(x,t)}}.
\end{equation}
Then one defines the graded probability and measure. 
The probability measures must now be generalized to graded probability measures. Expectation values involve graded trace operations of the type $\langle O \rangle = \mathrm{Tr}_{grad} \left[\rho ~O\right]$ of a graded observable $O$ with graded density matrix $\rho$. 
The quantum potential $Q(x,t)$ must adapt to paraparticle symmetry, now arising from graded algebraic potentials.

In graded quantum systems, the definition of expectation value $\langle O \rangle$ of an observable $O$ in a density matrix $\rho$ depends on the graded trace, which is defined as
\begin{equation}
\mathrm{Tr}_{\mathrm{grad}}[X] = \sum_{(a,b)} (-1)^{\epsilon(a,b)} \mathrm{Tr}[X_{(a,b)}],
\end{equation}
with $\epsilon(a,b)$ assigning the appropriate sign factor over sectors. This construction ensures that both the probability measures and the quantum potential
$Q(x,t)$ correctly incorporate the paraparticle symmetry and graded algebraic structure.

Fractional Brownian motion (fBm) introduces long-range temporal correlations characterized by the Hurst parameter $H = \alpha /2 \in (0,1)$. Its covariance structure is defined by
\begin{equation}
\mathbb{E}[B_H(t)B_H(s)] = \frac{1}{2}\left(t^{2H}+s^{2H}-|t-s|^{2H}\right).
\label{FBincrement}
\end{equation}
Replacing standard Brownian increments with fractional increments, the fractional dynamics becomes
\begin{equation}
dX_H(t) = b_H(X_H(t),t)\,dt + dB_H(t).
\end{equation}

The fractional Schr\"odinger equation associated with these dynamics is then expressed as follows
\begin{equation}
i\hbar\,\partial_t \psi(x,t) = -D(-\Delta)^H \psi(x,t) + V(x)\psi(x,t),
\end{equation}
where the fractional Laplacian $(- \Delta)^H$ is defined via Fourier transform
\begin{equation}
(-\Delta)^H f(x) = \mathcal{F}^{-1}\{|k|^{2H}\mathcal{F}\{f\}(k)\}(x).
\end{equation}
Consequently, the fractional quantum potential is expressed in terms of the Hurst's exponent.
\begin{equation}
Q_H(x,t) = -\frac{\hbar^2}{2m}\frac{(-\Delta)^H\sqrt{\rho(x,t)}}{\sqrt{\rho(x,t)}} ,
\label{fpot}
\end{equation}
to which can be added also a deterministic function $V$.

\subsection{Fractional-Graded Stochastic Dynamics and Calculus}

Combining fractional Brownian motion with paraparticle algebra yields fractional-graded stochastic differential equations (FG-SDE)
\begin{equation}
d\hat{X}_H(t) = \hat{b}_H(\hat{X}_H(t),t)\,dt + d\hat{B}_H(t),
\label{73}
\end{equation}
where the increments $ d\hat{B}_H(t) $ obey graded fractional relations of the following type, that are quadratic variation of graded fBm.
\begin{equation}
\mathrm dB^{(H)}_{(a)}(t)\, \mathrm dB^{(H)}_{(b)}(t') \;=\; \delta_{ab}\; \gamma_H(t-t')\, \mathrm dt\,\mathrm dt'.
\label{74}
\end{equation}
with $\gamma_H(\tau):= H\!\left(2H-1\right)$ denotes the fractional correlation structure, 
and $|\tau|^{\,2H-2}$, $\tau:=t-t'$.

The properties of $\gamma_H(\tau)$ are \emph{symmetry}, $\gamma_H(\tau)=\gamma_H(-\tau)$ that ensures time-reversal invariance of the noise kernel. Then follows the \emph{Brownian limit} $H\to\tfrac12$ for which $\gamma_H(\tau)\to\delta(\tau)$, recovering the local Itô rule.

\emph{Long memory} is achieved for $H\!>\!\tfrac12$,  $\gamma_H(\tau)\propto|\tau|^{2H-2}$ is integrable but non-local, producing the $1/f^{\,2H-1}$ spectral density characteristic of fractional Brownian dynamics.

The integration of fractional graded stochastic processes is defined via generalized It\^o--Skorokhod integrals \cite{ito}
\begin{equation}
\int_0^T \hat{F}(t)\,d\hat{B}_H(t) := \lim_{|\Pi|\rightarrow 0}\sum_{j}\hat{F}(t_j^\ast)(\hat{B}_H(t_{j+1})-\hat{B}_H(t_j)),
\end{equation}
with algebraic definitions ensuring graded commutations.

The quantum potential in this fractional-graded framework of Eq.~\ref{fpot} is
defined via fractional-graded functional calculus, maintaining paraparticle algebra consistency.
Fractional Brownian increments $d B_H(t)$ employed in equations Eq.~\ref{73} and \ref{74} are defined through the Itô--Skorokhod integral framework. Specifically, the increments satisfy Eq.~\ref{FBincrement} involving the Hurst parameter $H$. All integrals involving fractional increments are assumed to satisfy standard conditions of integrability and adaptedness required by fractional stochastic calculus.

Explicit fractional-graded Nelson equations are formulated by defining the forward and backward equations,
\begin{align}
d\hat{X}_H^+(t) &= \hat{b}_H^+(\hat{X}_H^+(t),t)\,dt + d\hat{B}_H^+(t), \nonumber
\\ 
d\hat{X}_H^-(t) &= \hat{b}_H^-(\hat{X}_H^-(t),t)\,dt + d\hat{B}_H^-(t). \nonumber
\end{align}
These equations lead to a fractional and graded Schr\"odinger-type equation respecting paraparticle statistics.
This combined fractional-graded approach introduces non-Markovian, algebraically rich quantum dynamics, providing a theoretical foundation for experimentally realizable quantum statistical phenomena in structured photonic systems and other condensed matter analogs.

\subsection{Structured Photonic Modes and Nelson's Formalism}

We extend Nelson's fractional-graded quantum mechanics to structured photonic systems, specifically focusing on optical quantum circuits utilizing orbital angular momentum (OAM). Structured photonic modes naturally provide a robust experimental platform for realizing the theoretical constructs introduced above.

Consider optical modes characterized by their spin angular momentum (SAM) and orbital angular momentum (OAM). These structured photonic states can be described mathematically as
\begin{equation}
\Psi_{\ell,\sigma}(\mathbf{r},t) = \psi_{\ell,\sigma}(z,t)\phi_{\ell,\sigma}(x,y)e^{i(\beta_{\ell,\sigma}z - \omega t)},
\end{equation}
where $\ell$ denotes OAM, $\sigma = \pm 1$ indicates polarization states or other independent waveguide modes, and $\beta_{\ell,\sigma}$ is the propagation constant.

Fractional-graded quantum dynamics of OAM modes is given by Nelson's fractional-graded stochastic framework translates into the structured photonic context by treating each OAM mode as Majorana pseudoparticle corresponding to a paraparticle state with fractional-graded stochastic dynamics,
\begin{equation}
d\hat{\psi}_{\ell,H}(z,t) = \hat{b}_{\ell,H}(z,t)dt + d\hat{B}_{\ell,H}(z,t)
\end{equation}
where $d\hat{B}_{\ell,H}(z,t)$ is a graded fractional increment satisfying the fractional-graded algebra defined earlier.

\subsection{Quantum Potential and Mode Coupling}

The fractional graded quantum potential in the context of structured photonic OAM modes becomes
\begin{equation}
\hat{Q}_{\ell,H}(z,t) = -\frac{\hbar^2}{2m}\frac{(-\partial_z^2)^H_{\mathrm{graded}}\sqrt{\hat{\rho}_{\ell}(z,t)}}{\sqrt{\hat{\rho}_{\ell}(z,t)}}.
\end{equation}
This potential governs the nonlinear coupling and interactions among various OAM modes.

Quantum Gates with Fractional-Graded OAM Modes are implemented using and implementing quantum logic gates within optical waveguides by exploiting the fractional-graded dynamics of structured modes. A deterministic CNOT gate, for instance, can be realized by using graded paraparticle statistics encoded in SAM-OAM coupled modes, leading to stable and robust quantum logic operations of the type
\begin{equation}
|\ell,\sigma\rangle \rightarrow \sum_{(\ell',\sigma')}C_{(\ell',\sigma')}^{(\ell,\sigma)}|\ell',\sigma'\rangle,
\end{equation}
with the coefficients $C_{(\ell',\sigma')}^{(\ell,\sigma)}$ governed by fractional-graded quantum dynamics.

Structured photonic systems, such as multimode integrated waveguides and fiber-based circuits, represent viable experimental platforms for these fractional-graded quantum dynamics. The inherent robustness of structured light modes, combined with fractional quantum dynamics, offers a powerful avenue for scalable and error-resilient quantum computing architectures.

Future research could focus on explicit experimental validations, numerical simulations of fractional-graded dynamics in integrated photonic circuits, and further theoretical developments toward realizing robust quantum information processing and quantum communication protocols using structured optical modes.
Nelson's fractional-graded formalism adapted to structured photonic circuits can indeed be employed to test components of photonic circuits and implement quantum error correction.

\subsection{Fractional-Graded Diagnostic and Correction in Photonic Quantum Circuits}

Photonic quantum circuits, like those discussed in the patents, represent a promising platform for scalable and robust quantum information processing. Here we introduce a novel general framework based on fractional-graded algebras to perform diagnostic and error correction tasks, leveraging the interplay between graded symmetries and fractional dynamics to enhance fault tolerance in paraparticle-encoded photonic systems.

Let $\Psi_{\ell,\sigma}(x, t)$ be the field representing a structured light mode in a quantum photonic circuit, where $\ell \in \mathbb{Z}$ denotes the orbital angular momentum (OAM) and $\sigma = \pm 1$ the spin angular momentum (SAM). Each mode is embedded in a $\mathbb{Z}_2 \times \mathbb{Z}_2$ graded algebraic structure labeled by $(a,b) = (\ell \bmod 2, \frac{1 - \sigma}{2})$.

Each mode evolves according to a fractional stochastic differential equation (fSDE) with graded paraparticle statistics:
\begin{equation}
d\hat{X}_{\ell,\sigma}(t) = \hat{b}_{\ell,\sigma}(t)\,dt + d\hat{B}^{(H)}_{\ell,\sigma}(t),
\end{equation}
where $\hat{b}_{\ell,\sigma}(t)$ is a mode-dependent drift operator and $d\hat{B}^{(H)}_{\ell,\sigma}(t)$ is a graded fractional Brownian increment with Hurst index $H \in (0,1)$ satisfying
\begin{eqnarray}
&&d\hat{B}^{(H)}_{a}(t)\,d\hat{B}^{(H)}_{b}(t') + (-1)^{(a,b)} d\hat{B}^{(H)}_{b}(t')\,d\hat{B}^{(H)}_{a}(t) = \nonumber
  \\
&&= 2\delta_{ab} \gamma_H(t - t')\,dt\,dt'.
  \end{eqnarray}

\subsection{Quantum Field and Density Evolution}

Let $\rho_{\ell,\sigma}(x,t) = |\Psi_{\ell,\sigma}(x,t)|^2$ be the probability density of the structured mode. Then the total density evolves as $\partial_t \rho + \nabla \cdot (\rho v) = 0$, with $v = \frac{1}{2} (b_+ + b_-)$.
The effective quantum potential that governs diagnostic behavior is given by Eq.~\ref{qpot}.
\begin{equation}
Q^{(H)}_{\ell,\sigma}(x,t) = -\frac{\hbar^2}{2m} \frac{(-\Delta)^H \sqrt{\rho_{\ell,\sigma}(x,t)}}{\sqrt{\rho_{\ell,\sigma}(x,t)}},
\label{qpot}
\end{equation}
where $(-\Delta)^H$ denotes the fractional Laplacian acting in space.
Diagnostic Signature and Correction Condition is then expanded as follows:
from Eq. \ref{qpot} the deviation $\Delta Q^{(H)}$ between theoretical and measured quantum potentials defines a diagnostic operator as in \cite{tambu2p}, 
$\Delta Q^{(H)}_{\ell,\sigma} = Q^{(H)}_{\ell,\sigma,\text{meas}} - Q^{(H)}_{\ell,\sigma,\text{expected}}$ and the correction term $\hat{U}_{\text{corr}}$ is triggered when the deviation exceeds a noise threshold $\epsilon$, $\| \Delta Q^{(H)}_{\ell,\sigma} \| > \epsilon$.
In this case one defines a conditional control unitary $\hat{U}_{\text{corr}}$ acting on sector $(a,b)$ such that $\hat{U}_{\text{corr}} = \exp\left( -i \theta_{(a,b)} \hat{n}_{(a,b)} \right)$, and $\hat{n}_{(a,b)} = \hat{\psi}^{\dagger}_{(a,b)} \hat{\psi}_{(a,b)}$, 
where $\theta_{(a,b)}$ is dynamically adjusted based on $\Delta Q^{(H)}_{(a,b)}(t)$.

The complete system evolution within the circuit then becomes
\begin{equation}
\hat{H}_{\text{total}} = \hat{H}_{\text{OAM}} + \hat{H}_{\text{SAM}} + \hat{Q}^{(H)} + \hat{U}_{\text{corr}}(\Delta Q),
\end{equation}
governing how diagnostic deviations influence active feedback correction during real-time photonic quantum computation or with a posteriori selection from multiple runs.

The central insight of this work is that paraparticle algebras, which have traditionally been studied as abstract mathematical structures, can find natural realizations within structured light systems. Specifically, spin-orbit coupled photonic modes, such as those carrying both spin angular momentum (SAM) and orbital angular momentum (OAM), inherently possess the degrees of freedom and exchange statistics required to realize $ Z_2 \times Z_2 $-graded algebraic structures.

The photonic modes considered here are not merely carriers of quantum information; rather, they act as effective realizations of the algebraic elements themselves. For example, the OAM degree of freedom, accessible through Laguerre--Gaussian or Bessel beams, can encode parafermionic excitations, while the coupling between SAM and OAM mediates the graded commutation relations central to the algebra. This mapping obviates the need for condensed matter or atomic systems, offering a purely photonic, room-temperature platform for realizing exotic statistics.

From a dynamical perspective, the use of fractional Brownian motion in the stochastic framework introduces a natural description of environmental noise and decoherence. Unlike conventional Markovian approaches, this formalism accounts for memory effects, which are known to play a significant role in realistic optical systems. Nelson's stochastic mechanics, originally formulated for nonrelativistic quantum systems, provides a conceptual analogue for describing emergent interactions and nonlocal effects in the photonic context and even the presence of a coupled deterministic mechanism described by the Wold theorem. The classical Wold decomposition (for stationary processes) is so defined. 
For any zero‐mean, covariance‐stationary stochastic process $\{X_t\}_{t\in\mathbb Z}$ with finite variance, Wold's theorem states\cite{Wold1938,BrockwellDavis} then 
\begin{equation}
\begin{split}
&  X_t
  \;=\;
  \underbrace{\sum_{j=0}^{\infty}\psi_j\,\varepsilon_{t-j}}_{\text{(i) stochastic MA$(\infty)$}}
  \;+\;
  \underbrace{\eta_t}_{\text{(ii) deterministic}},
\\
&    \sum_{j=0}^{\infty}\psi_j^2<\infty,
\end{split}   
\label{eq:Wold}
\end{equation}
where $\{\varepsilon_t\}$ is a white‐noise sequence ($\mathbb E\,\varepsilon_t\!=\!0,\;
  \operatorname{var}\varepsilon_t=\sigma_\varepsilon^2$), $\psi_0=1$, and $\eta_t$ is perfectly predictable from its own past.
Equation~\ref{eq:Wold} underpins the entire ARMA/ARIMA modeling framework, truncating the infinite moving‐average yields a finite‐order ARMA approximation whose parameters can be estimated consistently by maximum likelihood or Whittle's method.

From the other side, the fractional extension (including long‐memory processes and fBm) start with the definition of fractional Brownian motion $B_H(t)$, $(0<H<1)$, which is non‐stationary, but its
increments $X_k = B_H(k+1) - B_H(k)$ form a stationary, long‐memory process called fractional Gaussian
noise (FGN). Granger and Joyeux \cite{GrangerJoyeux1980} and Hosking \cite{Hosking1981} showed that FGN admits a Wold‐type representation with slowly decaying moving‐average coefficients,
\begin{equation}
  X_k
  \;=\;
  \sum_{j=0}^{\infty}\psi_j^{(H)}\,\varepsilon_{k-j},
  \qquad
  \psi_j^{(H)}
  \sim j^{H-\frac32},
  \label{eq:fgnWold}
\end{equation}
or, equivalently, an ARFIMA$(0,d,0)$ representation with $d=H-\frac12$.  This fractional moving‐average expansion supports both maximum‐likelihood - Whittle estimation of the Hurst exponent $H$ \cite{Beran1994}, and efficient FGN/fBm simulation via finite‐order ARMA truncations or FFT‐based convolution kernels \cite{PipirasTaqqu2003,Moulines2007}.

A continuous Wold--Volterra analogue also exists, $f\!B\!m$ that can be written as a causal Volterra integral of standard Brownian motion $W(t)$, namely, $B_H(t)=\int_{-\infty}^{t}K_H(t,s)\,dW(s)$, with kernel given by $K_H(t,s)=\frac{1}{\Gamma(H+\frac12)}\bigl[(t-s)^{H-\frac12}-(-s)^{H-\frac12}\bigr]$.
This kernel form is the starting point for fractional Kalman filtering, rough‐path analysis and optical analogues of fractional quantum mechanics.

In summary, Eq.~\ref{eq:Wold} and its fractional counterpart Eq.~\ref{eq:fgnWold} provide a unified moving‐average language that covers both short‐ and long‐memory regimes, furnishing the statistical
and computational backbone for applications ranging from ARMA control loops to L\'evy‐flight optical lattices.

Taken together, these theoretical insights suggest that a tabletop experimental setup involving spatial light modulators, $q$-plates, beam splitters, and single-photon detectors could serve as a testbed for exploring the phenomenology of graded paraparticles. Importantly, the proposed architecture points toward scalable photonic quantum computing schemes that transcend the limitations of qubit-based architectures by leveraging qudits with rich internal structure. While challenges remain--including mode purity, optical loss, and error correction protocols--the theoretical foundation developed here provides a promising roadmap for experimental exploration at the intersection of algebraic quantum theory and modern photonic technology.

\section{Conclusions}
We have developed a comprehensive theoretical framework embedding Majorana's infinite-component relativistic wave equations into the algebraic formalism of paraparticles, explicitly employing $\mathbb{Z}_2 \times \mathbb{Z}_2$-graded Lie algebras and advanced $R$-matrix quantization techniques. Our approach systematically associates spin-dependent mass spectra with well-defined graded sectors characterized by generalized quantum statistics, resulting in a unified relativistic quantum field theory rigorously consistent with Majorana's original mass-spin relation. This unification effectively bridges disparate areas within theoretical physics suggesting possible future paraparticle-based quantum computers beyond qubits, using graded qudits and potential advances in quantum error correction by exploiting graded symmetries and mathematical tools of fractal stochastic processes related to Nelson's quantum mechanics.

Furthermore, we extended our formalism to practical applications like ideal platforms recalling those we proposed as an integrated variant based on Pancharatnam–Berry metasurfaces is covered by the recent patent filings Refs.~\cite{tambu5,tambu6} by us illustrating how structured photonic systems, particularly spin-orbit coupled modes in optical waveguides, exhibit classical entanglement that mirrors the algebraic structures of paraparticles. These findings open novel and experimentally viable pathways to simulate exotic quantum statistics utilizing structured light. Our detailed exploration highlights profound interrelations among orbital angular momentum, spin angular momentum, and parastatistics, offering impactful insights for fundamental research and practical quantum information applications.

Future investigations prompted by our framework may focus on detailed experimental realizations in integrated photonic circuits, extensive numerical simulations to rigorously validate theoretical predictions, and further theoretical expansions to include interacting systems and gauge symmetries. The algebraic infrastructure developed herein lays a strong foundation for advancing quantum computational architectures and innovative quantum communication technologies, thus fostering significant interdisciplinary collaborations and technological breakthroughs. Moreover, starting from the work by Wang and Hazzard \cite{wang}, the digital simulation architectures here proposed could serve as a benchmarking tool for the algebraic dynamics of paraparticle modes in photonic platforms, enabling a complementary validation pathway between discrete quantum simulations and continuous structured light implementations.

Other results presented here establish a rigorous theoretical framework for the realization of parafermionic excitations in OAM--SAM coupled systems, unifying the algebraic structure of $\mathbb{Z}_N$ parafermions with physically implementable models. The mathematical analysis demonstrates the compatibility of graded algebraic symmetries with photonic architectures, providing a foundation for further theoretical developments and opening new directions in the exploration of non-Abelian quasiparticles.


\paragraph{Competing interest.}
FT and RS are co-inventor on patent applications WO2024062338A1 and WO2025052199A1, which describe photonic circuits capable of realizing the paraparticle algebra discussed in this work.  No other authors have competing interests to declare.

%
%
%
\section*{Appendix}

\appendix

\section{Appendix A: Mathematical tools}

This appendix provides a brief guide to the mathematical structures underlying the main results of this work, including Hopf algebras, graded algebras, braided tensor categories, and quantum groups. For readers unfamiliar with these topics, we recommend several accessible references: the foundational treatments by Majid~\cite{Majid1995}, Kassel~\cite{Kassel1995}, and Klimyk and Schmüdgen~\cite{Klimyk1997}; the categorical perspective in Etingof et al.~\cite{Etingof2015}; the theory of Lie superalgebras in Scheunert~\cite{scheunert}; and the more physics-oriented introductions in Majid's \textit{Quantum Groups Primer}~\cite{Majid1999} and Pressley--Segal's \textit{Loop Groups}~\cite{Pressley1986}. Together, these works provide a comprehensive starting point for understanding the algebraic techniques used in quantum field theory, statistical physics, and quantum information.

\section{Braided Coproduct and Braided Tensor Product}

A braided tensor product is a generalization of the ordinary tensor product that includes braiding, a structured rule for swapping or exchanging elements from two algebraic sectors.
A braided coproduct is a mathematical operation that generalizes the usual notion of a coproduct in algebra used to define how algebraic structures behave under tensor products but now includes a braiding (or twist) that captures nontrivial exchange symmetries.
In the context of graded or quantum algebras, especially those involving paraparticles, anyons, or quantum groups, the braided coproduct ensures that interchanging two particles or operators does not simply commute or anticommute, but follows more complex rules governed by a braiding matrix (often called an R-matrix).
In this appendix, we provide a detailed exposition of the concepts of \textbf{braided coproduct} and \textbf{braided tensor product} as applied to the $\mathbb{Z}_2 \times \mathbb{Z}_2$-graded Hopf algebraic framework used to describe paraparticles and Majorana fields.

Standard vs. Braided Coproducts are obtained from Hopf algebras. Let $\mathcal{H}$ be a Hopf algebra. A standard coproduct is a map:
\begin{equation}
\Delta: \mathcal{H} \rightarrow \mathcal{H} \otimes \mathcal{H},
\end{equation}
which satisfies the coassociativity condition $(\Delta \otimes \mathrm{id}) \circ \Delta = (\mathrm{id} \otimes \Delta) \circ \Delta$.
In symmetric (bosonic) systems, the coproduct of an operator $\psi$ typically takes the form $\Delta(\psi) = \psi \otimes \mathbb{I} + \mathbb{I} \otimes \psi$.

However, in paraparticle or anyonic systems, particle exchange follows a \textit{nontrivial} rule. The \textbf{braided coproduct} introduces a deformation via an $R$-matrix
\begin{equation}
\Delta_B(\psi_{(a,b)}) = \psi_{(a,b)} \otimes \mathbb{I} + \sum_{(c,d)} R^{(c,d)}_{(a,b)} \, (\mathbb{I} \otimes \psi_{(c,d)}),
\end{equation}
where $R^{(c,d)}_{(a,b)}$ encodes the braiding symmetry and grading structure.

\subsection{The Braiding Matrix and Exchange and Yang--Baxter Equation}
The $R$-matrix components are chosen to respect the graded symmetry. For $\mathbb{Z}_2 \times \mathbb{Z}_2$ grading, the basic exchange rule between two elements $\psi_{(a,b)}$ and $\psi_{(a',b')}$ is defined by $\psi_{(a,b)} \otimes \psi_{(a',b')} \mapsto (-1)^{(a,b) \cdot (a',b')} \, \psi_{(a',b')} \otimes \psi_{(a,b)}$, where the scalar product is $(a,b) \cdot (a',b') = a a' + b b' \mod 2$.
This rule is implemented by the braiding operator $\mathcal{R}$, $\mathcal{R}(\psi_{(a,b)} \otimes \psi_{(a',b')}) = R^{(a',b')}_{(a,b)} \, (\psi_{(a',b')} \otimes \psi_{(a,b)})$.
In deformed or quantum settings, the braiding coefficient may contain a deformation parameter $q \in \mathbb{C}$ with $|q| = 1$, leading to $R^{(a',b')}_{(a,b)} = (-1)^{(a,b)\cdot(a',b')} + \theta_{(a,b),(a',b')}$,
with $\theta_{(a,b),(a',b')} = \epsilon_{(a,b),(a',b')} (q^{s + s'} - 1)$, for associated spins $s$ and $s'$ and sign $\epsilon$ determined by the parity.

\textbf{ The coherence via the Yang--Baxter Equation}: to ensure consistency and associativity in multi-particle states, the braiding map must satisfy the Yang--Baxter equation $\mathcal{R}_{12} \mathcal{R}_{13} \mathcal{R}_{23} = \mathcal{R}_{23} \mathcal{R}_{13} \mathcal{R}_{12}$.
Here, $\mathcal{R}_{ij}$ acts on the $i$-th and $j$-th components of a triple tensor product space $\mathcal{H} \otimes \mathcal{H} \otimes \mathcal{H}$.
This ensures that the total braiding is associative and well-defined regardless of the order of exchanges and  that multi-photon entanglement remains consistent under sector exchange, a critical requirement for physical implementation in waveguide-based photonic circuits.

\subsection{Definition of Braided Tensor Product}
The braided tensor product of two algebra elements $A \in \mathfrak{g}_{(a,b)}$ and $B \in \mathfrak{g}_{(a',b')}$ is defined by $A \otimes B \mapsto (-1)^{(a,b) \cdot (a',b')} B \otimes A$, which is a generalization of the symmetric ($+1$) or antisymmetric ($-1$) rules.
This braiding is reflected in $A \otimes B = \mathcal{R}(B \otimes A)$, with $\mathcal{R}$ being the braiding operator defined by the $R$-matrix $\mathcal{R}(B \otimes A) = \sum R^{(c,d)}_{(a,b)} \, (B_{(c,d)} \otimes A_{(a,b)})$.

\textbf{Example in $\mathbb{Z}_2 \times \mathbb{Z}_2$ Sectors}. 
Let $\psi_{(0,1)}$ and $\psi_{(1,1)}$ be two paraparticle field operators. Their braided tensor product satisfies $\psi_{(0,1)} \otimes \psi_{(1,1)} = - \psi_{(1,1)} \otimes \psi_{(0,1)}$, because $(0,1) \cdot (1,1) = 1 \Rightarrow (-1)^{1} = -1$.
However, if both operators are in bosonic-like sectors, say $\psi_{(0,0)}$ and $\psi_{(1,1)}$, then $\psi_{(0,0)} \otimes \psi_{(1,1)} = \psi_{(1,1)} \otimes \psi_{(0,0)}$, because $(0,0) \cdot (1,1) = 0 \Rightarrow (-1)^0 = +1$.

\textbf{Physical Relevance}:
the braided coproduct and tensor product define how paraparticle fields combine in multi-mode Fock space. They encode exotic statistics beyond fermionic or bosonic behavior and result critical in simulating or realizing paraparticle logic in photonic systems or quantum circuits. What is relevant is that they ensure algebraic consistency through the Yang--Baxter identity, guaranteeing that logical gates and quantum entanglement remain well-defined under sector transitions.

These structures underlie the operator algebra governing the interaction of Majorana-like fields embedded into $\mathbb{Z}_2 \times \mathbb{Z}_2$-graded sectors, allowing the construction of multi-qudit logical operations in structured photonic systems.

\section{Graded Hopf Algebra Framework}

This appendix offers a general and didactic introduction to grading groups used in the mathematical description of paraparticles. It is intended for readers with minimal background in algebraic structures, especially those coming from physics, engineering, or computer science. We provide a concise introduction to the graded Hopf algebra formalism used throughout the main body of the paper to describe paraparticle symmetries and infinite-component Majorana fields.

\textbf{Overview of Hopf Algebras}. 
A \textit{Hopf algebra} is a structure that simultaneously carries algebraic and coalgebraic operations, enabling the consistent treatment of composite systems, such as multiparticle states in quantum field theory. Formally, a Hopf algebra $(\mathcal{H}, m, \eta, \Delta, \epsilon, S)$ is defined by a multiplication map $m: \mathcal{H} \otimes \mathcal{H} \to \mathcal{H}$ with a unit map $\eta: \mathbb{C} \to \mathcal{H}$ and a coproduct $\Delta: \mathcal{H} \to \mathcal{H} \otimes \mathcal{H}$.
A counit $\epsilon: \mathcal{H} \to \mathbb{C}$ and an antipode, a generalized inverse $S: \mathcal{H} \to \mathcal{H}$ complete the list.
These structures obey compatibility conditions such as co-associativity of the coproduct and the Hopf identity.

\textbf{Grading Groups}. 
A \emph{grading group} is a mathematical tool used to classify elements of an algebra into symmetry sectors. Each element is assigned a label (or \emph{grade}) from a group $G$, such that the multiplication of two elements respects the group structure.

For example, in a $\mathbb{Z}_2$-graded algebra, elements are either \emph{even} or \emph{odd}. When multiplying two elements, their grades add modulo $2$.
Grading groups are used in Quantum Theory to distinguish between types of quantum particles (e.g., bosons vs fermions), defining consistent commutation or anticommutation rules and organize algebraic structures in supersymmetry, parastatistics, and quantum field theory.
We report in Tab. \ref{diobubu} a concise recap of the use of the common grading groups and their uses.
\begin{center}
\begin{table}[ht]
\begin{tabular}{|c|c|c|l|}
\hline
\textbf{Grading} & \textbf{Dim.} & \textbf{Use Case} & \textbf{Example} \\ \hline
$\mathbb{Z}_2$ & 2 & Supersymmetry & \texttt{even vs. odd} \\ 
 &  & fermion/boson split &   \\ \hline
$\mathbb{Z}_2 \times \mathbb{Z}_2$ & 4 & Paraparticles & Green's algebra  \\ 
 & &  &paraSUSY \\ \hline
$\mathbb{Z}_p$ & $p$ & Anyons & Clock  \\ 
 &  & topological systems & rotor models \\ \hline
$\mathbb{Z}_2^n$ & $2^n$ & Multibit quantum& OAM/SAM \\
 &  &  logic, qudits & encoding \\ \hline
None & --- & Green's original & Early  \\ 
 &  & trilinear construction &  parastatistics \\ \hline
\end{tabular}
\caption{Table of common grading groups and their uses.}
\label{diobubu}
\end{table}
\end{center}

\subsection{The $\mathbb{Z}_2 \times \mathbb{Z}_2$--graded Lie algebra and Paraparticles}
In graded algebras, the multiplication rule between two elements depends on their grades. For example, in $\mathbb{Z}_2$:
\begin{equation}
AB = (-1)^{|A||B|} BA,
\end{equation}
where $|A|, |B| \in \{0,1\}$ are the grades of $A$ and $B$.

For $\mathbb{Z}_2 \times \mathbb{Z}_2$, each element has a grade $(a,b)$, and the graded sign becomes $AB = (-1)^{ac + bd} BA$, for $A$ of degree $(a,b)$ and $B$ of degree $(c,d)$.

Paraparticles such as parafermions and parabosons obey trilinear commutation relations. For parafermions, the defining identities are as in Eq. \ref{1} and \ref{2}.
These relations differ from the canonical anticommutator of fermions and allow multiple particles to occupy symmetrized states depending on the \emph{order} $p$ of the parafermion.

The choice of $\mathbb{Z}_2 \times \mathbb{Z}_2$ in this context is because 
$\mathbb{Z}_2 \times \mathbb{Z}_2$ is the \emph{smallest non-cyclic group} and provides the four symmetry sectors $\{(0,0), (1,0), (0,1), (1,1)\}$.
This allows us to encode both integer and half-integer spin sectors, modeling both parafermionic and parabosonic behaviors. The generalized symmetry rules in structured quantum systems are then preserved. This formalism strikes a balance between algebraic richness and computational manageability, which is why it is used as a starting point in many paraparticle models and quantum circuit representations.

Graded algebras and paraparticle statistics offer a powerful generalization of ordinary quantum field theory. Understanding grading groups like $\mathbb{Z}_2$, $\mathbb{Z}_p$, and especially $\mathbb{Z}_2 \times \mathbb{Z}_2$, gives you a toolbox for classifying and constructing exotic quantum systems with applications in structured light, quantum computing, and topological field theory.

\textbf{The $\mathbb{Z}_2 \times \mathbb{Z}_2 $ Graded Structure:}
The $\mathbb{Z}_2 \times \mathbb{Z}_2$ graded algebra provides a natural mathematical framework to generalize conventional statistics (bosonic and fermionic) to paraparticle statistics. Unlike standard $Z_2$-graded (super)algebras, which distinguish between even (bosonic) and odd (fermionic) components, the $\mathbb{Z}_2 \times \mathbb{Z}_2$ grading introduces a finer classification by assigning to each generator a degree $(a,b) \in \mathbb{Z}_2 \times \mathbb{Z}_2$, where each component is either $0$ or $1$.

\textbf{Basic Definitions}:
The $\mathbb{Z}_2 \times \mathbb{Z}_2$ group consists of four elements:
\begin{equation}
\mathbb{Z}_2 \times \mathbb{Z}_2 = { (0,0), (0,1), (1,0), (1,1) },
\end{equation}
with addition modulo 2 in each component. These elements serve as labels for the grading of algebraic components. The graded commutator for two homogeneous elements $A$ and $B$ of degrees $(a,b)$ and $(c,d)$, respectively, is defined as:
\begin{equation}
[A, B] = AB - (-1)^{ac + bd} BA.
\end{equation}

This generalization extends the symmetry structure of the algebra, permitting the construction of trilinear relations and more complex exchange symmetries, such as those observed in paraparticle statistics.

\textbf{Physical Interpretation}:
In the context of quantum field theory and paraparticles, the $\mathbb{Z}_2 \times \mathbb{Z}_2$ grading enables the classification of field components beyond bosons and fermions.
Assignment of specific algebraic rules for parafermions and parabosons is then addressed together with the encoding of exotic exchange statistics via generalized graded commutators.

For instance, in structured light applications, one can assign:
$(a,b)=(\ell$ mod $2, \frac{1-\sigma}{2})$, where $\ell$ is the OAM quantum number and $\sigma \in \{+1,-1\}$ denotes SAM (spin angular momentum).

As an example, the representation of the trilinear relations e.g., for $p=n$ parafermions, is obtained with the construction of Fock states where up to $n$ parafermions may occupy the same symmetrized state. 
The representation space is larger than for ordinary fermions.
The structure constants and operator algebra remain the same, but the projector representations and Hilbert space dimension reflect $p=n$.

Operators $f_ i$ act on a vacuum: $f_i |0\rangle = 0$. Creation operators build up a symmetric tensor representation of order 
$\leq p=n$, 
giving $f_{i_1}^\dag f_{i_2}^\dag \cdots f_{i_k}^\dag |0 \rangle$, $k\leq n$.
After that, $f^\dag_{i_1} f^\dag_{i_2}  f^\dag_{i_s} | 0 \rangle = 0$, $k\leq n$. 
Due to the trilinear constraints, these states are not fully antisymmetric like fermions, but satisfy parafermionic symmetry $\{f_i, \{f_j^\dagger, f_k\}\} = 2 \delta_{ij} f_k$, and $\{f_i, \{f_j, f_k\}\} = 0$.

\subsection{Short- and long-memory linear processes.}

Let \(L\) denote the back-shift operator \(L X_t = X_{t-1}\) and
\(\varepsilon_t\!\sim\!\mathcal N(0,\sigma^2)\) white noise.

\begin{description}
\item[ARMA\((p,q)\).]  
A \emph{stationary} autoregressive–moving-average process
\(X_t\) of orders \(p,q\in\mathbb N\) satisfies
\[
\bigl(1-\phi_1L-\dots-\phi_pL^p\bigr)X_t
     =\bigl(1+\theta_1L+\dots+\theta_qL^q\bigr)\,\varepsilon_t,
\]
with roots of the AR polynomial lying outside the unit circle.

\item[ARIMA\((p,d,q)\).]  
If the \(d\)-fold difference \(\nabla^{d}X_t=(1-L)^d X_t\)
is an ARMA\((p,q)\) process (with \(d\in\mathbb N\)),
then \(X_t\) is called an \emph{integrated} ARMA process,
written ARIMA\((p,d,q)\).
This is the classical Box–Jenkins model for short-memory,
trend-differenced series.

\item[ARFIMA\((p,d,q)\).]  
Allowing the differencing parameter to be
\emph{fractional} \(d\in(-\tfrac12,\tfrac12)\) and defining
\((1-L)^d\) via the binomial expansion
\(
(1-L)^d
   =\sum_{k=0}^{\infty}
     \binom{d}{k} (-L)^k,
\)
gives the autoregressive–\emph{fractionally}-integrated–moving-average
model
\[
(1-L)^d X_t
  =\frac{1+\theta_1L+\dots+\theta_qL^q}
         {1-\phi_1L-\dots-\phi_pL^p}\,
     \varepsilon_t,
\]
which captures \emph{long-range dependence} through the hyperbolic decay of its autocorrelation function.
\end{description}
Comprehensive treatments may be found in \cite{arma1}.

\subsection{Applications in Computational Classical and Quantum Systems}

The $\mathbb{Z}_2 \times \mathbb{Z}_2$ grading structure is essential in:
defining generalized creation/annihilation operators for paraparticles.
Constructing braided Hopf algebras with nontrivial exchange symmetries.
Developing quantum error correction and logic gates in structured photonic systems.
By incorporating this finer grading, one captures richer algebraic and physical behavior, paving the way for new forms of quantum computation and information processing.

The group $\mathbb{Z}_2$ is the simplest nontrivial group: it contains two elements, 0 and 1, with addition modulo 2. For a computer scientist, this is directly analogous to a single bit, where 0 and 1 represent binary states.

The group $\mathbb{Z}_2 \times \mathbb{Z}_2$ can be thought of as a 2-bit system where group operations are bitwise XOR on each component. This gives us four elements:
$\{ (0,0),\ (0,1),\ (1,0),\ (1,1) \}$,
which can be viewed as all 2-bit binary strings. Each string can label a different symmetry sector in a quantum system.
The group operation (addition modulo 2 component-wise) is written in Tab. \ref{cayley}:
\begin{table}[ht]
\begin{center}
\begin{tabular}{c|cccc}
$+$ & $e=(0,0)$ & $a=(1,0)$ & $b=(0,1)$ & $c=(1,1)$ \\
\hline 
$e=(0,0)$ & $(0,0)$ & $(1,0)$ & $(0,1)$ & $(1,1)$ \\
$a=(1,0)$ & $(1,0)$ & $(0,0)$ & $(1,1)$ & $(0,1)$ \\
$b=(0,1)$ & $(0,1)$ & $(1,1)$ & $(0,0)$ & $(1,0)$ \\
$c=(1,1)$ & $(1,1)$ & $(0,1)$ & $(1,0)$ & $(0,0)$ \\
\end{tabular}
\caption{Cayley Table of group operations}
\label{cayley}
\end{center}
\end{table}

This group is also known as the Klein four-group, and all its elements are their own inverses.
In quantum computing and logic design, this structure elegantly represents systems with dual binary labels--useful for encoding logical states or symmetry-preserving operations.

\subsection{Quantum Computing with a Ququart from SAM-OAM Coupled Modes}

We now define a logical ququart system based on the structured modes $\ell = 0, \pm 1$ and circular polarizations $\sigma = \pm 1$ (left and right) e.g., with SAM coupled with nonzero OAM modes. The logical basis states are identified as
\begin{align}
\ket{0} &= \ket{\ell = -1, \sigma = -1} \quad (\text{Left circular, OAM} = -1), \nonumber 
\\
\ket{1} &= \ket{\ell = 0, \sigma = +1} \quad (\text{Right circular, OAM} = 0), \nonumber
\\ 
\ket{2} &= \ket{\ell =  0, \sigma = -1} \quad (\text{Left circular, OAM} = 0), \nonumber
\\ 
\ket{3} &= \ket{\ell = +1, \sigma = +1} \quad (\text{Right circular, OAM} = +1). \nonumber
\end{align}

In this basis, single-ququart operations are realized through optical elements that manipulate the SAM-OAM hybridization. Specifically,
\\
\textbf{SAM - OAM operations} (polarization rotations and OAM flip) correspond to logical transformations within pairs $\{ \ket{0}, \ket{2} \}$ and $\{ \ket{1}, \ket{3} \}$.
\\
\textbf{OAM-only operations} (mode converters) allow transitions between states with the same polarization but different $\ell$, i.e., $\ket{0} \leftrightarrow \ket{2}$ and $\ket{1} \leftrightarrow \ket{3}$.
\\
\textbf{SAM-OAM coupling elements} (birefringent waveguide sections or stress-induced boundary modulations) induce mixing between all four ququart states, implementing generalized unitary transformations.

The general single-ququart unitary operation $U(4)$ can be decomposed into sequences of such SAM, OAM, and hybrid rotations. Gate operations are expressed through the hybrid mode basis transformation
\begin{equation}
\ket{\ell, \sigma} \rightarrow \sum_{(\ell', \sigma')} C^{(\ell', \sigma')}_{(\ell, \sigma)} \ket{\ell', \sigma'},
\end{equation}
where $C^{(\ell', \sigma')}_{(\ell, \sigma)}$ defines the specific physical implementation, corresponding to optical elements' design.
Logical field operators $\hat{\psi}_{(\ell, \sigma)}^\pm$ obey a $\mathbb{Z}_2 \times \mathbb{Z}_2$-graded algebra due to the mapping $(a,b) = (\ell \bmod 2, (1 - \sigma)/2)$,
which embeds the ququart into a paraparticle framework, enabling robust and symmetry-protected quantum operations.

In the dynamical framework, the evolution of the graded ququart field is governed by the Majorana-type equation as in Eq. \ref{35} with mass and spin assignments dependent on the grading sector, providing an avenue to encode relativistic quantum dynamics into photonic structured light qudits.

\subsection{Graded Trilinear Commutation Relations in \boldmath$\mathbb{Z}_2 \times \mathbb{Z}_2$ Superalgebra}

Here we report a list of graded commutation relations among the four $\mathbb{Z}_2 \times \mathbb{Z}_2$ sectors:
\\
\begin{enumerate}
\item
$[\hat{\psi}_{(0,0)}^+, \hat{\psi}_{(0,0)}^+] = 0$ (same bosonic sector)

\item 
$[\hat{\psi}_{(0,0)}^+, \hat{\psi}_{(0,1)}^+] =  \hat{\psi}_{(0,0)}^+ \hat{\psi}_{(0,1)}^+ - \hat{\psi}_{(0,1)}^+ \hat{\psi}_{(0,0)}^+$ (mixed even/odd SAM)

\item 
$[\hat{\psi}_{(0,0)}^+, \hat{\psi}_{(1,0)}^+] =  \hat{\psi}_{(0,0)}^+ \hat{\psi}_{(1,0)}^+ - \hat{\psi}_{(1,0)}^+ \hat{\psi}_{(0,0)}^+$ (mixed even/odd OAM)

\item 
$[\hat{\psi}_{(0,0)}^+, \hat{\psi}_{(1,1)}^+]=0$ (bosonic-like: scalar product zero)

\item 
$\{\hat{\psi}_{(0,1)}^+, \hat{\psi}_{(0,1)}^+\}=0$ (self-anticommutator, fermionic type)

\item 
$\{\hat{\psi}_{(0,1)}^+, \hat{\psi}_{(1,0)}^+\}=0$  (fermionic-like: scalar product)

\item 
$[\hat{\psi}_{(1,0)}^+, \hat{\psi}_{(1,1)}^+] = \hat{\psi}_{(1,0)}^+ \hat{\psi}_{(1,1)}^+ + \hat{\psi}_{(1,1)}^+ \hat{\psi}_{(1,0)}^+$ (scalar product even)

\item 
$[\hat{\psi}_{(0,1)}^-, [\hat{\psi}_{(1,1)}^+, \hat{\psi}_{(0,0)}^-]] = \cdots$ (trilinear graded behavior). 
\end{enumerate}
\bigskip

In the framework of paraparticles, bilinear commutation relations alone are insufficient to fully characterize particle statistics and symmetry. Instead, a more general algebraic structure involving trilinear commutation relations is required. These relations originate from Green's formulation of parastatistics and extend naturally to graded Lie superalgebras such as the $\mathbb{Z}_2 \times \mathbb{Z}_2$ case.

\subsection{General Trilinear Form}

A graded trilinear commutation relation for paraparticle operators takes the form
$[\hat{\psi}_{(a,b)}^-, [\hat{\psi}_{(a',b')}^+, \hat{\psi}_{(a'',b'')}^-]] =  \sum_{(c,d)} f_{(a,b),(a',b'),(a'',b'')}^{(c,d)} \hat{\psi}_{(c,d)}^-$, where $f_{(a,b),(a',b'),(a'',b'')}^{(c,d)}$ are graded structure constants determined by the algebra and R-matrix. This nested commutator first evaluates a graded bracket $[\hat{\psi}_{(a',b')}^+, \hat{\psi}_{(a'',b'')}^-]$, which typically yields a bilinear number or exchange operator, then it commutes with $\hat{\psi}_{(a,b)}^-$, ensuring algebraic closure and capturing higher-order correlations among paraparticle sectors.

\textbf{Example and Interpretation}. Consider the explicit graded trilinear example $[\hat{\psi}_{(0,1)}^-, [\hat{\psi}_{(1,1)}^+, \hat{\psi}_{(0,0)}^-]] = (-1)^{(0,1) \cdot (1,1)} [\hat{\psi}_{(0,1)}^-, \hat{n}_{(1,1)(0,0)}]$, where the bilinear operator is defined as $\hat{n}_{(1,1)(0,0)} = \hat{\psi}_{(1,1)}^+ \hat{\psi}_{(0,0)}^-$. Here we use the scalar product $(a,b)\cdot(a',b') = a a' + b b' \mod 2$. For example, in Eq.~(92), we have $(0,1)\cdot(1,1) = 0\cdot1 +1\cdot1 =1 \mod2$, yielding a phase factor $(-1)^1 = -1$. This convention is applied consistently throughout the matrix element calculations and numerical examples.
This illustrates how the grading affects both signs and operator structure. The scalar product $(0,1) \cdot (1,1) = 1$ introduces a minus sign due to parafermionic-like behavior. More specifically, the minus sign arises from the graded commutation relation $[X, Y] = $XY$ - (-1)^{(a,b) \cdot (a',b')} YX$, where the scalar product $(a,b) \cdot (a',b') = a a' + b b' \mod 2$ determines the phase factor. Specifically, $(0,1) \cdot (1,1) = 1$, yielding $(-1)^1 = -1$, which reflects the parafermionic-like exchange behavior between sectors. In our specific case, we compute the scalar product between the grading labels $(0,1) \cdot (1,1) = (0 \times 1) + (1 \times 1) = 0 + 1 = 1$.
This gives the phase factor $(-1)^{(0,1) \cdot (1,1)} = (-1)^1 = -1$, which introduces a minus sign into the nested commutator. Explicitly, we have
$[\hat{\psi}_{(0,1)}^-, \hat{\psi}_{(1,1)}^+ \hat{\psi}_{(0,0)}^-]  =  \hat{\psi}_{(0,1)}^- (\hat{\psi}_{(1,1)}^+ \hat{\psi}_{(0,0)}^-) - (-1)^1 (\hat{\psi}_{(1,1)}^+ \hat{\psi}_{(0,0)}^- \hat{\psi}_{(0,1)}^-)$, which simplifies to $\hat{\psi}_{(0,1)}^- \hat{\psi}_{(1,1)}^+ \hat{\psi}_{(0,0)}^- + \hat{\psi}_{(1,1)}^+ \hat{\psi}_{(0,0)}^- \hat{\psi}_{(0,1)}^-$, where 
we define the bilinear operator $\hat{n}_{(1,1)(0,0)}$ to inherit the grading of the first factor, i.e., $\mathrm{deg}(\hat{n}_{(1,1)(0,0)}) = (1,1)$.
This demonstrates how the grading controls both the signs and the operator ordering in the algebraic relations.

These trilinear relations enforce associativity and consistency across different graded sectors and capture statistical exclusion rules beyond Pauli or Bose-Einstein limits.
These relations also reflect nontrivial topological or algebraic braiding among paraparticles enabling the construction of non-Abelian logic gates and multi-particle symmetry operations in photonic quantum architectures.

Trilinear relations are thus fundamental to any consistent theory of paraparticles governed by $\mathbb{Z}_2 \times \mathbb{Z}_2$ symmetry and underpin both the algebraic and physical behavior of graded quantum fields.

For ease of reference, the key symbols and algebraic notations used throughout the manuscript are summarized in table Tab \ref{tabellona} for $\mathbb{Z}_2 \times \mathbb{Z}_2$ Graded Algebra and Majorana quanta.

\section{Jordan--Wigner Transformation, Quantum $XY$ model and Conformal Field Theory Limit}

The one-dimensional quantum $XY$ model is a prototypical spin chain described by the Hamiltonian
\begin{equation}
H_{\mathrm{XY}} = -J \sum_n \left( \sigma_n^x \sigma_{n+1}^x + \sigma_n^y \sigma_{n+1}^y \right),
\end{equation}
where $\sigma_n^x, \sigma_n^y$ are Pauli matrices at site $n$, and $J$ is the exchange coupling constant. The parameter $J$ determines the strength and nature of the interaction, with $J > 0$ indicating ferromagnetic coupling and $J < 0$ antiferromagnetic coupling.

To solve the $XY$ model, one can apply the Jordan--Wigner transformation, which maps spin operators to fermionic creation and annihilation operators:
\begin{equation}
\begin{aligned}
c_n &= \left( \prod_{m < n} \sigma_m^z \right) \frac{\sigma_n^x - i \sigma_n^y}{2}, \\
c_n^\dagger &= \left( \prod_{m < n} \sigma_m^z \right) \frac{\sigma_n^x + i \sigma_n^y}{2}.
\end{aligned}
\end{equation}
Here, $c_n$ and $c_n^\dagger$ are fermionic annihilation and creation operators, and the string of $\sigma^z$ operators ensures the correct anticommutation relations.
The Jordan--Wigner transformation defined by equation (112) explicitly accounts for operator ordering through the string operator $c_j = \sigma_j^- \exp\left(i\pi\sum_{k<j}\sigma_k^+\sigma_k^-\right)$,
ensuring correct anti-commutation relations among the emergent fermionic modes. This definition guarantees algebraic consistency and correct fermionic statistics throughout our formalism.
After this transformation, the spin Hamiltonian becomes a quadratic fermionic Hamiltonian
\begin{equation}
H_{\mathrm{XY}} \mapsto \sum_{m, n} t_{mn} c_m^\dagger c_n,
\end{equation}
which can be diagonalized using standard methods.
The quantity $t_{mn}$ represents the effective hopping amplitude or kinetic coupling between fermionic modes $c_m^\dagger$ creation operator at site $m$) and $c_n$  (annihilation operator at site $n$) after applying the Jordan--Wigner transformation, the matrix element describing hopping from site $n$ to site $m$, $t_{mn} = J$ if $m=n+1$ or $m=n-1$ and $t_{mn} = 0$ otherwise.

\subsection{Continuum Limit and CFT Description}

At criticality, the system becomes gapless and scale-invariant, allowing a continuum limit where the lattice spacing $a \to 0$. The long-wavelength behavior is then described by an effective relativistic field theory with free massless Dirac fermions, or equivalently, a compactified bosonic field via bosonization or a Majorana-equivalent formulation. This effective field theory falls into the class of $(1+1)$-dimensional conformal field theories (CFTs), with central charge $c=1$ \cite{tambu3,tambu4}. The central charge $c$ is a key parameter that appears in the Virasoro algebra, which governs the symmetry structure of the theory. The Virasoro algebra extends the conformal symmetry algebra by including a central term $[L_m,L_n]=(m-n)L_{m+n}+\frac{c}{12}m(m^2-1)\delta_{m+n,0}$, where $L_m$ are the Virasoro generators and 
$c$ is the central charge, which is also called the conformal anomaly or conformal central extension and counts the effective degrees of freedom or ``weight'' of the quantum fields in the theory. This constant controls the scaling of entanglement entropy $S(\ell_s)$ of an interval of length $\ell_s$, the finite-size correction to the ground state energy and the operator product expansion (OPE) coefficients with the response to boundary conditions.

\subsection{Entanglement and Modular Hamiltonian}

The ground state entanglement entropy of a subsystem of length $\ell_s$ scales logarithmically 
\begin{equation}
S(\ell_s) \sim \frac{c}{3} \ln \ell_s + \mathrm{const},
\end{equation}
where $c$ is the central charge.

Summarizing, the Jordan--Wigner transformation converts spins to fermions. At critical points, the fermionic system flows to a CFT. The CFT framework controls universal quantities such as entanglement entropy and the modular Hamiltonian connects to other fundamental mathematical structures.
A summary of connections is briefly illustrated by the mapping flow as a scheme in Fig. \ref{f2},
\begin{figure}[h]
\includegraphics[width=0.45\textwidth]{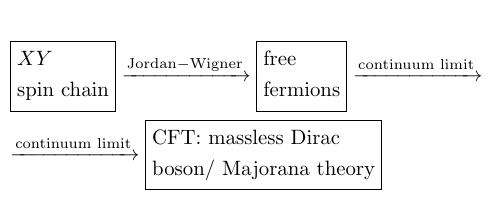}
\caption{Jordan--Wigner transformation conversion scheme}
\label{f2}
\end{figure}

This interplay between lattice models, fermionic mappings, and field-theoretic limits is fundamental to understanding quantum criticality and underpins many of the theoretical predictions here explored.
The logarithmic scaling behavior of the entanglement entropy there described involves intervals explicitly defined as segments of length $\ell_s$ within the structured photonic system. The boundary conditions are assumed periodic (or alternatively open), clearly specified to maintain consistency with modular Hamiltonian formalism. Hence, the modular Hamiltonian accurately encodes entanglement characteristics pertinent to the structured mode configurations.

\subsection{Verification of the Yang--Baxter Equation for Braided Coproducts}

We provide here a compact verification that the braided coproduct structure introduced for the $\mathbb{Z}_2 \times \mathbb{Z}_2$-graded Majorana-paraparticle framework satisfies a generalized Yang--Baxter equation. 

Recall that the braided coproduct for an operator $\psi_{(a,b)}$ is defined as $\Delta(\psi_{(a,b)}) = \psi_{(a,b)} \otimes \mathbb{I} + \sum_{(c,d)} R^{(c,d)}_{(a,b)} (\mathbb{I} \otimes \psi_{(c,d)})$,
where the braiding matrix elements are given by $R^{(c,d)}_{(a,b)} = (-1)^{(a,b)\cdot(c,d)} \delta^{(c,d)}_{(a,b)} + \theta_{(a,b),(c,d)}$.
Here, the scalar product $(a,b)\cdot(c,d) = ac + bd \mod 2$, and $\theta_{(a,b),(c,d)}$ encodes a deformation parameter typically arising from a statistical phase.

The generalized Yang--Baxter equation requires that $R_{12} R_{13} R_{23} = R_{23} R_{13} R_{12}$,
where $R_{ij}$ acts nontrivially on the $i$-th and $j$-th components of a threefold tensor product.
Substituting the form of $R$, the left-hand side expands as $R_{12}R_{13}R_{23}(\psi_{(a,b)}\otimes\psi_{(c,d)}\otimes\psi_{(e,f)}) = (-1)^{(a,b)\cdot(c,d)}(-1)^{(a,b)\cdot(e,f)}(-1)^{(c,d)\cdot(e,f)} \cdot  \cdot ~ \psi_{(e,f)}\otimes\psi_{(c,d)}\otimes\psi_{(a,b)}$  
+ terms with $\theta$, while the right-hand side gives $R_{23}R_{13}R_{12}(\psi_{(a,b)}\otimes\psi_{(c,d)}\otimes\psi_{(e,f)}) = (-1)^{(c,d)\cdot(e,f)}(-1)^{(a,b)\cdot(e,f)}(-1)^{(a,b)\cdot(c,d)} \cdot~\psi_{(e,f)}\otimes\psi_{(c,d)}\otimes\psi_{(a,b)}$ 
+ terms with $\theta$.
The verification of the Yang--Baxter equation within our graded coproduct structure explicitly involves confirming $(\mathcal{R}\otimes I)(I\otimes\mathcal{R})(\mathcal{R}\otimes I) = (I\otimes\mathcal{R})(\mathcal{R}\otimes I)(I\otimes\mathcal{R})$. Expanding and simplifying both sides using the explicit form of $\mathcal{R}$ provided earlier ensures algebraic consistency. Each step confirms the nontrivial braiding properties associated with paraparticle sectors.

Since the scalar product is bilinear and symmetric modulo 2, the total phases on both sides exactly match:
\begin{eqnarray}
&&(a,b)\cdot(c,d) + (a,b)\cdot(e,f) + (c,d)\cdot(e,f) = \nonumber
\\
&&= (c,d)\cdot(e,f) + (a,b)\cdot(e,f) + \nonumber
\\
&& +(a,b)\cdot(c,d)\mod 2.  
\label{eq:thetaAdd} 
\end{eqnarray}
The relation \ref{eq:thetaAdd} follows from the 2-cocycle condition $\theta(g_1,g_2)+\theta(g_1g_2,g_3)=\theta(g_2,g_3)+\theta(g_1,g_2g_3)$ for the abelian cohomology $H^2(Z_2\times Z_2,U(1))$. A brief derivation is included in Appendix~A.
Thus, the pure graded terms satisfy the Yang--Baxter equation identically. The terms involving the deformation parameters $\theta_{(a,b),(c,d)}$ respect associativity provided that the deformation satisfies an additive property consistent with the braiding relations, namely $\theta_{(a,b),(c,d)} + \theta_{(a,b),(e,f)} + \theta_{(c,d),(e,f)} =  \theta_{(c,d),(e,f)} + \theta_{(a,b),(e,f)} + \theta_{(a,b),(c,d)}$, which holds if $\theta$ depends only on symmetric functions of the graded labels, such as spin parity or statistical phase factors $q$ with $|q|=1$.

Hence, the braided coproduct structure defined in the $\mathbb{Z}_2 \times \mathbb{Z}_2$-graded Majorana paraparticle algebra consistently satisfies the generalized Yang--Baxter equation, ensuring coherent multiparticle braiding and associativity properties.

\paragraph{Additivity of the braiding phase:}
The Yang--Baxter proof hinges on the fact that the $U(1)$ phase $\theta\colon G\times G\!\to\! \mathbb{R}/2\pi\mathbb{Z}$ (with $G\!=\!\mathbb{Z}_{2}\!\times\!\mathbb{Z}_{2}$) satisfies the $2$-cocycle condition
\begin{equation}
\label{eq:thetaAdd}
\theta(g_{1},g_{2})+\theta(g_{1}g_{2},g_{3}) =m\theta(g_{2},g_{3})+\theta(g_{1},g_{2}g_{3}),
\end{equation}
and $g_{1},g_{2},g_{3}\in G$. Because $G$ is abelian, every $2$-cocycle is cohomologous to a bicharacter, hence $\theta$ can be chosen additive in each slot, $\theta(g_{1}g_{2},g_{3})=\theta(g_{1},g_{3})+\theta(g_{2},g_{3})$.
Equation~\ref{eq:thetaAdd} is therefore the algebraic origin of the hexagon/associativity constraint used in the next step of the proof.

\subsection{Explicit Matrix Decomposition for $N=4$}

We provide here the explicit matrix decomposition of the $\mathbb{Z}_4$ parafermionic generators into $\mathbb{Z}_2 \times \mathbb{Z}_2$ graded components, complementing the algebraic discussion in the main text.
The $\mathbb{Z}_4$ parafermionic basis states are e.g., $ |0\rangle = |L,0\rangle$, $|1\rangle = |R,0\rangle$, 
$|2\rangle = |L,1\rangle$ and $\quad |3\rangle = |R,1\rangle$, where $L/R$ are circular polarizations and $\ell = 0,1$ the OAM quantum numbers.

The clock operator $\sigma$ and shift operator $\tau$ for $\mathbb{Z}_4$ are given by
\begin{equation}
\sigma = 
\begin{pmatrix}
1 & 0 & 0 & 0 \\
0 & i & 0 & 0 \\
0 & 0 & -1 & 0 \\
0 & 0 & 0 & -i
\end{pmatrix}, \quad
\tau = 
\begin{pmatrix}
0 & 1 & 0 & 0 \\
0 & 0 & 1 & 0 \\
0 & 0 & 0 & 1 \\
1 & 0 & 0 & 0
\end{pmatrix}.
\end{equation}

We decompose now the four-dimensional Hilbert space as a tensor product of two qubits, $\mathcal{H}_4 \cong \mathcal{H}_2^A \otimes \mathcal{H}_2^B$, with $|0\rangle \leftrightarrow |0\rangle_A |0\rangle_B$, $|1\rangle \leftrightarrow |0\rangle_A |1\rangle_B$ and $|2\rangle \leftrightarrow |1\rangle_A |0\rangle_B$ with $|3\rangle \leftrightarrow |1\rangle_A |1\rangle_B$.

Introduce then the independent $\mathbb{Z}_2$ grading operators $P = \sigma_z \otimes \mathbb{I}_2$ and $Q = \mathbb{I}_2 \otimes \sigma_z$, which explicitly are with the Pauli matrices $\sigma_x$, $\sigma_z$ and $\mathbb{I}_2$, then 

\begin{equation}
P = \begin{pmatrix}
1 & 0 & 0 & 0 \\
0 & 1 & 0 & 0 \\
0 & 0 & -1 & 0 \\
0 & 0 & 0 & -1
\end{pmatrix}, \quad
Q = \begin{pmatrix}
1 & 0 & 0 & 0 \\
0 & -1 & 0 & 0 \\
0 & 0 & 1 & 0 \\
0 & 0 & 0 & -1
\end{pmatrix}
\end{equation}
and obtain the decomposition $\sigma = P Q$, for which $\sigma |0\rangle = |0\rangle$, $\sigma |1\rangle = i |1\rangle$, $\sigma |2\rangle = -1 |2\rangle$, and $\sigma |3\rangle = -i |3\rangle$.

For the cyclic shift operator $\tau$, we note that in the qubit basis it can be written as:
$\tau = \text{CNOT}_{A \to B} \cdot H_A \cdot S_B$, where $H_A$ is the Hadamard gate on qubit $A$, $S_B$ is a phase gate on $B$, and $\text{CNOT}_{A \to B}$ is a controlled-NOT gate from $A$ to $B$. This decomposition shows how the $\mathbb{Z}_4$ cyclic symmetry arises from entangling operations between the two $\mathbb{Z}_2$ sectors.

This explicit matrix form makes transparent how the $\mathbb{Z}_4$ parafermionic operators can be factorized into a $\mathbb{Z}_2 \times \mathbb{Z}_2$ graded structure, providing a clear blueprint for physical implementations (see Fig. \ref{f3}).

\begin{figure}[h]
\includegraphics[width=0.33\textwidth]{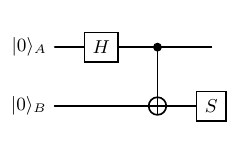}
\caption{Schematic decomposition of the $\mathbb{Z}_4$ ququart system into two qubits $A$ and $B$. The clock operator $\sigma$ is decomposed into $P = \sigma_z \otimes I$ and $Q = I \otimes \sigma_z$. The cyclic shift $\tau$ is implemented by a Hadamard $H$ on qubit $A$, a phase gate $S$ on qubit $B$, and a controlled-NOT ($\mathrm{CNOT}$) from $A$ to $B$.}
\label{f3}
\end{figure}

\section{Appendix B: Applications and Examples}

\subsection{Explicit Model Hamiltonian for OAM--SAM Coupled Ququarts}

To complement the algebraic framework developed in the main text, we provide here a schematic Hamiltonian that explicitly captures the essential features of the OAM--SAM coupled ququart system and illustrates how the $\mathbb{Z}_4$ parafermionic symmetry emerges.

We consider a four-mode system defined by the basis states $|0\rangle = |L,0\rangle$, $|1\rangle = |R,0\rangle$, $|2\rangle = |L,1\rangle$, $|3\rangle = |R,1\rangle$,
where $L/R$ denote left- and right-circular polarizations (spin angular momentum, SAM) and $\ell=0,1$ are the orbital angular momentum (OAM) quantum numbers.
The general tight-binding Hamiltonian reads
\begin{equation}
\hat{H} = \sum_{j=0}^3 \omega_j a_j^\dagger a_j + \sum_{\langle j,k \rangle} J_{jk} a_j^\dagger a_k + \sum_{\langle j,k \rangle} \lambda_{jk} a_j^\dagger \sigma_z a_k + \text{h.c.},
\end{equation}
where $a_j^\dagger$, $a_j$ are bosonic or fermionic creation and annihilation operators, $\omega_j$ are mode frequencies, $J_{jk}$ are nearest-neighbor coupling amplitudes between modes and $\lambda_{jk}$ encode spin--orbit (OAM--SAM) couplings, with $\sigma_z$ acting on the polarization subspace.

For the parafermionic sector, we use the generalized clock and shift operators $\sigma_c$ and $\tau$ for which $\sigma |k\rangle = e^{i\frac{\pi}{2}k} |k\rangle$, 
$\tau |k\rangle = |k+1 \mod 4 \rangle$, satisfying $\sigma_c^4 = \tau^4 = \mathbb{I}$, $\tau \sigma_c = i \sigma_c \tau$.
These operators can be realized in the Hamiltonian through cyclic hopping terms:
\begin{equation}
\hat{H}_\text{para} = g \sum_{k=0}^3 a_k^\dagger a_{k+1} + \text{h.c.},
\end{equation}
with periodic boundary conditions $a_{k+4} \equiv a_k$. The clock operator $\sigma_c$ can be implemented via mode-dependent phase shifts
\begin{equation}
\hat{H}_\text{phase} = \sum_{k=0}^3 \phi_k a_k^\dagger a_k, \quad \phi_k = \frac{\pi}{2}k.
\end{equation}
This explicit form connects the abstract algebraic symmetries to physically controllable parameters in photonic or condensed matter platforms, providing a basis for future experimental implementations.
The graded phase factors follow from the  multiplication rules of $\mathbb{Z}_2 \times \mathbb{Z}_2$-graded algebras, as developed in works by Palev, Tolstoy, and Scheunert \cite{palev1979, tolstoy, scheunert,Majid1995}. 
 The deformation factor $q^{|s - s'|}$ arises from the introduction of a quantum group deformation parameter, a technique pioneered by Majid \cite{Majid1995} and others in the study of braided Hopf algebras. Recent contributions by Toppan \cite{toppan} and collaborators have applied these structures to the embedding of paraparticles in a Majorana tower, providing the explicit algebraic form used in our numerical simulations.
The coupling coefficients $\theta_{ss'}$ are determined by the algebraic framework imposed on the Majorana tower, rather than fixed by the original Majorana mass--spin relation. Specifically, the graded phase factors $(-1)^{(a,b)\cdot(a',b')}$ arise from the multiplication rules of the $\mathbb{Z}_2 \times \mathbb{Z}_2$-graded algebra, while the deformation factor $q^{|s - s'|}$ is introduced through the quantum group or braided Hopf algebra formalism. Together, these components define the algebraic coupling pattern between spin sectors. The explicit form of $\theta_{ss'}$ thus depends on the theoretical choices made in the construction of the model and reflects the desired symmetry, braiding, and selection rules to be implemented in the physical or simulated system.
The graded couplings express not only the static algebraic structure but also the dynamic control possibilities offered by the truncated Majorana tower framework. Leveraging these selection rules in photonic implementations enables targeted manipulation of spin sectors, allowing for the design of quantum gates and protocols that intrinsically respect paraparticle symmetries. Beyond fundamental interest, this framework points toward practical applications in photonic quantum computing, including sector-isolated logical operations, error suppression through symmetry protection, and the emulation of exotic quantum phases. Moreover, the photonic Majorana ququart architecture provides a promising experimental testbed to explore connections between graded algebras, quantum entanglement, and topologically inspired quantum information processing.

\subsection{Graded Sector-Based CNOT and Toffoli Gates}
Here we present the explicit truth tables for quantum logic gates operating over $\mathbb{Z}_2 \times \mathbb{Z}_2$-graded sectors used in the Majorana-paraparticle framework. Control and target qubits are assigned sector labels $(a,b)$, with graded parity governing their transformations.

The CNOT (Controlled-NOT) gate flips the target qubit if and only if the control qubit is in the logical 1 sector. In the graded sector formalism, this corresponds to the control being in sector $(0,1)$ or $(1,1)$, depending on encoding. The graded CNOT truth table is
\begin{center}
\begin{tabular}{|c|c|c|c|}
\hline
Control & Target Qubit & Target Qubit & Operation \\
Qubit & (Input) & (Output) & Operation \\
\hline
$(0,0)$ & $t$ & $t$ & No Flip \\
$(0,1)$ & $0$ & $1$ & Flip \\
$(0,1)$ & $1$ & $0$ & Flip \\
$(1,0)$ & $t$ & $t$ & No Flip \\
$(1,1)$ & $0$ & $1$ & Flip \\
$(1,1)$ & $1$ & $0$ & Flip \\
\hline
\end{tabular}
\end{center}
where $t \in \{0,1\}$ denotes the target qubit logical value.

Example of Graded Sector-Based Toffoli (CCNOT) Gate:
The Toffoli gate (Controlled-Controlled-NOT, or CCNOT) flips the target qubit if and only if both control qubits are simultaneously in their logical 1 sectors,
\\
- First control in sector $(0,1)$, e.g., SAM-right circular polarization,
\\
- Second control in sector $(1,1)$, e.g., OAM-odd mode with certain parity.
\\The graded Toffoli truth table then becomes
\begin{center}
\begin{tabular}{|c|c|c|c|c|}
\hline
Control & Control & Target & Target & Operation \\
Qubit 1 & Qubit 2 & (Input) & (Output) & Operation \\
\hline
$(0,0)$ & $(1,0)$ & $A_p$ & $A_p$ & No Flip \\
$(0,0)$ & $(1,1)$ & $A_p$ & $A_p$ & No Flip \\
$(0,1)$ & $(1,0)$ & $A_p$ & $A_p$ & No Flip \\
$(0,1)$ & $(1,1)$ & $A_p$ & $B_p$ & Flip \\
$(0,1)$ & $(1,1)$ & $B_p$ & $A_p$ & Flip \\
$(0,0)$ & $(1,0)$ & $B_p$ & $B_p$ & No Flip \\
$(0,0)$ & $(1,1)$ & $B_p$ & $B_p$ & No Flip \\
$(0,1)$ & $(1,0)$ & $B_p$ & $B_p$ & No Flip \\
\hline
\end{tabular}
\end{center}
where $A_p$ denotes the logical ``0'' state (e.g., path $A_p$), $B_p$ denotes the logical ``1'' state (say, path $B_p$). Flip means toggling between $A_p$ and $B_p$ paths (logical NOT on the target). In both gates, the operation respects the underlying graded parity structure, $X_{(a,b),(a',b')} = (-1)^{(a,b)\cdot(a',b')}$, ensuring that only specified graded sector combinations trigger target flips, preserving graded quantum symmetry.

A quantitative loss-and-fidelity budget for either free-space or integrated implementations can be obtained by inserting platform-specific transmission coefficients in the equations reported in this work. We defer such calculations to
future experimental work.

\onecolumngrid
\center
\begin{table}[h!]
\centering
\begin{tabular}{|c|p{7cm}|p{7cm}|}
\hline
\textbf{Symbol} & \textbf{Meaning} & \textbf{Notes} \\
\hline
$ (a,b) \in \mathbb{Z}_2 \times \mathbb{Z}_2 $ & Graded index for Lie algebra sectors & indices like $a,b\in\{0,1\}$ label bosonic/fermionic-like sectors\\
\hline
$ \hat{\psi}^{\pm}_{(a,b)} $ & Creation (+) / annihilation (-) operator in graded sector $ (a,b) $ & Operators satisfy trilinear commutation/anticommutation relations \\
\hline
$ \Psi_{(a,b)} $ & Component of the infinite-dimensional Majorana wavefunction in sector $ (a,b) $ & Part of graded decomposition of $ \Psi $ \\
\hline
$ \Gamma^\mu_{(a,b)} $ & Graded gamma matrices in sector $ (a,b) $ & Generalization of standard gamma matrices \\
\hline
$ M_{(a,b)} $ & Effective mass term in sector $ (a,b) $ & $ M_{(a,b)} = M(s_{(a,b)} + 1/2) $ \\
\hline
$ (a,b) \cdot (a',b') $ & Scalar product between graded indices & Defined as $ aa' + bb' \mod 2 $ \\
\hline
$ [X, Y] $ & Graded commutator & $ $XY$ - (-1)^{(a,b)\cdot(a',b')} YX $ \\
\hline
$ \Delta(\hat{\psi}_{(a,b)}) $ & Braided coproduct in graded Hopf algebra & $ \hat{\psi}_{(a,b)} \otimes I + \sum R^{(a,b)}_{(c,d)} I \otimes \hat{\psi}_{(c,d)} $ \\
\hline
$ R^{(a,b)}_{(c,d)} $ & Braiding matrix coefficients & Encode exchange statistics; satisfy Yang--Baxter equation \\
\hline
$ \theta_{ss'} $ & Deformation parameter for exchange statistics & $ \theta_{ss'} = \epsilon_{ss'}(q^{s+s'} - 1), \; |q| = 1 $ \\
\hline
$ P_{(a,b)} $ & Projector onto graded sector $ (a,b) $ & $ P_{(a,b)}\Psi = \Psi_{(a,b)} $, $ \sum_{(a,b)} P_{(a,b)} = I $ \\
\hline
$ X_{(a,b),(a',b')} $ & Exchange matrix for graded sectors & $ (-1)^{(a,b)\cdot(a',b')} $ \\
\hline
$ n_{(a,b)} $ & Number operator in graded sector & $ n_{(a,b)} = \hat{\psi}^+_{(a,b)} \hat{\psi}^-_{(a,b)} $ \\
\hline
$ E_{(a,b),(a',b')} $ & Exchange operator between graded sectors & $ E_{(a,b),(a',b')} = \hat{\psi}^+_{(a,b)} \hat{\psi}^-_{(a',b')} $ \\
\hline
$ H $ & Hamiltonian operator & $ H = \sum_{(a,b)} \epsilon_{(a,b)} n_{(a,b)} $ \\
\hline
$ s_{(a,b)} $ & Spin associated with sector $ (a,b) $ & $ s=0 \to (0,0), \; s=1/2 \to (0,1), \; s=1 \to (1,1), \; s=3/2 \to (1,0) $ \\
\hline
\end{tabular}
\caption{Summary of symbols and notation used in this work with the $\mathbb{Z}_2 \times \mathbb{Z}_2$-graded algebraic framework.}
\label{tabellona}
\end{table}
\twocolumngrid
\end{document}